\documentclass[reqno, 11pt]{amsart}

\usepackage{amsfonts,amssymb,amsbsy,amsmath,amsthm,enumerate}
\usepackage{txfonts}
\usepackage{eucal}
\usepackage{dsfont}
\usepackage{mathrsfs}

\usepackage{diagrams}

\usepackage[dvips]{graphicx}
\usepackage{color}

\newtheorem{theorem}{Theorem}[section]
\newtheorem{proposition}[theorem]{Proposition}
\newtheorem{lemma}[theorem]{Lemma}
\newtheorem{corollary}[theorem]{Corollary}
\newtheorem{assumption}[theorem]{Assumption}
\newtheorem{definition}[theorem]{Definition}
\newtheorem{remark}[theorem]{Remark}

\numberwithin{equation}{section}

\numberwithin{figure}{section}
\numberwithin{table}{section}

\newcommand\beq{\begin{equation}}
\newcommand{\bea}{\begin{eqnarray}}
\newcommand{\eea}{\end{eqnarray}}
\newcommand{\beas}{\begin{eqnarray*}}
\newcommand{\eeas}{\end{eqnarray*}}
\newcommand{\beql}{\begin{equation} \label}
\newcommand{\eeq}{\end{equation}}

\newcommand{\R}{\mathbb R}
\newcommand{\N}{\mathbb N}
\newcommand{\C}{\mathbb C}                           

\newcommand{\Z}{\mathbb Z}
\newcommand{\T}{\mathbb T}

\newcommand{\s}[1]{\CMcal{#1}}
\newcommand{\f}[1]{\mathcal{#1}}                  
\newcommand{\bb}[1]{\mathscr{#1}}
\newcommand{\rr}[1]{\mathfrak{#1}}
\newcommand{\n}[1]{\mathds {#1}}

\newcommand{\expo}[1]{{\rm e}^{#1}}                 
\newcommand{\dd}{\,{\rm d}}
\newcommand{\ii}{\,{\rm i}\,}
\newcommand{\ncint}{\mathrel{{\ooalign{$\int$\cr\kern+.07em\raise.15ex\hbox{$\pmb{\scriptstyle-}$}\cr}}}}           \newcommand{\ncpartial}{\mathrel{{\ooalign{$\partial$\cr\kern+.29em\raise.79ex\hbox{$\pmb{\scriptstyle-}$}\cr}}}}

\newcommand{\virg}[1]{\lq\lq#1\rq\rq}                
\newcommand{\ie}{{\sl i.\,e.\,}}
\newcommand{\eg}{{\sl e.\,g.\,}}
\newcommand{\cf}{{\sl cf.\,}}

\begin{document}

\title[Classification of \emph{\virg{Quaternionic}} Bloch-bundles]{Classification of \emph{\virg{Quaternionic}} Bloch-bundles:\\
Topological Quantum Systems of type {\bf AII}}


\author[G. De~Nittis]{Giuseppe De Nittis}
\address[De~Nittis]{Department Mathematik, Universit\"{a}t Erlangen-N\"{u}rnberg,
Germany}
\email{denittis@math.fau.de}

\author[K. Gomi]{Kiyonori Gomi}
\address[Gomi]{Department of Mathematical Sciences, Shinshu University,  Nagano, Japan}
\email{kgomi@math.shinshu-u.ac.jp}

\thanks{{\bf MSC2010}
Primary: 57R22; Secondary:  55N25, 53C80, 19L64}

\thanks{{\bf Keywords.}
Topological insulators, Bloch-bundle, \virg{Quaternionic} vector bundle, FKMM-invariant, Fu-Kane-Mele index.}


\begin{abstract}
\vspace{-4mm}
We provide a classification of type {\bf AII} topological quantum systems in dimension $d=1,2,3,4$.
Our analysis is based on the construction of a topological invariant, the \emph{FKMM-invariant}, which completely classifies 
 \virg{Quaternionic} vector bundles (\emph{a.k.a.} \virg{symplectic} vector bundles) in dimension $d\leqslant 3$. This invariant takes value in a proper \emph{equivariant} cohomology theory and, in the case of examples of physical interest, it reproduces
 the familiar Fu-Kane-Mele index. In the case $d=4$ the classification requires a combined use of the FKMM-invariant and the second Chern class. Among the other things, we prove that  the FKMM-invariant is a \emph{bona fide} characteristic class for the category of  \virg{Quaternionic} vector bundles in the sense that it can be realized as the pullback of a universal topological invariant.
\end{abstract}


\maketitle

\vspace{-5mm}
\tableofcontents

\section{Introduction}\label{sect:intro}

In this paper  we continue the classification of \emph{topological quantum systems}  started in \cite{denittis-gomi-14} and, in particular, we focus on the so-called class {\bf AII} according to the {Altland-Zirnbauer-Cartan 
classification} (AZC) of topological insulators \cite{altland-zirnbauer-97,schnyder-ryu-furusaki-ludwig-08}.
Both classes {\bf AI} and {\bf AII} describe (quantum) systems that are invariant under a \emph{time-reversal} (TR) symmetry but it is the behavior of the spin that distinguishes  between the two classes. Systems in class 
{\bf AI} describe \emph{spinless} (as well as \emph{integer} spin) particles and, from a topological point of view,
this class is poor of interesting effects, as it emerges from the accurate analysis carried out in \cite{denittis-gomi-14}.
On the other side, systems in class {\bf AII} show interesting physical phenomena of topological origin like the so-called \emph{Quantum Spin Hall Effect} (QSHE) \cite{kane-mele-05,moore-balents-07,roy-09}. Phenomena of this type were first described by L. Fu, C. L. Kane and E. J. Mele in a series of seminal works \cite{kane-mele-05',kane-mele-05,fu-kane-06,fu-kane-mele-07}
and nowadays  they are source of
 great interest
among the physics community (see \eg the recent review \cite{maciejko-hughes-zhang-11} and references therein).
The principal result obtained by Fu, Kane  and Mele (at least from a mathematical point of view)
was the identification of the SQHE with a topological invariant  today known as \emph{Fu-Kane-Mele index}. This index characterizes the topology of the Bloch energy bands for periodic systems of free fermions in the presence of a TR-symmetry in the same way as the  Chern numbers describe the topology of the Bloch  bands when the TR-symmetry is broken. However, a correct mathematical  understanding of the  topological nature of the Fu-Kane-Mele index seems to be still missing in the literature and the most recent \emph{mathematical} works \cite{avila-schulz-baldes-villegas-13,porta-graf-13} only treat the case of two-dimensional lattice systems. Our main goal is to fill this gap.

\medskip

In the absence of disorder,  the Bloch-Floquet analysis relates topological insulators of class {\bf AII} with a special category of complex vector bundles called \emph{\virg{Quaternionic}}. These vector bundles, introduced for the first time by  J. L. Dupont  in \cite{dupont-69} with the (ambiguous) name of \emph{symplectic} vector bundles (see also \cite{seymour-73,dos-santos-lima-filho-04,lawson-lima-filho-michelsohn-05}), are
complex vector bundles defined over an \emph{involutive} base space $(X,\tau)$ and endowed with an \emph{anti}-involutive automorphism of the total space which covers the involution $\tau$ and restricts to an \emph{anti}-linear map between conjugate  fibers. 
In Section \ref{sect:Q-vb} we provide a precise description of this category.
Let us just point out that throughout the paper we will often use the short expression \emph{$\rr{Q}$-bundle}
instead {\virg{Quaternionic}} vector bundle. The principal results achieved in this paper can be summarized as follows:
\begin{itemize}
\item We provide a classification of  topological quantum systems of class {\bf AII} by inspecting the \emph{homotopy} classification of the underlying \virg{Quaternionic} category of vector bundles.
 In this way we obtain a classification which is, in spirit, finer than the usual $K$-theoretical classification \cite{kitaev-09} (see also Appendix \ref{sect:KQ-theory}) since it covers also the \emph{unstable} case. Moreover, our classification is a priori valid also for quite general base spaces and not only for spheres or tori.
 \vspace{1.3mm}
 
\item We introduce a topological invariant that discriminates between non-isomorphic $\rr{Q}$-bundles and which is sufficiently fine to provide a complete description of the category of \virg{Quaternionic}  vector bundles
in low dimensions, \ie when the base space is a (closed) manifold $X$ of dimension $d\leqslant 3$.
The construction of this invariant is based on an original idea by M. Furuta, Y. Kametani, H. Matsue, and N. Minami described in an unpublished work \cite{furuta-kametani-matsue-minami-00} dated 2000 (five years earlier than the first paper by Kane and Mele!) and for this reason we decide to call it the \emph{FKMM-invariant}.
We provide a precise description of the FKMM-invariant in Section \ref{sec:FKMM-inv}
and we prove that in special cases (including all cases of interest for the description of topological insulators)
this invariant  reproduces the Fu-Kane-Mele index (\cf in particular Remark \ref{rk:FKMM_w_det}).
\vspace{1.3mm}
 
\item We prove that the FKMM-invariant is a genuine \emph{characteristic class} for the category of   \virg{Quaternionic}  vector bundles in the sense that there exists a \emph{universal} version of it which 
provides by pullback the FKMM-invariant of (almost) each $\rr{Q}$-bundle. This point of view is developed in Section \ref{subsec:univ-FKMM_inv} and represents a quite important point of novelty with respect to the original definition proposed in \cite{furuta-kametani-matsue-minami-00}. Moreover, our interpretation of the FKMM-invariant
is still liable to further generalizations that provide a notion of {characteristic class} which turns out to be well defined in great generality for $\rr{Q}$-bundles over involutive base spaces of each \virg{reasonable} type.
This point of view will be discussed in a future paper \cite{denittis-gomi-15} (\cf also Remark \ref{rk:gen_FKMM_inv}); 
\end{itemize}

\medskip

The topological classification of the family of   \virg{Quaternionic}  vector bundles strongly depends on the nature of the involutive space $(X,\tau)$ which in turns reflects some physical features of the system under consideration. 
An important aspect of an involutive space is the structure of its \emph{fixed point} set
$
X^{\tau}:= \{x\in X\ |\ \tau(x)=x\}
$.
The family of FKMM-spaces is rich enough to contain all the systems of interest in condensed matter physics:
 \begin{definition}[FKMM-space]\label{ass:2}
An \emph{FKMM-space} is an involutive space $(X,\tau)$ such that:
\begin{enumerate}
\item[0.]  $X$ is a compact and path connected Hausdorff space which admits the structure of a $\Z_2$-CW-complex;\vspace{1.3mm}
\item[1.] The fixed point set $X^\tau\neq\emptyset$  consists of a finite number of points;
\vspace{1.3mm}
\item[2.]  $H^2_{\Z_2}(X,\Z(1))=0$.
\end{enumerate}
\end{definition}
For the sake of completeness, let us recall that an involutive space $(X,\tau)$ has the structure of a $\Z_2$-CW-complex if it admits a skeleton decomposition given by gluing cells of different dimensions which carry a $\Z_2$-action. 
For a precise definition of the notion of 
$\Z_2$-CW-complex, the reader can refer to  \cite[Section 4.5]{denittis-gomi-14} or \cite{matumoto-71,allday-puppe-93}. In addition, the cohomology group that  appears in 2.  is the \emph{equivariant Borel} cohomolgy of the space $(X,\tau)$ computed with respect to the \emph{local system} of coefficients $\Z(1)$. 
A short reminder of this theory is sketched in Section \ref{subsec:borel_cohom}. Here, let us only recall  that the group $H^2_{\Z_2}(X,\Z(1))$ provides the classification for \virg{Real} line bundles  over $(X,\tau)$ where the adjective  \virg{Real} refers to the category of vector  bundles   introduced by M.~F. Atiyah in \cite{atiyah-66} and extensively studied in 
\cite{denittis-gomi-14}.
Given an involutive space $(X,\tau)$ let us denote by ${\rm Vec}^{n}_\rr{Q}(X,\tau)$
the set of isomorphic classes of $\rr{Q}$-bundles with  fiber of dimension $n$. We point out that the existence of fixed points $X^\tau\neq \emptyset$ and the  the connectedness of $X$ imply that the dimension of the fibers is forced to be even, \ie, $n=2m$ (\cf Proposition \ref{prop:Qv0}).  
The key result  at the basis of our classification can be stated as follows:
\begin{theorem}[Injective group homomorphism: low dimensional cases]
\label{theorem:inv_inject2}
Let $(X,\tau)$ be an FKMM-space such that the dimension of the  free $\Z_2$-cells in the skeleton decomposition of $X$  does not exceed $d=3$. Then, the FKMM-invariant defines an injective  map  
$$
\kappa\;:\;{\rm Vec}^{2m}_\rr{Q}(X,\tau)\;\longrightarrow\; H^{2}_{\Z_2}(X|X^\tau,\n{Z}(1))\;\qquad\qquad m\in\N\;.
$$
Moreover, ${\rm Vec}^{2m}_\rr{Q}(X,\tau)$ can be endowed with a group structure in such a way that $\kappa$ becomes an injective group homomorphism. 
\end{theorem}
The proof of this result is postponed to Section \ref{sec:injFKMM-inv} and a precise description of the map $\kappa$ is given in Section \ref{subsec:constr-FKMM_inv}.
 Let us just comment that the abelian group $H^{2}_{\Z_2}(X|X^\tau,\n{Z}(1))$ describes the \emph{relative} equivariant cohomology  associated with the pair $X^\tau\subset X$. The importance of Theorem \ref{theorem:inv_inject2} becomes apparent when we compare it with the well-known similar result for the classification of complex vector bundles over a CW-complex $X$ of dimension $d\leqslant 3$; In this situation one has the isomorphism
 $$
c_1\;:\;{\rm Vec}^{m}_\C(X)\;\longrightarrow\; H^{2}(X,\Z)\;\qquad\qquad m\in\N
$$
 induced by the first Chern class. The parallelism between $\kappa$ and $c_1$ is quite evocative: The FKMM-invariant  is the proper characteristic  class that classifies    \virg{Quaternionic}  vector bundles as elements of the  cohomology group $H^{2}_{\Z_2}(X|X^\tau,\n{Z}(1))$. 
In effect, this turns out to be 
 the correct point of view   for developing a more general theory of the FKMM-invariant (\cf Remark \ref{rk:gen_FKMM_inv}).

 \medskip

 In order to connect the general result of Theorem \ref{theorem:inv_inject2} with the specific problem of the classification of topological insulators of class {\bf AII}
we need to recall some basic facts. 
A typical representative of this class is a system of quantum particles with  half-integer spin and subjected to an \emph{odd} time-reversal symmetry (-TR).
More in detail, let us consider a (self-adjoint) Hamiltonian $\hat{H}$ acting on the Hilbert space $\s{H}:= L^2(\R^d)\otimes \C^L$ (continuous case) or 
$\s{H}:= \ell^2(\Z^d)\otimes \C^L$ (tight-binding approximation) where the number $L\in\N$ takes into account the spinorial degrees of freedom.
We say that $\hat{H}$  is of type {\bf AII} (\cf \cite{altland-zirnbauer-97,schnyder-ryu-furusaki-ludwig-08})
if there exists an \emph{anti-}unitary  symmetry of type $\hat{\Theta}:=\hat{C}\hat{J}$ where 
$\hat{C}$ is the  {complex conjugation}   $\hat{C}\psi=\overline{\psi}$  and $\hat{J}$ is a {unitary operator} that verifies 
\beql{intro:1}
\left\{
\begin{aligned}
\hat{C}\hat{J}\hat{C}&=\;-\hat{J}^\ast\\
\hat{J}\;\hat{H}\;\hat{J}^\ast\;&=\;\hat{C}\hat{H}\hat{C}\;
\end{aligned}
\right.\qquad\quad (\text{{\bf AII} - symmetry})\;.
\eeq
Equation \eqref{intro:1} is equivalent to $\hat{\Theta}\; \hat{H}\; \hat{\Theta}^\ast=\hat{H}$ and the first condition in  
\eqref{intro:1} implies that
$\hat{\Theta}$ is an \emph{anti}-involution in the sense that 
$\hat{\Theta}^2=-\n{1}$ (or equivalently $\hat{\Theta}=-\hat{\Theta}^\ast$). 

\begin{remark}[Topological insulators in class {\bf AI}]\label{rk:classAII}{\upshape
If in  \eqref{intro:1} one replaces the first condition with $\hat{C}\hat{J}\hat{C}=\hat{J}^\ast$, or equivalently 
$\hat{\Theta}^2=\n{1}$, one obtains the class  {\bf AI} of topological insulators. Systems  of this type  possess
 an \emph{even} {time-reversal symmetry} (+TR) and  have been classified in  \cite{denittis-gomi-14}.   
}\hfill $\blacktriangleleft$
\end{remark}

For each $a\in\R^d$ let $\hat{U}_a$ be the unitary operator on $\s{H}$ that implements the translation by $a$ in the sense that  $\hat{U}_a\psi(\cdot)=\psi(\cdot-a)$ for  $\psi\in\s{H}$. The condition $[\hat{H},\hat{U}_a ]=0$ means the invariance of the
Hamiltonian $\hat{H}$  under the translation $\hat{U}_a$.
According to a standard nomenclature, one says that $\hat{H}$ describes a \emph{free} (resp. \emph{periodic}) system
if it is translational invariant for all $a\in\R^d$ (resp. for all  $a\in\Z^d$).
In both cases one can  represent $\hat{H}$ as a fibered operator over the momentum space (one uses the Fourier transform in the free case or  the Bloch-Floquet transform in the periodic case). By repeating verbatim the construction explained in  \cite[Section 2]{denittis-gomi-14} one associates to each gapped spectral region of $\hat{H}$ a vector bundle which is usually called \emph{Bloch-bundle}. The presence of the symmetry 
\eqref{intro:1} endows the {Bloch-bundle} with a \virg{Quaternionic} structure in the sense of \cite{dupont-69}
(we refer to Section \ref{sect:Qu-vb_main} for a precise definition). This construction justifies the following definition.
\begin{definition}[Topological insulators of type {\bf AII}]
\label{def:AII_sys}
A $d$-dimensional \emph{free} system of type {\bf AII} is a {\virg{Quaternionic}} vector bundle over the involutive TR-sphere
\beql{eq:TR-sph0}
\tilde{\n{S}}^d\;:=\;({\n{S}}^d,\tau)
\eeq
where $\n{S}^{d}:=\big\{k\in\R^{d+1}\,|\; \|k\|= 1\big\}$
and the involution $\tau$ is defined by 
\begin{equation}
\label{eq:TR-sph1}
\tau(k_0,k_1,\ldots,k_d)\;:=\;(k_0,-k_1,\ldots,-k_d)\;.
\end{equation}
A $d$-dimensional \emph{periodic} system of type {\bf AII} is a {\virg{Quaternionic}} vector bundle over the involutive TR-torus
\beql{eq:TR-tori0}
\tilde{\n{T}}^d\;:=\;({\n{T}}^d,\tau)
\eeq
where
$
\n{T}^{d}:=\n{S}^{1}\times\ldots\times\n{S}^{1}
$ ($d$-times) 
  and the {involution} $\tau$ extends diagonally the {involution} on $\tilde{\n{S}}^1$  given by \eqref{eq:TR-sph1} in such a way that $\tilde{\n{T}}^d=\tilde{\n{S}}^1\times\ldots\times\tilde{\n{S}}^1$.
\end{definition}

Let us point out that the involutive spaces $\tilde{\n{S}}^d$ and $\tilde{\n{T}}^d$ are particular examples of FKMM-spaces. Conditions 0 and 1  come from the explicit description of the $\Z_2$-CW-complex structure of these spaces 
 \cite[Examples 4.20 \& 4.21]{denittis-gomi-14} and condition
 2 follows from the computation of the equivariant cohomology based on a recursive application of the \emph{Gysin sequence}  \cite[Sections 5.3 \& 5.4]{denittis-gomi-14}. In particular, Theorem \ref{theorem:inv_inject2} applies to the description of ${\rm Vec}_{\rr{Q}}^{2m}({\n{S}}^d,\tau)$ and 
${\rm Vec}_{\rr{Q}}^{2m}({\n{T}}^d,\tau)$ when $d\leqslant 3$.
\begin{theorem}[Classification of {\bf AII} topological insulators: low dimensional case]\label{theo:AI_class}
Let $(\n{S}^d,\tau)$ and $(\n{T}^d,\tau)$ be the TR-involutive spaces described in Definition \ref{def:AII_sys}. Then:
\begin{enumerate}
\item[(i)]  \virg{Quaternionic} vector bundles over $(\n{S}^{d},\tau)$ and $(\n{T}^{d},\tau)$ can have only even rank. In particular  there are no \virg{Quaternionic} line-bundles for all $d\in\N$;\vspace{1.3mm}
\item[(ii)] ${\rm Vec}_{\rr{Q}}^{2m}\big(\n{S}^{1},\tau\big)=0$ for all $m\in\N$;
\vspace{1.3mm}
\item[{(iii)}] For TR-spheres in dimension $d=2,3$ one has group isomorphisms
$$
{\rm Vec}_{\rr{Q}}^{2m}\big(\n{S}^{2},\tau\big)\;=\;\Z_2\;,\qquad\qquad
{\rm Vec}_{\rr{Q}}^{2m}\big(\n{S}^{3},\tau\big)\;=\;\Z_2\;,
$$
given by the FKMM-invariant $\kappa$;
\vspace{1.3mm}
\item[{(iv)}] For TR-tori in dimension $d=2,3$ one has group isomorphisms
$$
{\rm Vec}_{\rr{Q}}^{2m}\big(\n{T}^{2},\tau\big)\;=\;\Z_2\;,\qquad\qquad
{\rm Vec}_{\rr{Q}}^{2m}\big(\n{T}^{3},\tau\big)\;=\;\Z_2^4\;,\
$$
given by the FKMM-invariant $\kappa$. 
\end{enumerate}
\end{theorem}
The various items listed in Theorem \ref{theo:AI_class} are proved separately in the paper:  (i) is a consequence of Proposition \ref{prop:Qv0}; The proof of (ii) is
contained in Proposition \ref{prop_(ii)} and is based on the   \emph{homotopy} classification of \virg{Quaternionic} vector bundles (\cf Theorem \ref{theo:honotopy_class_Q});  
Item (iii) is proved in Proposition \ref{propos_classif_low_S} while the proof of  (iv)  is contained in Proposition \ref{propos_classif_T_2d} and Proposition \ref{propos_classif_T_3d}. Let us spend few words about
 the strategy of the proofs of (iii) and (iv). Firstly, one computes the relevant cohomology groups
$$
H^{2}_{\Z_2}\big(\tilde{\n{S}}^d|(\tilde{\n{S}}^d)^\tau,\n{Z}(1)\big)\;\simeq\;\Z_2\;,\quad\qquad H^{2}_{\Z_2}\big(\tilde{\n{T}}^d|(\tilde{\n{T}}^d)^\tau,\n{Z}(1)\big)\;\simeq\;\Z_2^{2^d-(d+1)} \qquad\qquad\ \ \forall\ d\geqslant 2
$$
(\cf
Proposition \ref{prop:cond_coker} and Proposition \ref{prop:coker_tori}). Then, one shows that the injective group morphism $\kappa$ in Theorem \ref{theorem:inv_inject2} is indeed surjective by constructing suitable explicit realizations of non-trivial $\rr{Q}$-bundles. An important remark is that, due to the low dimensionality of the base space, the FKMM-invariant can be described in terms of a \emph{sign map} $\rr{d}_w:X^\tau\to{\pm1}$  that generalizes, to some extent, the Fu-Kane-Mele index (\cf Remark \ref{rk:FKMM_w_det}). Moreover, in the case of the three-dimensional TR-torus $\tilde{\T}^3$ one recovers the distinction between \emph{week} and \emph{strong} invariants
according to the original definition given in \cite{fu-kane-mele-07} (\cf Remark \ref{rk:FKMM_week_strong}).
Finally, the non-vanishing of the FKMM-invariant can be also interpreted as the obstruction to the 
existence of a global frame of (Bloch) sections which supports the \virg{Quaternionic} symmetry (\cf Remarks \ref{rk:obstr1} \& \ref{rk:FKMM_week_strong}) and this fact recovers and generalizes the point of view investigated in \cite{porta-graf-13} for the particular case of the involutive space $\tilde{\n{T}}^2$.

\medskip

The classification for $d>3$ is more complicated  since in this case the FKMM-invariant $\kappa$ is no longer sufficient to establish an injective morphism between ${\rm Vec}_{\rr{Q}}^{2m}(X,\tau)$ and some cohomology group.  In the case $d=4$ one needs also the \emph{second} Chern class
and, under conditions which are  slightly more restrictive  that those in Definition \ref{ass:2},
one can prove the following result:

\begin{theorem}[Injective group homomorphism: d=4]
\label{theo:injectD4}
Let $(X,\tau)$ be an FKMM-space  and assume in addition that $X$ is a closed and  oriented $4$-manifold
with an involution $\tau$ which is smooth.
Then, the FKMM-invariant $\kappa$ and the second Chern class $c_2$ define a   map
$$
(\kappa,c_2):{\rm Vec}^{2m}_\rr{Q}(X,\tau)\;\longrightarrow\; H^{2}_{\Z_2}(X|X^\tau,\n{Z}(1))\;\oplus\; H^4(X,\Z)\;\qquad\qquad m\in\N
$$
that is injective. Moreover, ${\rm Vec}^{2m}_\rr{Q}(X,\tau)$ can be endowed with a group structure in such a way that the pair $(\kappa,c_2)$ sets  an injective group homomorphism. 
\end{theorem}
The proof of a slightly weaker version of
this theorem is postponed to Section \ref{sec:d=4}, where we prove the injectivity of the pair $(\kappa,c_2)$ under some  extra hypothesis (\cf Assumption \ref{ass:3.1} \& Theorem
\ref{theo:injectD4_weak_version}) which are still verified by the involutive TR-spaces $\tilde{\n{S}}^4$ and $\tilde{\n{T}}^4$.
Let us just comment that the claim of Theorem \ref{theo:injectD4} is true as it is and, at the cost of increasing the technical difficulty of the proof (one needs obstruction theory), it can be proved in full generality \cite{denittis-gomi-15}. 

\medskip

Before discussing the ramifications of Theorem \ref{theo:injectD4} for the classification of topological insulators, it is important to comment about the relevance of the dimension $d=4$ for condensed matter effects. In fact, since $d\leqslant 3$ covers all the physical spatial dimensions, one can consider the case $d=4$ just a mathematical curiosity. However, there are physical phenomena that involve time-periodic adiabatic variations of the system like the \emph{piezoelectric polarization} (see \cite{denittis-lein-13} and references therein) or the \emph{isotropic magneto-electric response} \cite{essin-moore-vanderbilt-09,hughes-prodan-bernevig-11}. In these cases the relevant topological response of the system depends on the full spacetime dimensionality
and higher invariants like the second Chern number become relevant (see \eg \cite{lin-yau-13}).
An interesting difference between the case $d=4$ and the low dimensional cases $d\leqslant 3$ is that the two topological invariants $\kappa$ and $c_2$ are not independent. In particular, it is possible to show that the \emph{strong component} of the FKMM-invariant is uniquely fixed by the \emph{week components} and by the parity of the second \emph{Chern number} $C:=\langle c_2,[X]\rangle\in\Z$ (we denote by $[X]\in H_4(X)$  the fundamental class of $X$).  This fact is made evident in the following result, which completes the classification of topological insulators in class  {\bf AII} for all physically interesting dimensions.
\begin{theorem}[Classification of {\bf AII} topological insulators: $d=4$]\label{theo:AI_classd=4}
Let $(\n{S}^4,\tau)$ and $(\n{T}^4,\tau)$ be the TR-involutive spaces described in Definition \ref{def:AII_sys}. 
\begin{enumerate}
\item[(i)]  
The map
$$
(\kappa,c_2)\;:\;{\rm Vec}^{2m}_\rr{Q}({\n{S}}^4,\tau)\;\longrightarrow\; \Z_2\;\oplus\; \Z
$$
is injective and the image $(\kappa,c_2):\bb{E}\mapsto (\epsilon,C)\in \Z_2 \oplus  \Z$ of
 each \virg{Quaternionic} vector bundle $(\bb{E},\Theta)$ is contained in the subgroup
$$
\left\{(\epsilon,C)\in \Z_2 \oplus  \Z\ |\ \epsilon=(-1)^C\right\}\;\simeq\;\Z\;.
$$
More precisely, elements in ${\rm Vec}^{2m}_\rr{Q}({\n{S}}^4,\tau)\simeq \Z$  are completely classified by the second Chern class $c_2(\bb{E})$ and the value $\kappa(\bb{E})\simeq\epsilon$ of the (strong component of the) FKMM-invariant is fixed by the reduction mod. 2 of the second Chern number $C=\langle c_2(\bb{E}),[\n{S}^4]\rangle$.
\vspace{1.3mm}
\item[(ii)] The map
$$
(\kappa,c_2)\;:\;{\rm Vec}^{2m}_\rr{Q}({\n{T}}^4,\tau)\;\longrightarrow\; \Z_2^{11}\;\oplus\; \Z
$$
is injective and the image $(\kappa,c_2):\bb{E}\mapsto (\epsilon_1,\ldots,\epsilon_{11},C)\in \Z^{11}_2 \oplus  \Z$ of each \virg{Quaternionic} vector bundle $(\bb{E},\Theta)$ is contained in the subgroup
$$
\left\{(\epsilon_1,\ldots,\epsilon_{10},\epsilon_{11},C)\in \Z_2^{11} \oplus  \Z\ \Big|\ \prod_{j=1}^{11}\epsilon_j=(-1)^C\right\}\;\simeq\;\Z_2^{10}\;\oplus\; \Z\;.
$$
More precisely, elements in ${\rm Vec}^{2m}_\rr{Q}({\n{T}}^4,\tau)\simeq \Z_2^{10}\oplus \Z$  are completely classified by the second Chern 
class $c_2(\bb{E})$ and  the FKMM-invariant. However, only the first ten \emph{week} components of the FKMM-invariant  $\kappa(\bb{E})\simeq\epsilon$  are independent since the \emph{strong} component $\epsilon_{11}=(-1)^C\prod_{j=1}^{10}\epsilon_j$
is fixed by the values of  the {week} components and the reduction mod. 2 of the second Chern number $C=\langle c_2(\bb{E}),[\n{T}^4]\rangle$.
\end{enumerate}
\end{theorem}
The proof of this theorem is explained in Section \ref{sec:d=4}.
It should be remarked that item (ii) above was originally shown in \cite{furuta-kametani-matsue-minami-00}.
The content of Theorem \ref{theo:AI_class} and Theorem \ref{theo:AI_classd=4} is
summarized in Table \ref{tab:01}.1  together with the classification of type {\bf A} topological insulators (systems without symmetries, \cf with \cite[Table 1.1]{denittis-gomi-14}). 
 \begin{center}
 \begin{table}[h]\label{tab:01}
 \centering
 \begin{tabular}{|c||c||c|c|c|c||c|}
\hline
VB  & AZC  & $d=1$ & $d=2$&$d=3$&$d=4$&\\
\hline
 \hline
 \rule[-3mm]{0mm}{9mm}
 ${\rm Vec}_{\C}^m(\n{S}^d)$& {\bf A} & $0$ & $\Z$ & $0$ &   \begin{tabular}{ll}
 {$\boxed{0}$}& {($m=1$)}\vspace{1mm}\\
 \ \  $\Z$& ($m\geqslant2$)\\
\end{tabular}   &Free\\
\cline{1-6}
 \rule[-3mm]{0mm}{9mm}
${\rm Vec}_{\rr{Q}}^{2m}(\n{S}^d,\tau)$ & {\bf AII} & $0$ & $\Z_2$ &$\Z_2$&$\Z$&systems\\
\hline
 \hline
  \rule[-3mm]{0mm}{9mm}
 ${\rm Vec}_{\C}^m(\n{T}^d)$ & {\bf A} & $0$ & $\Z$ & $\Z^3$&    \begin{tabular}{ll}
$\boxed{\Z^6}$& ($m=1$)\vspace{1mm}\\	
  \ \ $\Z^7$& ($m\geqslant2$)\\
\end{tabular}   &Periodic\\
\cline{1-6}
 \rule[-3mm]{0mm}{9mm}
${\rm Vec}_{\rr{Q}}^{2m}(\n{T}^d,\tau)$ & {\bf AII} & 0 &$\Z_2$&$\Z_2^4$&
$\Z_2^{10}\oplus\Z$&systems\\
\hline
\end{tabular}\vspace{2mm}
 \caption{The column VB lists the relevant  equivalence classes of vector bundles and the related
 Altland-Zirnbauer-Cartan labels \cite{altland-zirnbauer-97,schnyder-ryu-furusaki-ludwig-08}  are displayed in column AZC.  The  \virg{boxed} entries provide the classification in the \emph{unstable} regime ($d> 2m$)
 which is not covered by the $K$-theoretical classification.
  The rank of the \virg{Quaternionic} vector bundles is forced to be even since the involutive spaces $(\n{S}^d,\tau)$ and $(\n{T}^d,\tau)$
 described in Definition \ref{def:AII_sys} have fixed points.}
 \end{table}
 \end{center}
A comparison between Table \ref{tab:01}.1 and  the calculations in Appendix \ref{sect:KQ-theory} shows that our classification completely agrees with the predictions supplied by  $K$-theory. In particular, for \virg{Quaternionic} vector bundles in dimension $d\leqslant 4$ there is no distinction between \emph{stable} and \emph{unstable} regime
and so the $K$-theory provides a precise description for the \virg{labeling sets} of various isomorphism classes of $\rr{Q}$-bundles. Nevertheless, our classification is strictly stronger than the $K$-theoretical analysis since it provides an explicit description for the classifying invariants as elements of a proper cohomology theory. This information is not trivial at all! On the contrary, it plays a prominent role in the description  of physical effects 
like the stability of spin currents.

\medskip
 
Before ending this introductory part, we would briefly compare our results with the existing  physical  literature on the classification of topological insulators in class {\bf AII}. Table \ref{tab:01} is certainly not a novelty in the literature and, as already mentioned, the first classification of the distinct topological phases of  time-reversal invariant fermionic systems has been described in  \cite{kane-mele-05',kane-mele-05,fu-kane-06} for the case $d=2$ and 
then extended to $d=3$ in  \cite{fu-kane-mele-07} where the distinction between strong and weak phases appeared for the first time. In these works the $\Z_2$-order is characterized by the existence of stable spin currents (QSHE) or by the net effect of a \virg{spin pump}. Mathematically, the different phases are distinguished by the signs, the so-called Fu-Kane-Mele indices, that a suitable (Pfaffian) function of    the Bloch waves  takes on the fixed points of the torus (see Remark \ref{rk:FKMM_w_det} for more details). In \cite{moore-balents-07} the authors proposed a different description of the  Fu-Kane-Mele indices in terms of an implicit version of the $\Z_2$-equivariant homotopy classification described in Theorem \ref{theo:honotopy_class_Q}.
Subsequently,  an independent  elegant argument has been developed  in \cite{roy-09}  where the classification for  $d=2,3$ is recovered
by considering the 
exchange of Chern numbers during an adiabatic process of band-touching. This line of reasoning has its abstract counterpart in our Corollary \ref{coroll:d2manif_FKMM}.
Our approach to the classification of topological phases of quantum systems subjected to an odd time-reversal symmetry differs from these pervious researches essentially in two aspects. First of all, our approach based on the definition of the FKMM-invariant is independent on the particular nature of the base space, while all the previous techniques are  designed  \emph{ad hoc} over the peculiar structure of the torus. This clarification is not of secondary importance, since condensed matter systems are not the only sources of quantum topological effects; \eg our results are applicable to certain topological quantum field theories with symmetries as well as to time-reversal quantum systems perturbed by external fields parametrized by the points of some complicated manifold. 
Second, the FKMM-invariant used for our classification is a cohomological class, a fact that has two immediate implications: cohomology is algorithmically computable as opposed to homotopy and cohomological classes are liable to be described in terms of differential forms \cite{denittis-gomi-15}.
None of the invariants previously considered in the literature manifestly shows this important cohomological nature, although all of them are described by our FKMM-invariant (\cf Section \ref{subsec_Fu-Kane-Mele}).

\medskip

As a final remark let us point out that one of the merits of the present  work is that it shows how the FKMM-invariant can be understood as a \emph{bona fide} characteristic class for the category of \virg{Quaternionic} vector bundles (\cf Section \ref{subsec:univ-FKMM_inv}). This discovery, which in our opinion may have future implications, is an important point of novelty with respect to the original (and certainly inspiring) ideas contained in \cite{furuta-kametani-matsue-minami-00}.  Let us also recall that in the literature there already exist  works devoted to the construction of characteristic classes for \virg{Quaternionic} vector bundles. Among these, we mention at least the two papers
\cite{dos-santos-lima-filho-04,lawson-lima-filho-michelsohn-05} in which the authors developed the notion of \virg{Quaternionic} Chern classes. We feel that it should be of some  interest and utility to understand the link between the FKMM-invariant described in this work and these  \virg{Quaternionic}  Chern classes.

\medskip

\noindent
{\bf Acknowledgements.} 
The authors are immensely grateful to M. Furuta 
for allowing them to use ideas contained in the preprint \cite{furuta-kametani-matsue-minami-00}. 
GD wants to thank F. Rivera for the excellent hospitality at SJIH, San Juan, Puerto Rico where this investigation has begun.
GD's research is supported
 by
the Alexander von Humboldt Foundation. 
KG's research is supported by 
the Grant-in-Aid for Young Scientists (B 23740051), JSPS.
\medskip

\section{\virg{Quaternionic} vector bundles}
\label{sect:Q-vb}

This section is devoted to the description of the category of {\virg{Quaternionic}} vector bundles introduced for the first time in
 \cite{dupont-69}.  
Through the paper we often use the shorter expression \emph{$\rr{Q}$-bundle} instead of {\virg{Quaternionic}} vector bundle.

\subsection{\virg{Quaternionic} structure on vector bundles}
\label{sect:Qu-vb_main}

The first ingredient to define a \virg{Quaternionic} structure on a complex vector bundle is an \emph{involution} on the base space.
We recall that
an {involution} $\tau$ on a topological space $X$ is a homeomorphism of period 2, \ie  $\tau^2={\rm Id}_X$. The pair 
 $(X,\tau)$ will be called an  \emph{involutive space}. The spaces  
 $\tilde{\n{S}}^d$ and $\tilde{\n{T}}^d$ described in Definition \ref{def:AII_sys} are examples of involutive spaces. Other examples have been discussed in \cite[Section 4.1]{denittis-gomi-14}. We tacitly assume  through the paper that all the involutive spaces $(X,\tau)$ verify at least condition 0 in Definition \ref{ass:2}.

\medskip

A {\virg{Quaternionic}} vector bundle, or $\rr{Q}$-bundle, over  $(X,\tau)$
is a complex vector bundle $\pi:\bb{E}\to X$ endowed with a (topological) homeomorphism $\Theta:\bb{E}\to \bb{E}$
 such that:
\begin{itemize}

\item[$(\rr{Q}_1)$] the projection $\pi$ is \emph{equivariant} in the sense that $\pi\circ \Theta=\tau\circ \pi$;

\item[$(\rr{Q}_2)$] $\Theta$ is \emph{anti-linear} on each fiber, \ie $\Theta(\lambda p)=\overline{\lambda}\ \Theta(p)$ for all $\lambda\in\C$ and $p\in\bb{E}$ where $\overline{\lambda}$ is the complex conjugate of $\lambda$;

\item[$(\rr{Q}_3)$] $\Theta^2$ acts fiberwise as the multiplication by $-1$, namely $\Theta^2|_{\bb{E}_x}=-\n{1}_{\bb{E}_x}$.

\end{itemize}
It is always possible to endow
$\bb{E}$ with
a (essentially unique) Hermitian metric with respect to which $\Theta$ is an \emph{anti-unitary} map between conjugate fibers
(\cf Proposition \ref{rk:eq_metric}).
 
  \medskip
 
A vector bundle \emph{morphism} $f$ between two vector bundles  $\pi:\bb{E}\to X$ and $\pi':\bb{E}'\to X$ 
over the same base space
is a  continuous map $f:\bb{E}\to \bb{E}'$ which is \emph{fiber preserving} in the sense that  $\pi=\pi'\circ f$
 and that restricts to a \emph{linear} map on each fiber $\left.f\right|_x:\bb{E}_x\to \bb{E}'_x$. Complex (resp. real) vector bundles over  $X$ together with vector bundle morphisms define a category and we use ${\rm Vec}^m_\C(X)$
 (resp. ${\rm Vec}^m_\R(X)$) to denote the set of equivalence classes of isomorphic vector bundles of rank $m$.
    Also 
 $\rr{Q}$-bundles define a category with respect to  \emph{$\rr{Q}$-morphisms}. A $\rr{Q}$-morphism $f$ between two  $\rr{Q}$-bundles
 $(\bb{E},\Theta)$ and $(\bb{E}',\Theta')$ over the same involutive space $(X,\tau)$ 
  is a vector bundle morphism  commuting with the involutions, \ie $f\circ\Theta\;=\;\Theta'\circ f$. The set of equivalence classes of isomorphic $\rr{Q}$-bundles of  rank $m$ over $(X,\tau)$ 
  is denoted with ${\rm Vec}_{\rr{Q}}^m(X,\tau)$. 
  
\medskip

The set ${\rm Vec}^m_\C(X)$ is non-empty since it contains at least the
\emph{product} vector bundle $X\times\C^m\to X$ with canonical projection $(x,{\rm v})\mapsto x$. Similarly, in the real case one has that 
$X\times\R^m\to X$ provides  an element of ${\rm Vec}^m_\R(X)$. A complex (resp. real) vector bundle is 
called \emph{trivial}  if it is isomorphic to the  complex (resp. real) product vector bundle. In order to extend these definitions to \virg{Quaternionic} vector bundles we need to investigate the structure of the fibers of a $\rr{Q}$-bundles $(\bb{E},\Theta)$ over fixed points of the base space $(X,\tau)$. Let $x\in X^\tau$ and $\bb{E}_x\simeq\C^m$ be the related fiber. In this case the restriction $\Theta|_{\bb{E}_x}\equiv J$ defines an  \emph{anti}-linear map $J : \bb{E}_x \to \bb{E}_x$ such that $J^2 = -\n{1}_{\bb{E}_x}$. This means that the  fibers $\bb{E}_x$ over fixed points $x\in X^\tau$ are endowed with a \emph{quaternionc} structure in the following sense:

\begin{remark}[Quaternionic structure over complex vector spaces]\label{rk:quat_str}{\upshape
We shall denote with $\n{H}$ the skew-field of  quaternions and by $(1,{\rm i},{\rm j},{\rm k})$ its usual basis over $\R$, 
$$
\n{H}\;=\;\R\;\oplus\;\R\;{\rm i}\;\oplus\;\R\;{\rm j}\;\oplus\;\R\;{\rm k}\;\qquad\quad \big({\rm i}^2\;=\;{\rm j}^2\;=\;{\rm k}^2\;=\;{\rm i}{\rm j}{\rm k}\;=\;-1\big)\;.
$$
Similarly, the pair $(1,{\rm j})$ provides a basis of $\n{H}$ over $\C$
$$
\n{H}\;=\;\C\;\oplus\;\C\;{\rm j}\;=\;(\R\;\oplus\;\R\;{\rm i})\;\oplus\;(\R\;\oplus\;\R\;{\rm i})\;{\rm j}\;.
$$
where the relation ${\rm i}{\rm j}={\rm k}$ has been used.
Let $\bb{V}$ be a complex vector space of complex dimension $n$. One says  that $\bb{V}$ has a \emph{quaternionic} structure if there is an \emph{anti}-linear map $J:\bb{V}\to \bb{V}$ such that $J^2=-\n{1}$ (\cf \cite[Section 1]{vaisman-90} and references therein). 
A complex vector space $\bb{V}$ admits a {quaternionic structure} if and only if it has even complex dimension $n=2m$ and in this case the  pair $(\bb{V},J)$  turns out to be isomorphic to the space $\n{H}^m=(\C\oplus\C\;{\rm j})^m$ understood as \emph{left}-module over $\C$ and 
endowed with the \emph{left} multiplication by ${\rm j}$. Since ${\rm j}{\rm i}=-{\rm i}{\rm j}$ this multiplication is automatically \emph{anti}-linear with respect to the complex structure.
Said differently, we can
 identify $(\bb{V},J)$ with $\C^{2m}$ endowed with the \emph{standard} quaternionic structure ${\rm v}\mapsto Q\;\overline{\rm v}$ where $\overline{\rm v}$ is the complex conjugate of ${\rm v}$ and $Q$ denotes the real matrix
\beql{eq:Q-mat}
Q
\;:=\;
\left(
\begin{array}{rr|rr|rr}
0 & -1 &        &        &   &    \\
1 &  0 &        &        &   &    \\
\hline
  &    & \ddots &        &   &    \\
  &    &        & \ddots &   &    \\
\hline
  &    &        &        & 0 & -1 \\
  &    &        &        & 1 &  0
\end{array}
\right)\;.
\eeq
Let us recall that the quaternionic unitary group $\n{U}_{\n{H}}(m)$ coincides with the symplectic group $\n{S}p(m)$. Moreover, $\n{U}_{\n{H}}(m)$ has the following characterization: let $\mu:\n{U}(2m)\to\n{U}(2m)$ be the involution defined by $\mu(U):=-Q\overline{U}Q$, then $\n{U}_{\n{H}}(m)\simeq \n{U}(2m)^\mu$ where   $\n{U}(2m)^\mu$ is the set of fixed points under the action of  $\mu$. Finally $\n{U}(2m)^\mu\subset{\rm S}\n{U}(2m)$. Indeed if $\mu(U)=U$ and $U{\rm v}=\lambda{\rm v}$ for some $\lambda\in\C$ and ${\rm v}\in\C^{2m}$ then also $U(Q\overline{\rm v})=\overline{\lambda}(Q\overline{\rm v})$
and the vectors ${\rm v}$ and $Q\overline{\rm v}$ are linearly independent (even in the case $\lambda=\pm1$). Iterating this procedure one shows that the set of eigenvalues of $U$ is given by $m$ pairs $\{\lambda_i,\overline{\lambda_i}\}_{i=1,\ldots,m}$ such that $|\lambda_i|=1$ and this implies ${\rm det}(U)=1$.}\hfill $\blacktriangleleft$
\end{remark}

\noindent
The first consequence of Remark \ref{rk:quat_str} is: 

\begin{proposition}\label{prop:Qv0}
Let $(X,\tau)$ be an involutive and path-connected space. If $X^\tau\neq \emptyset$ then every $\rr{Q}$-bundle over $(X,\tau)$ necessarily  has even rank.
\end{proposition}
\proof
Fibers over fixed points needs an even complex dimension in order to support  a quaternionic structure. Moreover, if the base space $X$ is path-connected the dimension of the fibers is constant.
\qed
\begin{remark}[$\rr{Q}$-bundles of odd rank]\label{rk:quat_vb_od}{\upshape
In Proposition \ref{prop:Qv0} the condition $X^\tau\neq \emptyset$ can not be removed. In fact, if the base space $X$ is endowed with a free involution then it is possible to realize  $\rr{Q}$-bundles with fibers of odd rank.
For instance, an example of $\rr{Q}$-line bundle  has been worked out in \cite{dupont-69}:  Let $\Z_2 = \{ \pm 1 \}$  endowed with the free involution $\tau:\epsilon \mapsto -\epsilon$, $\epsilon\in\{\pm 1\}$. Then the complex line bundle $\bb{L} = \Z_2 \times \C$ gives rise to a \virg{Quaternionic} vector bundle by $(\epsilon,z)\mapsto (-\epsilon,\epsilon\overline{z})$.
This example, together with Proposition \ref{prop:Qv0}, shows that a base space with a free involution is a necessary condition for the construction of a $\rr{Q}$-bundle with odd fibers.
}\hfill $\blacktriangleleft$
\end{remark}

The set
 ${\rm Vec}_{\rr{Q}}^{2m}(X,\tau)$ is non-empty since it contains at least the \emph{\virg{Quaternionic} product  bundle} 
 $X\times\C^{2m}\to X$ endowed with the \emph{product} $\rr{Q}$-structure
$\Theta_0 (x,{\rm v})=(\tau(x),Q\;\overline{{\rm v}})$ 
where the matrix $Q$ is the same as in \eqref{eq:Q-mat}.
Moreover, as a consequence of  Proposition \ref{prop:Qv0}, this is   the only type of product  $\rr{Q}$-bundle which is possible to build if $X^\tau\neq \emptyset$.
We say that a $\rr{Q}$-bundle is \emph{$\rr{Q}$-trivial}  if it is isomorphic  to the product $\rr{Q}$-bundle in the category of $\rr{Q}$-bundles.
Since the Whitney sum of a rank $2$ product  $\rr{Q}$-bundle  defines a map ${\rm Vec}_{\rr{Q}}^m(X,\tau) \to {\rm Vec}_{\rr{Q}}^{m+2}(X,\tau)$ 
  one introduces the inductive limit ${\rm Vec}_{\rr{Q}}(X,\tau):=\bigcup_{m\in\N}{\rm Vec}_{\rr{Q}}^{2m}(X,\tau)$ which describes isomorphism classes of $\rr{Q}$-bundles over involutive spaces  with fixed points independently of  the (even) rank of the fibers.

\medskip

The name  of \virg{Quaternionic} vector bundles (in the category of involutive spaces) for elements in ${\rm Vec}_{\rr{Q}}(X,\tau)$
is justified by the following result:

\begin{proposition}\label{prop:Qv1.0}
Let ${\rm Vec}^m_{\n{H}}(X)$ be the set of equivalence classes of vector  bundles over $X$ with typical fiber $\n{H}^m$. Then,
$$
{\rm Vec}^m_{\n{H}}(X)\;\simeq\;{\rm Vec}^{2m}_{\rr{Q}}(X,{\rm Id}_X)\qquad\quad \forall\ m\in\N\;.
$$ 
\end{proposition}

\proof[
sketch of]
Let $\bb{E}$ be an element of ${\rm Vec}^m_{\n{H}}(X)$. Each fiber of $\bb{E}_x$ is a \emph{left} $\n{H}$-module. Now,  if one considers $\bb{E}_x$ simply as a \emph{left} $\n{C}$-module endowed with an extra \emph{left} multiplication by ${\rm j}$ one obtains, by virtue of Remark \ref{rk:quat_str}, a map
 $\rr{c}:{\rm Vec}_{\n{H}}^m(X)\to{\rm Vec}_{\rr{Q}}^{2m}(X,{\rm Id}_X)$. On the other side, if  
 $\bb{E}$ is an element of ${\rm Vec}^{2m}_{\rr{Q}}(X,{\rm Id}_X)$ then each fiber $\bb{E}_x$ turns out to be a complex vector space of dimension $2m$ endowed with a quaternionic structure so that  $\bb{E}_x\simeq\n{H}^m$. This leads to a map $\rr{q}:{\rm Vec}^{2m}_{\rr{Q}}(X,{\rm Id}_X)\to{\rm Vec}^m_{\n{H}}(X)$.
By construction one verifies that  $\rr{c}$ and $\rr{q}$
are inverses of each other.
  \qed

\medskip

Given a \virg{Quaternionic} bundle $(\bb{E},\Theta)$ over the involutive space $(X,\tau)$ we can \virg{forget} the $\rr{Q}$-structure
and consider only the complex vector bundle $\bb{E}\to X$. This forgetting procedure goes through isomorphism classes. In fact,  a $\rr{Q}$-isomorphism between two $\rr{Q}$-bundles is, in particular, an isomorphism of complex vector bundles plus an extra condition of equivariance which is lost under the process of forgetting the \virg{Quaternionic} structure. 
\begin{proposition}\label{prop:Qv101}
The process of forgetting the \virg{Quaternionic} structure defines a map
$$
\jmath\;:\;{\rm Vec}^m_{\rr{Q}}(X,\tau)\;\longrightarrow \; {\rm Vec}^m_{\C}(X)
$$
such that $\jmath:[0]\to[0]$ where  $[0]$ denotes the  trivial class in the appropriate category.
\end{proposition}

\subsection{\virg{Quaternionic} sections}\label{sect:Qs-vb_sect}
Let $\Gamma(\bb{E})$ be the set of  \emph{sections}
of a $\rr{Q}$-bundle $(\bb{E},\Theta)$ over the involute space $(X,\tau)$. We recall that a section $s$
is a continuous maps $s:X\to\bb{E}$ such that $\pi\circ s={\rm Id}_X$
where $\pi:\bb{E}\to X$ is the bundle projection.
 The set $\Gamma(\bb{E})$ has the structure of a module over the algebra $C(X)$  
and inherits from the \virg{Quaternionic} structure of $(\bb{E},\Theta)$  an \emph{anti}-linear \emph{anti}-involution $\tau_\Theta: \Gamma(\bb{E})\to \Gamma(\bb{E})$ defined by 
$$
\tau_\Theta(s)\;:=\;\Theta\;\circ\; s\;\circ\;\tau\;.
$$
This means that the $C(X)$-module $\Gamma(\bb{E})$
is endowed with a quaternionic structure (in the jargon of Remark \ref{rk:quat_str}) given by $\tau_\Theta$ and the left multiplication by ${\rm i}$.

\medskip

The product $\rr{Q}$-bundle $(X\times \C^{2m},\Theta_0)$ over   $(X, \tau)$ has a special family of sections $\{r_1,\ldots,r_{2m}\}$ given by $r_j:x\mapsto (x,{\rm e}_j)$ with ${\rm e}_j:=(0,\ldots,0,1,0,\ldots,0)$ the $j$-th vector of the canonical  basis of $\C^m$. 
These sections verify
\beql{eq:Q_sec0.1}
\tau_{\Theta_{0}}(r_j)(x)\;=\;\Theta_{0}\big(\tau(x),{\rm e}_j\big)\;=\;(x,Q\;\overline{{\rm e}}_j)=:\; (\bb{Q}\;r_j)(x)
\eeq
where we exploited the reality  of the canonical basis ${\rm e}_j=\overline{{\rm e}}_j$. The matrix $Q$ is the one defined in \eqref{eq:Q-mat} and  the {constant} endomorphism
$\bb{Q}\in\Gamma\big({\rm End}(X\times \C^{2m})\big)$
is   specified pointwise by the last equality in \eqref{eq:Q_sec0.1}. Let us point out that $\bb{Q}$ is invertible, real in the sense $\tau_{\Theta_{0}}\circ \bb{Q}=\bb{Q}\circ \tau_{\Theta_{0}}$ and
anti-involutive $\bb{Q}^2=-{\rm Id}$.
 Moreover, 
for all $x\in X$ the vectors $\{r_1(x),\ldots,r_{2m}(x)\}$ provide a complete basis for the fiber $\{x\}\times\C^{2m}$ over $x$. 

\medskip

For a  $\rr{Q}$-bundle $(\bb{E},\Theta)$ this kind of behavior is \emph{locally} general. Let $s_1\in\Gamma(\bb{E})$ and $\f{U}\subset X$ a $\tau$-invariant open set such that $s_1(x)\neq 0$ for all $x\in{\f{U}}$. Set $s_2:=\tau_\Theta(s_1)$ which implies $\tau_\Theta(s_2)=-s_1$.  Since $\Theta$ is a
homeomorphism  one has also $s_2(x)\neq 0$ for all $x\in{\f{U}}$. Moreover, it is easy to check that 
 $s_1$ and $s_2$ are linearly independent. In fact, if one assumes that 
 $s_1=\lambda s_2$ then the application of $\tau_\Theta$ to both sides provides $s_2=-\bar{\lambda} s_1$ which is possible if and only if $\lambda=0$. We say that $(s_1,s_2)$ is a
 {\virg{Quaternionic} pair} (or a \emph{Kramers pair} using a physical terminology!) over $\f{U}$.
 Now, let us add  a third section $s_3$ which is independent from $s_1$ and $s_2$ and such that $s_3(x)\neq 0$ for all $x\in{\f{U}}$. The section $s_4:=\tau_\Theta(s_3)$ does not vanish on $\f{U}$ and  is independent from $s_1,s_2,s_3$. The last claim
 can be easily proved along the same strategy used for the independence between $s_1$ and $s_2$. If $\bb{E}$ has rank $2m$
we can iterate this procedure $m-1$ times  and we end up with a family of sections $\{s_1,s_2,s_3,s_4\ldots,s_{2m-1},s_{2m}\}$  which are
 independent and non-zero over $\f{U}$ and that verify the relations 
$s_{2j}:=\tau_\Theta(s_{2j-1})$ for all $j=1,\ldots,m$. In other words, we have realized
 a frame for $\bb{E}|_{\f{U}}$ made by \virg{Quaternionic} pairs $(s_{2j-1},s_{2j})$. This implies also 
$\tau_\Theta(s_j)=\bb{Q}s_j$ where the endomorphism
$\bb{Q}\in\Gamma\big({\rm End}(\bb{E}|_{\f{U}})\big)$ acts as the multiplication by the matrix $Q$ with respect to the local basis $\{s_1,s_2,\ldots,s_{2m-1},s_{2m}\}$. This discussion justifies the following definition:
\begin{definition}[Global $\rr{Q}$-frame]
\label{def:Q_pair}
Let $(\bb{E},\Theta)$ be a \virg{Quaternionic} vector bundle of rank $2m$ over the involutive space $(X,\tau)$. We say that $(\bb{E},\Theta)$ admits a global $\rr{Q}$-frame if there is a collection of sections $\{s_1,s_2,s_3,s_4\ldots,s_{2m-1},s_{2m}\}\subset\Gamma(\bb{E})$ such that:
\begin{enumerate}
\item[(a)] For each $x\in X$ the set of vectors $\{s_1(x),s_2(x),\ldots,s_{2m-1}(x),s_{2m}(x)\}$ spans the fiber $\bb{E}_x$ over $x$;
\item[(b)] $s_{2j}=\tau_\Theta(s_{2j-1})$ for all $j=1,\ldots,m$.
\end{enumerate}
\end{definition}

\noindent
The existence of a global $\rr{Q}$-frame characterizes the $\rr{Q}$-triviality of a \virg{Quaternionic} vector bundle.

\begin{theorem}[$\rr{Q}$-triviality]
\label{theo:triv_glob}
An even rank $\rr{Q}$-bundle $(\bb{E},\Theta)$ over $(X,\tau)$   is $\rr{Q}$-trivial if and only if it  admits a global $\rr{Q}$-frame.
\end{theorem}
\proof
Let us assume that  $(\bb{E},\Theta)$ is $\rr{Q}$-trivial. This means that there is a $\rr{Q}$-isomorphism  $h:X\times \C^{2m}\to\bb{E}$ between $(\bb{E},\Theta)$ and the product $\rr{Q}$-bundle $(X\times \C^{2m},\Theta_0)$. Let us define sections  $s_j\in \Gamma(\bb{E})$ by $s_j:=h\circ r_j$ where $\{r_1,r_2,\ldots,r_{2m-1},r_{2m}\}$ is  the global $\rr{Q}$-frame of the product bundle. The fact that $h$ is an isomorphism implies that $\{{s}_1,s_2,\ldots,s_{2m-1},{s}_m\}$ spans each fiber of $\bb{E}$. Moreover, the equivariance condition $\Theta\circ h=h\circ \Theta_{0}$ implies that
$
\tau_{\Theta}(s_{2j-1})=\Theta\circ h\circ r_{2j-1}\circ \tau=h\circ\tau_{\Theta_0}(r_{2j-1})=h\circ r_{2j}=s_{2j}
$
and this shows that $\{{s}_1,s_2,\ldots,s_{2m-1},{s}_m\}$ is a global $\rr{Q}$-frame.

Conversely, let us assume that $(\bb{E},\Theta)$ has a  global $\rr{Q}$-frame $\{{s}_1,s_2,\ldots,s_{2m-1},{s}_{2m}\}$. For each $x\in X$ we can set the linear isomorphism $h_x:\{x\}\times\C^{2m}\to \bb{E}_x$ defined by $h_x(x,{\rm e}_j):=s_j(x)$. The collection of $h_x$ defines an isomorphism $h:X\times\C^{2m}\to \bb{E}$ between complex vector bundles \cite[Theorem 2.2]{milnor-stasheff-74}. Moreover, 
$\Theta\circ h(x,{\rm e}_{2j-1})=\Theta\circ s_{2j-1}(x)=s_{2j}(\tau(x))=h(\tau(x),{\rm e}_{2j})=h\circ\Theta_{0}(x,{\rm e}_{2j-1})$ for all $x\in X$
and all $j=1,\ldots,m$ and this proves that $h$ is a $\rr{Q}$-isomorphism.
\qed

\subsection{Local $\rr{Q}$-triviality}\label{sect:Q-vb_loc_triv}

A $\rr{Q}$-bundle is locally trivial in the category  of complex vector bundles by definition. Less obvious is that  a $\rr{Q}$-bundle is also locally trivial in the category of vector bundles over an involutive space.
In order to discuss this point we start with a classical result.

\begin{lemma}[Extension]\label{lemma:ext}
Let $(X,\tau)$ be an involutive space and assume that $X$ verifies (at least) condition 0 of Definition \ref{ass:2}. Let
 $(\bb{E},\Theta)$  be a $\rr{Q}$-bundle over $(X,\tau)$ and $Y\subset X$  a closed subset such that $\tau(Y)=Y$.
 Then each  $\rr{Q}$-pair $\{\tilde{s}_1,\tilde{s}_2\}\in\Gamma(\bb{E}|_Y)$ extends to a $\rr{Q}$-pair $\{s_1,s_2\}\subset\Gamma(\bb{E})$.
\end{lemma}
\proof
By definition of $\rr{Q}$-pair one has $\tilde{s}_2=\tau_\Theta(\tilde{s}_1)$.
Using  \cite[Lemma 1.1]{atiyah-bott-64} we know that we can extend 
$\tilde{s}_1$ to a section $s_1\in\Gamma(\bb{E})$. By setting $s_2:=\tau_\Theta({s}_1)$ one has that $(s_1,s_2)$ is a $\rr{Q}$-pair and
$\left.s_2\right|_Y=\tilde{s}_2$.
\qed

\begin{proposition}[Local $\rr{Q}$-triviality]
\label{prop:loc_triv}
Let $(\bb{E},\Theta)$ be a $\rr{Q}$-bundle of even rank
over the involutive space $(X,\tau)$ such that $X$  verifies (at least) condition 0 of Definition \ref{ass:2}.
 Then, $\pi:\bb{E}\to X$ is \emph{locally $\rr{Q}$-trivial} meaning that for all $x\in X$ there exists a $\tau$-invariant neighborhood $\f{U}$ of $x$ and a $\rr{Q}$-isomorphism $h:\pi^{-1}(\f{U})\to \f{U}\times\C^{2m}$ 
 with respect to the trivial $\rr{Q}$-structure on the product bundle $\f{U}\times\C^{2m}$
 given by
$\Theta_0 (x,{\rm v})=(\tau(x),Q\;\overline{{\rm v}})$ 
  (the matrix $Q$ is defined by \eqref{eq:Q-mat}). Moreover, if $x\in X^\tau$  the neighborhood $\f{U}$  can be chosen connected otherwise,  when $x\neq\tau(x)$,   $\f{U}$ can be taken as the union of two disjoint open sets $\f{U}:=\f{U}'\cup\f{U}''$ with $x\in \f{U}'$
and $\tau:\f{U}'\to \f{U}''$ an homeomorphism.
\end{proposition}
\proof
This proof is an adaption  of the argument in \cite[Lemma 1.2]{atiyah-bott-64} and of the discussion in  \cite[pg. 374]{atiyah-66}.
Let us start with the case of a $x\in X^\tau$. On  the fiber $\bb{E}_x\simeq\C^{2m}$  the procedure described in Section \ref{sect:Qs-vb_sect} leads to a basis of vectors $\{s_1(x),s_2(x),\ldots,s_{2m-1}(x),s_{2m}(x)\}$ such that each $(s_{2j-1}(x),s_{2j}(x))$, $j=1,\ldots,m$, is a $\rr{Q}$-pair with respect to $\tau_\Theta$. By the extension Lemma \ref{lemma:ext} we can extend these vectors to a family of sections $\{s_1,s_2,\ldots,s_{2m-1},s_{2m}\}\subset\Gamma(\bb{E})$ formed by $\rr{Q}$-pairs $(s_{2j-1},s_{2j})$. Moreover, there exists an open 
neighborhood $\f{U'}$ of $x$ where this family of sections behaves as a global frame  for $\bb{E}|_{\f{U'}}$ (this is a consequence of the fact that the linear group ${\rm GL}_{2m}(\C)$ is open). In order to have a $\tau$-invariant neighborhood it is enough to consider $\f{U}:=\f{U'}\cap\tau(\f{U'})$. Moreover, since $X$ is assumed to be connected,  $\f{U}$ is at least locally connected and so it can be chosen sufficiently small to be connected around the fixed point $x$.
The  family of sections $\{s_1,s_2,\ldots,s_{2m-1},s_{2m}\}$ provides a global $\rr{Q}$-frame for $\bb{E}|_{\f{U}}$, hence Theorem \ref{theo:triv_glob}
assures the existence of a $\rr{Q}$-isomorphism $h$ between 
$\bb{E}|_{\f{U}}=\pi^{-1}(\f{U})$ and the product $\rr{Q}$-bundle $\f{U}\times\C^{2m}$.

In the case $x\neq\tau(x)$ we can start with the construction of a global $\rr{Q}$-frame defined on the closed set $Y:=\{x,\tau(x)\}$.
To do this let us consider a nowhere vanishing section $\tilde{s}_1\in\Gamma(\bb{E}|_Y)$ defined by a pair of non-zero vectors $\tilde{s}_1(x)\in \bb{E}_x$ and $\tilde{s}_1(\tau(x))\in\bb{E}_{\tau(x)}$. The section $\tilde{s}_2:=\tau_\Theta(\tilde{s}_1)$ is a well defined nowhere vanishing element of 
$\Gamma(\bb{E}|_Y)$ which is independent of $\tilde{s}_1$ for the same argument sketched in Section \ref{sect:Qs-vb_sect}.
After iterating the procedure $m-1$ times by adding at each step a new independent $\rr{Q}$-pair, one ends up eventually with a global $\rr{Q}$-frame for $\bb{E}|_Y$. At this point the proof  proceeds exactly as before: the extension Lemma \ref{lemma:ext}  and a continuity argument assure the existence of a $\tau$-invariant neighborhood $\f{U}\supset Y$ such that $\bb{E}|_{\f{U}}$ admits a global $\rr{Q}$-frame.  Clearly  we can  choose $\f{U}$ sufficiently localized around $x$ and $\tau(x)$ in such a way that the requirements in the claim are fulfilled.\qed

\medskip

One of the consequences of Proposition \ref{prop:loc_triv} is  that it allows us to define $\tau$-invariant \emph{partitions of unity} subordinate to  a given  $\tau$-invariant  covering $\{\f{U}_i\}$
associated with a $\rr{Q}$-trivialization of the $\rr{Q}$-bundle;  $\tau$-invariant {partitions of unity} can be then used to build \emph{equivariant} Hermitian  metrics compatible with the \virg{Quaternionic} structure (\cf \cite[Remark 4.11]{denittis-gomi-14}). The same result can be achieved also with a direct average of the metric. Let $(\bb{E},\Theta)$ be a $\rr{Q}$-bundle over the involutive space $(X,\tau)$ and consider the set $\bb{E}\times_X \bb{E}:=\{(p_1,p_2)\in\bb{E}\times \bb{E}\ |\ \pi(p_1)=\pi(p_2)\}$
associated with the underlying complex vector bundle $\pi:\bb{E}\to X$. A Hermitian  metric is a map $\rr{m}':\bb{E}\times_X \bb{E}\to\C$ which is a positive-definite Hermitian form on each fiber.
By a standard result \cite[Chapter I, Theorem 8.7]{karoubi-97} we know that each complex vector bundle over a paracompact base space admits a 
Hermitian  metric. Moreover if $\rr{m}'$ and $\rr{m}''$ are two different Hermitian  metrics for $\pi:\bb{E}\to X$ there exists an isomorphism $f:\bb{E}\to\bb{E}$ such that $\rr{m}'(p_1,p_2)=\rr{m}''(f(p_1),f(p_2))$ \cite[Chapter I, Theorem 8.8]{karoubi-97}. This mans that the choice of a Hermitian  metric is essentially unique.
A Hermitian  metric compatible with the \virg{Quaternionic} structure must verify the condition $\rr{m}(\Theta(p_1),\Theta(p_2))=\rr{m}(p_2,p_1)$ for all $(p_1,p_2)\in\bb{E}\times_X \bb{E}$. Such a metric is called \emph{equivariant}. With respect to an equivariant metric the involution $\Theta$ acts as an \virg{anti-unitary} map. The existence of an equivariant metric $\rr{m}$
follows directly from the existence of any Hermitian  metric $\rr{m}'$ by means of the average procedure 
$$
\rr{m}\big(p_1,p_2\big)\;:=\;\frac{1}{2}\Big[\rr{m}'\big(p_1,p_2\big)\;+\rr{m}'\big(\Theta(p_2),\Theta(p_1)\big)\;\Big]\;,\qquad\quad (p_1,p_2)\in\bb{E}\times_X \bb{E}\;.
$$
Also in this case an equivariant generalization of \cite[Chapter I, Theorem 8.8]{karoubi-97} assures that two equivariant metrics for the 
$\rr{Q}$-bundle $(\bb{E},\Theta)$ are related  by a $\rr{Q}$-isomorphism $f:\bb{E}\to\bb{E}$. Summarizing one has:

\begin{proposition}[Equivariant metric]
\label{rk:eq_metric}
Each $\rr{Q}$-bundle $(\bb{E},\Theta)$ 
over an involutive space $(X,\tau)$ 
such that $X$  verifies (at least) condition 0 of Definition \ref{ass:2}
 admits an equivariant Hermitian  metric which is essentially  unique up to $\rr{Q}$-isomorphisms.
\end{proposition}

The main implication  of Proposition \ref{rk:eq_metric} is that the problem of the classification of \virg{Quaternionic}  vector bundles coincides with the problem of the classification of \virg{Quaternionic}  vector bundles
endowed with an equivariant Hermitian  metric.
For this reason, we tacitly assume hereafter  that:
 \begin{assumption}\label{ass:metric}
Each $\rr{Q}$-bundle is endowed with an equivariant Hermitian  metric and vector bundle maps 
between $\rr{Q}$-bundles are assumed to be metric-preserving (\ie isometries).
\end{assumption}
%

\subsection{Homotopy classification of \virg{Quaternionic} vector bundles}\label{sect:Q-vb_homotopy}
The assumption that the base space $X$ is compact    allows us to extend  usual homotopy properties valid  for complex vector bundles  to the category of $\rr{Q}$-bundles. Given two involutive spaces  $(X_1,\tau_1)$ and $(X_2,\tau_2)$ 
we say that a continuous map $\varphi:X_1\to X_2$ is  \emph{equivariant} if and only if $\varphi\circ\tau_1=\tau_2\circ \varphi$. An \emph{equivariant homotopy} between  equivariant maps $\varphi_0$ and $\varphi_1$ is a continuous map $F:[0,1]\times X_1\to X_2$ such that  $\varphi_t(\cdot):=F(t,\cdot)$ is  equivariant for all $t\in[0,1]$. The set of the equivalence classes of equivariant homotopic maps between $(X_1,\tau_1)$ and $(X_2,\tau_2)$ will be denoted by $[X_1,X_2]_{\Z_2}$.
\begin{theorem}[Homotopy property]
\label{theo:equi_homot1}
Let $(X_1,\tau_1)$ and $(X_2,\tau_2)$ be two involutive spaces with  $X_1$ verifies (at least) condition 0 of Definition \ref{ass:2}.  Let $(\bb{E},\Theta)$ be a  $\rr{Q}$-bundle of rank  $m$ over $(X_2,\tau_2)$ and $F:[0,1]\times X_1\to X_2$ an equivariant homotopy between the equivariant maps $\varphi_0$ and $\varphi_1$. Then, the pullback bundles $\varphi_t^\ast\bb{E}\to X_1$ have an induced $\rr{Q}$-structure for all $t\in[0,1]$ and
$\varphi_0^\ast\bb{E}\simeq \varphi_1^\ast\bb{E}$ in  ${\rm Vec}_{\rr{Q}}^m(X_1,\tau_1)$.
\end{theorem}
\proof[
sketch of]
This theorem can be proved by a suitable equivariant generalization of \cite[Proposition 1.3]{atiyah-bott-64}. The only new point is the 
$\rr{Q}$-structure on $\varphi_t^\ast\bb{E}\to X_1$.
By definition  $\varphi_t^\ast\bb{E}|_x:= \{x\}\times\bb{E}|_{\varphi_t(x)}$ for all $x\in X_1$. Hence, the equivariance of $\varphi_t$  implies  $\varphi_t^\ast\bb{E}|_{\tau_1(x)}= \{\tau_1(x)\}\times \bb{E}|_{\tau_2(\varphi_t(x))}$. Then, the pullback construction induces 
 a $\rr{Q}$-structure on $\varphi_t^\ast\bb{E}$  by the map $\Theta_t^\ast$ that acts \emph{anti}-linearly between the fibers $\varphi_t^\ast\bb{E}|_x$ and $\varphi_t^\ast\bb{E}|_{\tau_1(x)}$ as the  product $\tau_1\times\Theta$. \qed

\medskip

Theorem \ref{theo:equi_homot1} is the starting point for a 
homotopy classification   of $\rr{Q}$-bundles. Complex and real vector bundles are classified by the set of homotopy equivalent maps from the base space to the complex or real Grassmann manifold, respectively \cite{milnor-stasheff-74}. A similar result holds true also for \virg{Real}-bundles provided that the Grassmann manifold is endowed with a suitable involution and the homotopy equivalence is restricted to equivariant maps \cite{edelson-71} (see also \cite[Section 4.4]{denittis-gomi-14}). In this section we provide a similar result for  
$\rr{Q}$-bundles of even rank.

\medskip

We recall that the \emph{Grassmann manifold} is defined as
$$
G_m(\C^\infty)\;:=\;\bigcup_{n=m}^{\infty}\;G_m(\C^n)\;,
$$
where,
for each pair $m\leqslant n$, $G_m(\C^n)\simeq\n{U}(n)/\big(\n{U}(m)\times \n{U}(n-m)\big)$ is the set of $m$-dimensional (complex) subspaces of $\C^n$. 
The space $G_m(\C^n)$ can be endowed with the structure of a finite CW-complex, making it into
a closed (\ie compact without boundary) manifold  of (real) dimension $2m(n- m)$.
 The inclusions $\C^n \subset \C^{n+1} \subset\ldots$  yield inclusions $G_m(\C^n)\subset G_m(\C^{n+1})\subset\ldots$ and one can equip $G_m(\C^\infty)$
with the direct limit topology. The resulting space $G_m(\C^\infty)$ has the structure of an infinite CW-complex which is, in particular, paracompact and path-connected. Following \cite{lawson-lima-filho-michelsohn-05} or \cite{biswas-huisman-hurtubise-10}, in case of even dimension one can endow $G_{2m}(\C^\infty)$ with an involution of quaternionic-type in the following way:
let
$\Sigma=\langle{\rm v}_1,{\rm v}_2,\ldots{\rm v}_{2m-1},{\rm v}_{2m}\rangle_\C$ be any
 $2m$-plane in $G_{2m}(\C^{2n})$  generated by the basis $\{{\rm v}_1,{\rm v}_2,\ldots{\rm v}_{2m-1},{\rm v}_{2m}\}$ and define $\rho({\Sigma})\in G_{2m}(\C^{2n})$ as  the $2m$-plane spanned by $\langle Q\overline{\rm v}_1,Q\overline{\rm v}_2,\ldots,Q\overline{\rm v}_{2m-1},Q\overline{\rm v}_{2m}\rangle_{\C}$ where $\overline{\rm v}_j$ is the complex conjugate of  ${\rm v}_j$
and  $Q$ is the $2n\times 2n$ matrix \eqref{eq:Q-mat}. Clearly, the definition of 
$\rho(\Sigma)$ does not depend on the choice of a particular basis and one can immediately check that the map $\rho:G_m(\C^n)\to G_m(\C^n)$  is an involution that makes the pair $(G_{2m}(\C^{2n}),\rho)$ into an involutive space. 
Since all the inclusions $G_{2m}(\C^{2n})\hookrightarrow G_{2m}(\C^{2(n+1)})\hookrightarrow\ldots$ are equivariant, the involution extends to 
 the infinite Grassmann manifold in such a way that $\hat{G}_{2m}(\C^\infty)\equiv(G_{2m}(\C^\infty),\varrho)$ becomes an involutive space. Let $\Sigma=\rho(\Sigma)$ be a fixed point of $\hat{G}_{2m}(\C^\infty)$. Since $\rho$ acts on vectors as a quaternionic structure one has $\Sigma\simeq \n{H}^m$ (\cf Remark \ref{rk:quat_str}).
More precisely, if $\Sigma$ is $\rho$-invariant we can find a basis of $\langle{\rm v}_1,{\rm v}_2,\ldots{\rm v}_{2m-1},{\rm v}_{2m}\rangle_\C$ made by quaternionic pairs $({\rm v}_{2k-1},{\rm v}_{2k}:=Q\overline{\rm v}_{2k-1})$, $k=1,\ldots,m$ which leads to
$\Sigma=\langle{\rm v}_1,{\rm v}_3,\ldots{\rm v}_{2m-3},{\rm v}_{2m-1}\rangle_{\n{H}}$ (one has ${\rm v}_{2k}={\rm j}{\rm v}_{2k-1}$ with respect to the left quaternionic multiplication). 
Let  $G_m(\n{H}^n)\simeq\n{S}p(n)/\big(\n{S}p(m)\times \n{S}p(n-m)\big)$ be  the set of 
$m$-dimensional quaternionic hyperplanes passing through the origin of
$\n{H}^n$. As for the complex case we can define the \emph{quatrnionic} Grassmann manifold as the inductive limit $G_m(\n{H}^\infty):=\bigcup_{n=m}^{\infty}\;G_m(\n{H}^n)$.
This space has the structure of an infinite CW-complex. In particular it is paracompact and path-connected.
If $\Sigma$ is a fixed point of $\hat{G}_{2m}(\C^\infty)$ under the involution $\varrho$ then the  map $G_{m}(\n{H}^\infty)\to \hat{G}_{2m}(\C^\infty)$ given  by
$\langle{\rm v}_1,\ldots,{\rm v}_m\rangle_\n{H}\mapsto \langle{\rm v}_1, Q\overline{\rm v}_1,\ldots,{\rm v}_{m}, Q\overline{\rm v}_m\rangle_\C$ is an embedding of 
$G_{m}(\n{H}^\infty)$ onto the fixed point set of $\hat{G}_{2m}(\C^\infty)$. This shows that $G_{m}(\n{H}^\infty)\simeq \hat{G}_{2m}(\C^\infty)^\rho$.

\medskip

Each manifold $G_{m}(\C^{n})$ is the base space of a {canonical} rank $m$ complex vector bundle $\pi:\bb{T}_{m}^n\to G_m(\C^n)$ where the total space $\bb{T}_m^n$ consists of all pairs $(\Sigma,{\rm v})$ with $\Sigma\in G_m(\C^n)$  and ${\rm v}$ any vector in $\Sigma$ and the bundle projection is
 $\pi(\Sigma,{\rm v})=\Sigma$. Now, when  $n$ tends to infinity, the same construction leads to 
 the \emph{tautological} $m$-plane bundle $\pi:\bb{T}_m^\infty\to G_m(\C^\infty)$. 
This vector bundle is the universal object which classifies  complex vector bundles in the sense that any rank $m$ complex vector bundle $\bb{E}\to X$ can be realized, up to isomorphisms, as the pullback of  $\bb{T}_m^\infty$ with respect to a \emph{classifying map} $\varphi : X \to G_m(\C^\infty)$, that is $\bb{E}\simeq \varphi^\ast \bb{T}_m^\infty$.
Since pullbacks of homotopic maps yield isomorphic vector bundles (\emph{homotopy property}), the isomorphism class of  $\bb{E}$ only depends on the homotopy class  of $\varphi$.  This leads to the fundamental result 
 $
{\rm Vec}^m_\C(X) \simeq[X , G_m(\C^\infty)]
$ where in the right-hand side  there is the set of homotopy equivalence classes of  maps between $X$ and $G_m(\C^\infty)$.
This classical result can be extended to the category of even rank \virg{Quaternionic} vector bundles provided that the total space $\bb{T}_{2m}^\infty$ is endowed with a $\rr{Q}$-structure compatible with the involution $\rho$ on the Grassmann manifold
$G_{2m}(\C^\infty)$. This can be done by means of the 
\emph{anti}-linear map $\Xi:\bb{T}_{2m}^\infty\to\bb{T}_{2m}^\infty$
   defined by $\Xi:(\Sigma,{\rm v})\mapsto \big(\rho(\Sigma),Q\overline{\rm v}\big)$. The relation 
$\pi\circ\Xi=\rho\circ{\pi}$ can be easily verified, therefore  $(\bb{T}_{2m}^\infty,\Xi)$ is a $\rr{Q}$-bundle over the involutive space  $\hat{G}_{2m}(\C^\infty)\equiv ({G}_{2m}(\C^\infty),\rho)$. This  is the  \emph{universal} object for the  homotopy classification of $\rr{Q}$-bundles:

\begin{theorem}[Homotopy classification]\label{theo:honotopy_class_Q}
Let $(X,\tau)$ be an involutive space and assume that $X$  verifies (at least) condition 0  of Definition \ref{ass:2}.
 Each rank $2m$ $\rr{Q}$-bundle $(\bb{E},\Theta)$ over $(X,\tau)$ can be obtained, up to isomorphisms, as a pullback $\bb{E}\simeq \varphi^\ast \bb{T}_{2m}^\infty$ with respect to a map $\varphi:X\to \hat{G}_{2m}(\C^\infty)$ which is equivariant, $\varphi\circ\tau=\rho\circ \varphi$. Moreover, the homotopy property implies that $(\bb{E},\Theta)$ depends only on the equivariant homotopy class of $\varphi$, \ie one has  the isomorphism
$$
{\rm Vec}^{2m}_\rr{Q}(X,\tau)\;\simeq\;[X,\hat{G}_{2m}(\C^\infty)]_{\Z_2}\;.
$$
\end{theorem}

\proof[
sketch of]
The proof of this theorem is a direct equivariant adaption of  standard arguments for the complex case (\cf \cite{milnor-stasheff-74,husemoller-94}) and it is not new in the literature (\eg \cite[Theorem 11.2]{lawson-lima-filho-michelsohn-05} or \cite[Section 4.2]{biswas-huisman-hurtubise-10}). Just for  sake of completeness let us comment that Theorem \ref{theo:equi_homot1} assures that each equivariant pullback of the tautological vector bundle $\bb{T}_{2m}^\infty$ provides a $\rr{Q}$-bundle and that homotopic equivalent pullbacks provide isomorphic $\rr{Q}$-bundles. Therefore, one has only to show that each  $\rr{Q}$-bundle can be realized via an equivariant pullback.
The easiest way to prove this point is to adapt equivariantly the construction in \cite[Chapter 3, Section 5]{husemoller-94}.
\qed

\begin{remark}\label{rk:iso_categ}{\upshape
If one consider the trivial involutive space $(X,{\rm Id}_X)$, then  equivariant maps $\varphi:X\to \hat{G}_{2m}(\C^\infty)$ are characterized  by $\varphi(x)=\rho(\varphi(x))$ for all $x\in X$.
Since the fixed point set of $G_{2m}(\C^\infty)$ is parameterized by $G_m(\n{H}^\infty)$ one has that $[X,\hat{G}_{2m}(\C^\infty)]_{\Z_2}\simeq[X,G_m(\n{H}^\infty)]$. 
However, $G_m(\n{H}^\infty)$ is the classifying space for quaternionic vector bundles \cite[Chapter 8, Theorem 6.1]{husemoller-94} hence
${\rm Vec}^{2m}_{\rr{Q}}(X,{\rm Id}_X)\simeq{\rm Vec}^m_{\n{H}}(X)$ in agreement with
Proposition \ref{prop:Qv1.0}.
}\hfill $\blacktriangleleft$
\end{remark}

\noindent
The homotopy classification provided by  Theorem \ref{theo:honotopy_class_Q} 
can be directly used to classify $\rr{Q}$-bundles over the involutive space $\tilde{\n{S}}^1\equiv(\n{S}^1,\tau)$.   
\begin{proposition}
\label{prop_(ii)}
$
{\rm Vec}^{2m}_{\rr{Q}}(\n{S}^1, \tau)\;=\;0\;.
$
\end{proposition}
\proof
In view of Theorem \ref{theo:honotopy_class_Q} it is enough to  show that $[\tilde{\n{S}}^1,\hat{G}_{2m}(\C^\infty)]_{\Z_2}$ reduces to the equivariant homotopy class  of the constant map.  This fact  is a consequence of a general result called \emph{$\Z_2$-homotopy reduction}
\cite[Lemma 4.26]{denittis-gomi-14} which is based on the $\Z_2$-skeleton decomposition of $\tilde{\n{S}}^1$.  The involutive space $\tilde{\n{S}}^1$ has two fixed cells of dimension $0$ and one free cell of dimension 1 (\cf \cite[Example 4.20]{denittis-gomi-14}). Moreover,  $\pi_0(G_{2m}(\C^\infty))\simeq \pi_1(G_{2m}(\C^\infty))\simeq 0$ and the identification  $G_{2m}(\C^\infty)^\rho\simeq G_m(\n{H}^\infty)$ also implies $\pi_0(G_{2m}(\C^\infty)^\rho)\simeq0$.
These data are sufficient for the application of the $\Z_2$-homotopy reduction Lemma.
\qed


\subsection{Stable range condition}
\label{ssec:stab_reng_equi_constr}

In topology,  spaces which are homotopy equivalent to CW-complexes are very important. A similar notion can be extended to $\Z_2$-spaces like
 $\tilde{\n{S}}^d\equiv(\n{S}^d,\tau)$ or $\tilde{\n{T}}^d\equiv(\n{T}^d,\tau)$. These  spaces have 
  the structure of a CW-complex with respect to a skeleton decomposition made by cells of various dimension that carry a $\Z_2$-action. 
Such $\Z_2$-cells can be only of two types: They are \emph{fixed} if the action of $\Z_2$ is trivial or are \emph{free} if they have no fixed points. We refer to \cite[Section 4.5]{denittis-gomi-14} for a precise definition of the notion of 
$\Z_2$-CW-complex (see also \cite{matumoto-71,allday-puppe-93}). Moreover, the  $\Z_2$-CW-complex structure of  
$\tilde{\n{S}}^d$ and $\tilde{\n{T}}^d$ is explicitly described in  \cite[Example 4.20 and Example 4.20]{denittis-gomi-14}, respectively. We point out that  this construction is modelled after the usual definition of CW-complex, replacing the \virg{point} by \virg{$\Z_2$-point}. For this reason (almost) all topological and homological properties of CW-complexes have their \virg{natural} counterparts in the equivariant setting.

\medskip

This notion of $\Z_2$-CW-complex  plays an important role 
in the proof of the  next result which is the equivariant generalization of  \cite[Chapter 2, Theorem 7.1]{husemoller-94} to the category of even rank $\rr{Q}$-bundles.
\begin{proposition}[existence of a global $\rr{Q}$-pair of section]
\label{prop:stab_ran_Q}
Let $(X,\tau)$ be an involutive space such that $X$ has a finite $\Z_2$-CW-complex decomposition with fixed cells only in dimension 0. 
Let us denote with $d$ the dimension of $X$ (\ie the maximal dimension of the cells in the decomposition of $X$). Let $(\bb{E},\Theta)$ be a $\rr{Q}$-vector bundle over $(X,\tau)$ with fiber of rank $2m$. If $d\leqslant 4m-3$ there exists a pair of sections $(s_1,s_2)\in\Gamma(\bb{E})$ which is a global $\rr{Q}$-pair in the sense of Definition \ref{def:Q_pair}.  
\end{proposition}

\proof

 The \emph{zero} section $s_0(x)=0\in\bb{E}_x$ for all
$x\in X$  is  $\Theta$-invariant since $\Theta$ is an anti-linear isomorphism in each fiber. Let $\bb{E}^{\times}\subset\bb{E}$ be the subbundle of nonzero vectors. The fibers $\bb{E}^{\times}_x$ are all isomorphic to $(\C^{2m})^\times:=\C^{2m}\setminus\{0\}$
which is a $2(2m-1)$-connected space (\ie the first $2(2m-1)$  homotopy groups  vanish identically). The anti-involution $\Theta$ endows $\bb{E}^{\times}$ with a $\rr{Q}$-structure over  $(X,\tau)$. A $\rr{Q}$-pair of sections $(s_1,s_2:=\tau_\Theta(s_1))$ of $\bb{E}^{\times}$ can be seen as an 
everywhere-nonzero (\ie global) $\rr{Q}$-pair for the $\rr{Q}$-vector bundle $(\bb{E},\Theta)$. We will show that if $d\leqslant 4m-3$ such a $\rr{Q}$-pair of sections of  $\bb{E}^{\times}$ always exists.

We prove the claim by induction on the dimension of the skeleton.
 This is the case for $X^0$ which is a finite collection of fixed points $\{x_j\}_{j=1,\ldots,N_0}$ and conjugated pairs $\{x_j, \tau(x_j)\}_{j=1,\ldots,\tilde{N}_0}$.
  In this case a global $\rr{Q}$-pair can be defined as in the proof of Proposition \ref{prop:loc_triv}. On the fixed points $x_j=\tau(x_j)$ one sets a pair of vectors
   $(s'_1(x_j),s'_2(x_j))\subset\bb{E}^{\times}_x$, with $s'_2(x_j):=\tau_\Theta(s'_1)(x_j)$. This pair is automatically independent.
 For each free pairs $\{x_j, \tau(x_j)\}$ one starts with a section $s'_1:=(s'_1(x_j), s'_1(\tau(x_j)))\in \bb{E}^{\times}_x\times \bb{E}^{\times}_{\tau(x)}$ and, as usual, one defines $s'_2:=\tau_\Theta(s'_1)$. In this way $(s_1',s_2')$ is a $\rr{Q}$-pair of sections for $\bb{E}^{\times}_x\cup \bb{E}^{\times}_{\tau(x)}$.
 Assume now that the claim is true for the $\Z_2$-CW-subcomplex $X^{j-1}$ of dimension $j-1$
 for all  $1 \leqslant j \leqslant d$. By the inductive hypothesis we have a $\rr{Q}$-pair of sections $(s'_1,s_2')$ of the restricted bundle $\bb{E}^{\times}|_{X^{j-1}}$.
Let $Y\subset X$ be a free $j$-cell of $X$ with equivariant attaching map $\phi:\Z_2\times \n{D}^j\to Y\subset X$.
The pullback bundle $\phi^\ast(\bb{E}^{\times})\to \Z_2\times \n{D}^j$ has a $\rr{Q}$-structure since $\phi$ is equivariant
and it is locally $\rr{Q}$-trivial.
The $\rr{Q}$-pair $(s_1',s_2')$ defines a $\rr{Q}$-pair $(\sigma'_1,\sigma'_2)$  on $\phi^\ast(\bb{E}^{\times})|_{\Z_2\times\partial \n{D}^j}$ by $\sigma'_j:=s'_j\circ\phi$.

Since $\n{D}^j$ is contractible we know from Theorem \ref{theo:equi_homot1}
that
 $\phi^\ast(\bb{E}^{\times})|_{\Z_2\times \n{D}^j}$ is $\rr{Q}$-isomorphic to $(\Z_2\times \n{D}^j)\times (\C^{2m})^\times$
 (endowed with the standard trivial anti-involution $\Theta_0$, even if the particular form of the involution is not important for the rest of the proof). Then, 
the $\rr{Q}$-pair $(\sigma'_1,\sigma'_2)$
 defined on $\Z_2\times\partial \n{D}^j$ can be  identified with a pair of independent  equivariant maps  $\Z_2\times\partial \n{D}^j\to (\C^{2m})^\times$.
Because $j - 1 \le 2(2m-1)$ and $\pi_{j-1}((\C^{2m})^\times) \simeq 0$, the restriction  of the map induced by $\sigma'_1$ to $\{ 1 \} \times \partial \n{D}^j \to (\C^{2m})^\times$  prolongs to a map $\{ 1 \} \times \n{D}^j \to (\C^{2m})^\times$. At this point  $\sigma'_2$ can be seen as a map $\Z_2\times\partial \n{D}^j\to (\C^{2m}/\langle\sigma'_1\rangle_\C)^\times\simeq(\C^{2m-1})^\times$. Since 
$j - 1 \le 2(2m-2)$ and $\pi_{j-1}((\C^{2m-1})^\times) \simeq 0$, also the restriction  of the map induced by $\sigma'_2$ to $\{ 1 \} \times \partial \n{D}^j \to (\C^{2m})^\times$  prolongs to a map $\{ 1 \} \times \n{D}^j \to (\C^{2m})^\times$.
Moreover, these two prolonged maps are independent by construction. Using the equivariant constraints $\sigma'_2(-1,x):=(\Theta_0\sigma_1)(1,x)$ and $\sigma'_1(-1,x):=-(\Theta_0\sigma_2)(1,x)$
one obtains a $\rr{Q}$-pair of maps $\Z_2\times \n{D}^j \to (\C^m)^\times$. 
This prolonged maps yields a $\rr{Q}$-pair of sections $(\sigma_1,\sigma_2)$ of $\phi^\ast(\bb{E}^{\times})$. Using the natural morphism $\hat{\phi}:\phi^\ast(\bb{E}^{\times})\to \bb{E}^{\times}$ over $\phi$ (defined by the pullback construction) we have a unique $\rr{Q}$-pair of sections $(s_1^Y,s_2^Y)$ of $\bb{E}^{\times}|_{\overline{Y}}$ defined by $\hat{\phi}\circ\sigma_j=s_j^Y\circ \phi$
such that $s_j^Y=s'_j$ on $X^{j-1}\cap \overline{Y}$. 
Now, one defines a global $\rr{Q}$-pair $(s_1,s_2)$ of $\bb{E}^{\times}|_{X^{j-1} \cup Y}$
 by the requirements that $s_j|_{X^{j-1}}\equiv s'_j$  and $s_j|_Y \equiv s_j^Y$ for the free $j$-cell $Y$.   By the weak topology property of CW-complex, the $s_j$'s are also continuous. This argument applies to every  other free $j$-cell, and the claim is true on $X^j$ and eventually on $X^d = X$.
  \qed

\begin{remark}{\upshape
In principle the condition that the involutive space $(X,\tau)$ must have fixed cells only in dimension 0 should be removed
since fibers over fixed points $x=\tau(x)$ are isomorphic to $\n{H}^m$
(\cf Proposition \ref{prop:Qv1.0}) and 
 $\bb{E}^{\times}_x\simeq\n{H}^m\setminus\{0\}$  is  $2(2m-1)$-connected as well.
However, for the aims of this work we do not need this kind of generalization.
}\hfill $\blacktriangleleft$
\end{remark}

\noindent
The next theorem provides the stable range decomposition for \virg{Quaterninic} vector bundles.

\begin{theorem}[Stable range]
\label{theo:stab_ran_Q}
Let $(X,\tau)$ be an involutive space such that $X$ has a finite $\Z_2$-CW-complex decomposition of dimension $d$ with fixed cells only in dimension 0. 
Each rank $2m$ $\rr{Q}$-vector bundle $(\bb{E},\Theta)$  over $(X,\tau)$ such that $d\leqslant 4m-3$
splits as
\beql{eq:stab_rank_Q}
\bb{E}\;\simeq\;\bb{E}_0\;\oplus\;(X\times\C^{2(m-\sigma)})
\eeq
where $\bb{E}_0$ is a $\rr{Q}$-vector bundle  over $(X,\tau)$, $X\times\C^{2(m-\sigma)}\to X$ is the trivial product $\rr{Q}$-bundle  over $(X,\tau)$  and $\sigma:=[\frac{d+2}{4}]$ (here $[x]$ denotes the integer part of $x\in\R$).
\end{theorem}
\proof
By Proposition \ref{prop:stab_ran_Q} there is a global $\rr{Q}$-pair of sections $(s_1,s_2)\subset\Gamma(\bb{E})$.
This sections determines a monomorphism $f:X\times\C^2\to \bb{E}$ given by $f(x,(a_1,a_2)):=a_1s_1(x)+a_2s_2(x)$. This monomorphism is equivariant, \ie
$f(\tau(x),(-\overline{a}_2,\overline{a}_1))=-\overline{a}_2s_1(\tau(x))+\overline{a}_1s_2(\tau(x))=\Theta(a_1s_1(x)+a_2s_2(x))$. Let $\bb{E}'$ be the \emph{cokernel} of $f$ in $\bb{E}$, namely
$\bb{E}'$ is the quotient of $\bb{E}$ by the  relation: $p\sim p'$  if $p$ and $p'$ are in the same fiber of $\bb{E}$ and if $p-p'\in{\rm Im}(f)$. The map $\bb{E}'\to X$ is a vector bundle of rank $2(m-1)$ (\cf \cite[Chapter 3, Corollary 8.3]{husemoller-94}) which inherits a $\rr{Q}$-structure from $\bb{E}$ and the equivariance of $f$.
Since $X$ is  compact, by \cite[Chapter 3, Theorem 9.6]{husemoller-94} there is an isomorphism of $\rr{Q}$-bundles between 
$\bb{E}$ and $\bb{E}'\oplus(X\times\C^2)$. If $d\leqslant 4(m-1)-3$ one can repeat the argument for $\bb{E}'$ and by iterating this procedure  one gets eventually \eqref{eq:stab_rank_Q}.
\qed

\begin{corollary}
\label{corol_(iii)}
Let $(X,\tau)$ be an involutive space such that $X$ has a finite $\Z_2$-CW-complex decomposition of dimension $2\leqslant d\leqslant 5$  with fixed cells only in dimension 0. 
Then 
$$
{\rm Vec}^{2m}_{\rr{Q}}(X, \tau)\;\simeq\; {\rm Vec}^2_{\rr{Q}}(X, \tau)\qquad\quad \forall\ m\in\N\;.
$$
\end{corollary}


\section{The FKMM-invariant}
\label{sec:FKMM-inv}

In an unpublished work 
M. Furuta, Y. Kametani, H. Matsue, and N. Minami
proposed  a topological invariant capable to classify \virg{Quaternionic} vector bundles over certain involutive spaces, provided that certain conditions are met \cite{furuta-kametani-matsue-minami-00}. 
Interestingly, this object was originally introduced to classify  $\rr{Q}$-bundles on $\tilde{\n{T}}^4$.
We present in this  section  a more general and slightly different definition for this invariant recognizing, of course, that the original ideas contained in \cite{furuta-kametani-matsue-minami-00} has been of inspiration to us.

\subsection{A short reminder of the equivariant Borel cohomology}
\label{subsec:borel_cohom}

The proper cohomology theory for the analysis of vector bundles  in the category of spaces with involution is the \emph{equivariant cohomolgy} introduced by  A.~Borel in \cite{borel-60}. This cohomology plays an important role for the classification of \virg{Real} vector bundles \cite{denittis-gomi-14} and we will show that it is also relevant for the study of \virg{Quaternionic} vector bundles. A short   self-consistent summary of this cohomology theory can be found in \cite[Section 5.1]{denittis-gomi-14}.
For an  introduction to the subject  we refer to
\cite[Chapter 3]{hsiang-75} and \cite[Chapter 1]{allday-puppe-93}.

\medskip

Since we need this tool we briefly recall the main steps of the
{Borel construction} for the equivariant cohomology. 
The \emph{homotopy quotient} of an involutive space   $(X,\tau)$ is the orbit space
$$
{X}_{\sim\tau}\;:=\;X\times\hat{\n{S}}^\infty /( \tau\times \vartheta)\;.
$$
Here $\vartheta$ is the \emph{antipodal map} on the infinite sphere $\n{S}^\infty$ 
(\cf \cite[Example 4.1]{denittis-gomi-14}) and $\hat{\n{S}}^\infty$ is used for the pair $(\n{S}^\infty,\vartheta)$.
The product space $X\times{\n{S}}^\infty$ (forgetting for a moment the $\Z_2$-action) has the \emph{same} homotopy type of $X$ 
since $\n{S}^\infty$ is contractible. Moreover, since $\vartheta$ is a free involution,  also the composed involution $\tau\times\vartheta$ is free, independently of $\tau$.
Let $\s{R}$ be any commutative ring (\eg, $\R,\Z,\Z_2,\ldots$). The \emph{equivariant cohomology} ring 
of $(X,\tau)$
with coefficients
in $\s{R}$ is defined as
$$
H^\bullet_{\Z_2}(X,\s{R})\;:=\; H^\bullet({X}_{\sim\tau},\s{R})\;.
$$
More precisely, each equivariant cohomology group $H^j_{\Z_2}(X,\s{R})$ is given by the
 singular cohomology group  $H^j({X}_{\sim\tau},\s{R})$ of the  homotopy quotient ${X}_{\sim\tau}$ with coefficients in $\s{R}$ and the ring structure is given, as usual, by the {cup product}.
As the coefficients of
the usual singular cohomology are generalized to {local coefficients} (see \eg \cite[Section 3.H]{hatcher-02} or
\cite[Section 5]{davis-kirk-01}), the coefficients of the Borel's equivariant cohomology are also
generalized to local coefficients. Given an involutive space $(X,\tau)$ let us consider the fundamental group $\pi_1({X}_{\sim\tau})$
and the associated  \emph{group ring} $\Z[\pi_1({X}_{\sim\tau})]$. Each module $\s{Z}$ over the group $\Z[\pi_1({X}_{\sim\tau})]$ is, by definition,
a \emph{local system} on $X_{\sim\tau}$.  Using this local system one defines, as usual, the equivariant cohomology with local coefficients in $\s{Z}$:
$$
H^\bullet_{\Z_2}(X,\s{Z})\;:=\; H^\bullet({X}_{\sim\tau},\s{Z})\;.
$$
We are particularly interested in modules $\s{Z}$ whose underlying groups are identifiable with $\Z$. 
For each involutive space  $(X,\tau)$, there always exists a particular family of local systems $\Z(m)$
labelled by $m\in\Z$. Here
 $\Z(m)\simeq X\times\Z$ denotes the $\Z_2$-equivariant local system on $(X,\tau)$  made equivariant  by the $\Z_2$-action $(x,l)\mapsto(\tau(x),(-1)^ml)$.
Because the module structure depends only on the parity of $m$, we consider only the $\Z_2$-modules ${\Z}(0)$ and ${\Z}(1)$. Since ${\Z}(0)$ corresponds to the case of the trivial action of $\pi_1(X_{\sim\tau})$ on $\Z$ one has $H^k_{\Z_2}(X,\Z(0))\simeq H^k_{\Z_2}(X,\Z)$ \cite[Section 5.2]{davis-kirk-01}.

\medskip

Let us recall two important group isomorphisms  involving the 
first two equivariant cohomology groups. Let $(X,\tau)$ be an involutive space, then
\begin{equation}\label{eq:iso:eq_cohom}
H^1_{\Z_2}(X,\Z(1))\;\simeq\;[X,\n{U}(1)]_{\Z_2}\;,\qquad\qquad H^2_{\Z_2}(X,\Z(1))\;\simeq\;{\rm Vec}_{\rr{R}}^1(X,\tau)\;.
\end{equation}
The first isomorphism \cite[Proposition A.2]{gomi-13} says that the first equivariant cohomology group is isomorphic to the set of $\Z_2$-homotopy classes of equivariant maps $\varphi:X\to\n{U}(1)$ where the involution on $\n{U}(1)$ is induced by the complex conjugation, \ie $\varphi(\tau(x))=\overline{\varphi(x)}$. The second isomorphism is due to B.~Kahn \cite{kahn-59} and 
expresses the equivalence between the Picard group of \virg{Real} line bundles (in the sense of \cite{atiyah-66,denittis-gomi-14}) over  $(X,\tau)$ and the second equivariant cohomology group of this space.
A more modern proof of this result can be found in \cite[Corollary A.5]{gomi-13}.

\medskip

The fixed point subset $X^\tau\subset X$ is closed and $\tau$-invariant and the inclusion $\imath:X^\tau\hookrightarrow X$ extends to an inclusion $\imath:X^\tau_{\sim\tau}\hookrightarrow X_{\sim\tau}$ of the respective homotopy quotients. The \emph{relative} equivariant cohomology can be defined as usual by the identification
$$
H^\bullet_{\Z_2}(X|X^\tau,\s{Z})\;:=\; H^\bullet({X}_{\sim\tau}|X^\tau_{\sim\tau},\s{Z})
$$
and one has a related long exact sequence in cohomology
$$
\ldots\;\longrightarrow\;H^j_{\Z_2}(X|X^\tau,\s{Z})\;\longrightarrow\;H^j_{\Z_2}(X,\s{Z})\;\stackrel{r}{\longrightarrow}\;H^j_{\Z_2}(X^\tau,\s{Z})\;\longrightarrow\;H^{j+1}_{\Z_2}(X|X^\tau,\s{Z})\;\longrightarrow\;\ldots\;
$$
where the map $r:=\imath^*$ restricts  cochains on $X$ to  cochains on $X^\tau$. The $j$-th \emph{cokernel} of $r$ is by definition
$$
{\rm Coker}^j(X|X^\tau,\s{Z})\;:=\;H^j_{\Z_2}(X^\tau,\s{Z})\;/\;r\big(H^j_{\Z_2}(X,\s{Z})\big)\;.
$$

\begin{lemma}\label{lemma:coker}
Let $(X,\tau)$ be an involutive space such that $X^\tau\neq \emptyset$. Then
\beql{eq:coker_iso1}
[X^\tau,\n{U}(1)]_{\Z_2}\;/\;[X,\n{U}(1)]_{\Z_2}\;\simeq\; {\rm Coker}^1(X|X^\tau,\n{Z}(1))
\eeq
where the group action of $[f]\in [X,\n{U}(1)]_{\Z_2}$ on $[g]\in [X^\tau,\n{U}(1)]_{\Z_2}$ is given by multiplication and restriction, namely  $[f]:[g]\mapsto [f|_{X^\tau}g]$.
Moreover, if  $H^2_{\Z_2}(X,\Z(1))=0$  then
$$
[X^\tau,\n{U}(1)]_{\Z_2}\;/\;[X,\n{U}(1)]_{\Z_2}\;\simeq\; H^2_{\Z_2}(X|X^\tau,\n{Z}(1))\;.
$$
\end{lemma}
\proof[
sketch of] The first isomorphism \eqref{eq:coker_iso1} is a direct consequence of the isomorphism
\eqref{eq:iso:eq_cohom} proved in \cite[Proposition A.2]{gomi-13}. Under the extra assumption  $H^2_{\Z_2}(X,\Z(1))=0$ the long exact sequence in cohomology provides the exact sequence
$$
H^1_{\Z_2}(X,\n{Z}(1))\;\stackrel{r}{\longrightarrow}\;H^1_{\Z_2}(X^\tau,\n{Z}(1))\;\stackrel{s}{\longrightarrow}\;H^{2}_{\Z_2}(X|X^\tau,\n{Z}(1))\;\longrightarrow\;0\;.
$$ 
Since, ${\rm Ker}(s)={\rm Im}(r)$ and ${\rm Im}(s)=H^{2}_{\Z_2}(X|X^\tau,\n{Z}(1))$ one deduces from the homomorphism
theorem 
$$
{\rm Im}(s)\;\simeq\;H^1_{\Z_2}(X^\tau,\n{Z}(1))\;/\;{\rm Ker}(s)
$$
the required isomorphism ${\rm Coker}^1(X|X^\tau,\n{Z}(1))\simeq H^{2}_{\Z_2}(X|X^\tau,\n{Z}(1))$.
\qed

\medskip

In many situations of interest the cokernel in \eqref{eq:coker_iso1} has a very simple form. In the  cases of the TR-spheres $\tilde{\n{S}}^d$ described in Definition \ref{def:AII_sys}
one has that
\beql{eq:coker_TR_sphe}
{\rm Coker}^1\big(\tilde{\n{S}}^d|(\tilde{\n{S}}^d)^\tau,\n{Z}(1)\big)\;\simeq\;H^{2}_{\Z_2}\big(\tilde{\n{S}}^d|(\tilde{\n{S}}^d)^\tau,\n{Z}(1)\big)\;\simeq\;\Z_2 \qquad\quad\ \ \forall\ d\geqslant 2\;.
\eeq
The first isomorphism is a consequence of $H^{2}_{\Z_2}(\tilde{\n{S}}^d,\n{Z}(1))=0$ (\cf \cite[eq. 5.26]{denittis-gomi-14}) while the second isomorphism is justified in Proposition \ref{prop:cond_coker}. For the  TR-tori $\tilde{\n{T}}^d$ the evaluation of the cokernel depends on the dimension according to the formula
\beql{eq:coker_TR_tori}
{\rm Coker}^1\big(\tilde{\n{T}}^d|(\tilde{\n{T}}^d)^\tau,\n{Z}(1)\big)\;\simeq\;H^{2}_{\Z_2}\big(\tilde{\n{T}}^d|(\tilde{\n{T}}^d)^\tau,\n{Z}(1)\big)\;\simeq\;\Z_2^{2^d-(d+1)} \qquad\quad\ \ \forall\ d\geqslant 1\;
\eeq
where we used the convention $\Z_2^0\equiv\{0\}$.
Again the first isomorphism follows from $H^{2}_{\Z_2}(\tilde{\n{T}}^d,\n{Z}(1))=0$ (\cf \cite[eq. 5.19]{denittis-gomi-14}) while the second isomorphism is proved in  Proposition \ref{prop:coker_tori}.

\subsection{The determinant construction}
\label{subsec:det_construct}

In order to define the {FKMM-invariant} we need the notion of \emph{determinat line bundle} associated with a (complex) vector bundle. Let $\bb{V}$ a complex vector space of dimension $n$. The \emph{determinant} of  $\bb{V}$ is by definition ${\rm det}(\bb{V}):=\bigwedge^n\bb{V}$ where the symbol  $\bigwedge^n$ denotes the top exterior power of $\bb{V}$ (\ie the skew-symmetrized $n$-th tensor power of $\bb{V}$). This is a complex vector space of dimension one.
If $\bb{W}$ is a second vector space of same dimension $n$
and $T:\bb{V}\to \bb{W}$ is a linear map then there is a naturally associated map ${\rm det}(T):{\rm det}(\bb{V})\to{\rm det}(\bb{W})$ which in the special case $\bb{V}= \bb{W}$
 coincides with the multiplication by the determinant of the endomorphism $T$.
 This determinant construction
is a functor from the category of vector spaces  to itself
and by a standard argument \cite[Chapter 5, Section 6]{husemoller-94} it induces a functor on the category of complex vector bundles over an arbitrary space $X$. Given a rank $n$ complex vector bundle $\bb{E}\to X$, one defined the associated determinant line bundle ${\rm det}(\bb{E})\to X$ as the rank 1 complex vector bundle with fiber description 
\beql{eq:fib_descr}
 {\rm det}(\bb{E})_x\;=\; {\rm det}(\bb{E}_x)\qquad\quad x\in X \;.
\eeq
 If $\{s_1,\ldots,s_n\}$ is a local trivializing frame for $\bb{E}$ over the open set $\f{U}\subset X$ then ${\rm det}(\bb{E})$ is trivialized over the same open set $\f{U}$ by the section $s_1\wedge\ldots\wedge s_n$. For each map $\varphi:X\to Y$ one has the isomorphism ${\rm det}(\varphi^*(\bb{E}))\simeq \varphi^*({\rm det}(\bb{E}))$ which is a special case of the compatibility between pullback and tensor product operations.
Finally, if $\bb{E}=\bb{E}_1\oplus\bb{E}_2$ in the sense of the
Whitney sum then ${\rm det}(\bb{E})={\rm det}(\bb{E}_1)\otimes {\rm det}(\bb{E}_2)$.

\medskip

If $(\bb{E},\Theta)$ is a rank $n$ $\rr{Q}$-bundle over
the involutive space $(X,\tau)$ then the associated determinant line bundle ${\rm det}(\bb{E})$ inherits an involutive structure given by the map ${\rm det}(\Theta)$ which acts \emph{anti}-linearly between the fibers ${\rm det}(\bb{E})_x$ and ${\rm det}(\bb{E})_{\tau(x)}$ according to ${\rm det}(\Theta)(p_1\wedge\ldots\wedge p_n)= \Theta(p_1)\wedge\ldots\wedge \Theta(p_n)$. Clearly ${\rm det}(\Theta)^2$ is a fiber preserving map which coincides with the multiplication by $(-1)^n$. Hence:

\begin{lemma}\label{lemma:R_Q_det_bun}
Let $(\bb{E},\Theta)$ be a rank $n$ $\rr{Q}$-bundle over
 $(X,\tau)$ and $({\rm det}(\bb{E}),{\rm det}(\Theta))$ the associated determinant line bundle endowed with  the involutive structure ${\rm det}(\Theta)$.

\begin{enumerate}
\item[(i)] If $n=2m$ then  $({\rm det}(\bb{E}),{\rm det}(\Theta))$ is a \virg{Real} line bundle over $(X,\tau)$;\vspace{1.3mm}
\item[(ii)] If $n=2m+1$ then  $({\rm det}(\bb{E}),{\rm det}(\Theta))$ is a \virg{Quaternionic} line bundle over $(X,\tau)$.
\end{enumerate}
\end{lemma}

\noindent
We recall once more that the adjective   \virg{Real} is used in the sense of \cite{atiyah-66,denittis-gomi-14}.

\begin{remark}[Metric, line bundle, circle bundle]\label{rk:circ_bund}
{\upshape
Let $(\bb{E},\Theta)$ be a $\rr{Q}$-bundle over
 $(X,\tau)$ of even degree. According to Assumption \ref{ass:metric}, $\bb{E}$  carries an equivariant Hermitian metric $\rr{m}$ that fixes a unique Hermitian metric $\rr{m}_{\rm det}$ on ${\rm det}(\bb{E})$ which is equivariant with respect to the induced $\rr{R}$-structure
${\rm det}(\Theta)$. More explicitly, if $(p_i,q_i)\in \bb{E}|_x\times\bb{E}|_x$, $i=1,\ldots,2m$ then,
$$
\rr{m}_{\rm det}(p_1\wedge\ldots\wedge p_{2m},q_1\wedge\ldots\wedge q_{2m})\;:=\;\prod_{i=1}^{2m}\rr{m}(p_i,q_i)\;.
$$
The $\rr{R}$-line bundle $({\rm det}(\bb{E}),{\rm det}(\Theta))$ endowed with the equivariant Hermitian metric $\rr{m}_{\rm det}$  is $\rr{R}$-trivial if and only if there exists an isometric $\rr{R}$-isomorphism with $X\times\C$ , or equivalently, if and only if
 there exists a global $\rr{R}$-section $s:X\to {\rm det}(\bb{E})$ of unit length (\cf \cite[Theorem 4.8]{denittis-gomi-14}). Let us introduce the \emph{circle bundle} $\n{S}({\rm det}(\bb{E})):=\{p\in{\rm det}(\bb{E})\ |\  \rr{m}_{\rm det}(p,p)=1\}$. Then, the  
$\rr{R}$-triviality of ${\rm det}(\bb{E})$ is equivalent to the existence of an $\rr{R}$-section for the circle bundle $\n{S}({\rm det}(\bb{E}))\to X$.
}\hfill $\blacktriangleleft$
\end{remark}

\begin{corollary}\label{cor:triv_det_01}
Let $(\bb{E},\Theta)$ be a rank $n$ $\rr{Q}$-bundle over
$(X,\tau)$ and $({\rm det}(\bb{E}),{\rm det}(\Theta))$ the associated determinant line bundle endowed with  the involutive structure ${\rm det}(\Theta)$. If $X^\tau\neq\emptyset$ and $H^2_{\Z_2}(X,\Z(1))=0$ then there exists a global trivializing map
$$
h_{\rm det}\;:\;{\rm det}(\bb{E})\;\longrightarrow\; X\;\times \C 
$$
which is equivariant in the sense that $h_{\rm det}\circ {\rm det}(\Theta)=\Upsilon_0\circ h_{\rm det}$ where $\Upsilon_0$ is the standard   $\rr{R}$-structure on the product bundle 
$X\;\times \C$ defined by $\Upsilon_0(x,{\rm v}):=(\tau(x),\overline{{\rm v}})$. \end{corollary}
\proof
Since  $X^\tau\neq\emptyset$
 Proposition \ref{prop:Qv0} assures that $(\bb{E},\Theta)$ has even rank and Lemma \ref{lemma:R_Q_det_bun} implies that $({\rm det}(\bb{E}),{\rm det}(\Theta))$ is an 
$\rr{R}$-line bundle. Therefore, the condition $H^2_{\Z_2}(X,\Z(1))=0$  implies the $\rr{R}$-triviality of $({\rm det}(\bb{E}),{\rm det}(\Theta))$
as showed by \eqref{eq:iso:eq_cohom}. This means the existence of a global $\rr{R}$-trivialization $h_{\rm det}$. \qed

\medskip

\noindent
The restriction of the map $h_{\rm det}$ to $X^\tau$ provides a trivialization for the restricted line bundle  ${\rm det}(\bb{E})|_{X^\tau}$ and, as a consequence of the fiber description \eqref{eq:fib_descr}, one has
\beql{eq:triv_1}
h_{\rm det}|_{X^\tau}\;:\; {\rm det}(\bb{E}|_{X^\tau})\;\longrightarrow\; X^\tau\;\times\; \C\;. 
\eeq
\begin{lemma}
\label{lemma:R_Q_det_bun2}
Let $(\bb{E},\Theta)$ be a $\rr{Q}$-bundle over a space $X$ with trivial involution $\tau={\rm Id}_X$. Then, the associated determinant line bundle ${\rm det}(\bb{E})$ is a trivial $\rr{R}$-bundle which admits a unique  \emph{canonical} trivializing map
\beql{eq:triv_001}
{\rm det}_{X}\;:\;{\rm det}(\bb{E})\;\longrightarrow\; X\;\times \C 
\eeq
such that ${\rm det}_{X}\circ {\rm det}(\Theta)=\Upsilon_0\circ {\rm det}_{X}$. 
Moreover, the map ${\rm det}_{X}$ can be chosen to be metric-preserving and  it leads to a unique \emph{canonical} $\rr{R}$-section $s_{X}:X\to\n{S}({\rm det}(\bb{E}))$ defined by
$$
s_{X}(x)\;:=\;{\rm det}_{X}^{-1}(x,1)\;,\qquad\quad\forall x\in X\;.
$$
\end{lemma}
\proof
Since  $X$ has trivial involution  $\bb{E}$ has even rank (Proposition \ref{prop:Qv0}) and  $({\rm det}(\bb{E}),{\rm det}(\Theta))$ is an 
$\rr{R}$-line bundle (Lemma \ref{lemma:R_Q_det_bun}). Let $\{\f{U}_\alpha\}$ be a cover of $X$ associated with a system of local trivializations $h_\alpha:\bb{E}|_{\f{U}_\alpha}\to \f{U}_\alpha\times\C^{2m}$ such that $h_\alpha\circ\Theta=\Theta_0\circ h_\alpha$
where $\Theta_0(x,{\rm v})=(x,Q\overline{\rm v})$  is the standard trivial $\rr{Q}$-structure  on the   product bundle $\f{U}_\alpha\times\C^{2m}$ (the matrix $Q$ is defined  in \eqref{eq:Q-mat}).
Such a trivialization exists in view of Proposition \ref{prop:loc_triv}. 
Associated with each $h_\alpha$ there is a local trivialization ${\rm det}(h_{\alpha}):{\rm det}(\bb{E})|_{\f{U}_\alpha}\to \f{U}_\alpha\times \C$ for the restricted  determinant line bundle ${\rm det}(\bb{E}|_{\f{U}_\alpha})={\rm det}(\bb{E})|_{\f{U}_\alpha}$.
The map $h_\alpha$ is usually not unique: a choice for $h_\alpha$ is equivalent to a choice of a 
global $\rr{Q}$-frame (\cf Definition \ref{def:Q_pair}) $\{s^{(\alpha)}_1,s^{(\alpha)}_2,s^{(\alpha)}_3,s^{(\alpha)}_4\ldots,s^{(\alpha)}_{2m-1},s^{(\alpha)}_{2m}\}\subset\Gamma(\bb{E}|_{\f{U}_\alpha})$  such that
$
s^{(\alpha)}_{j}(x)\;:=\;h_{\alpha}^{-1}(x,{\rm e}_j)
$.
Different $\rr{Q}$-frames lead to different trivializations  $h_{\alpha}$ and each pair of $\rr{Q}$-frames is related by a gauge transformation $f_\alpha$. If one assumes that $\bb{E}$ is endowed with an invariant Hermitian metric  (\cf Proposition\ref{rk:eq_metric}) each $\rr{Q}$-frame can be chosen to be orthonormal and so $f_\alpha:\f{U}_\alpha\to\n{U}(2m)$. 
The $\rr{Q}$-structure and the fact that $X$ has trivial involution imply that $f_\alpha(x)=-Q\overline{f_\alpha(x)}Q$ for all $x\in \f{U}_\alpha$, namely $f_\alpha:\f{U}_\alpha\to\n{U}(2m)^\mu\subset {\rm S}\n{U}(2m)$ (\cf Remark \ref{rk:quat_str}). 
The trivialization ${\rm det}(h_{\alpha})$ is uniquely specified by the \virg{Real} section $s^{(\alpha)}:=s^{(\alpha)}_1\wedge s^{(\alpha)}_2\wedge\ldots\wedge s^{(\alpha)}_{2m-1}\wedge s^{(\alpha)}_{2m}$. If $h_{\alpha}$ and $h_{\alpha}'$ are two different trivializations for $\bb{E}|_{\f{U}_\alpha}$ related by the gauge transformation $f_\alpha$ then 
${\rm det}(h_{\alpha})$ and ${\rm det}(h'_{\alpha})$ are related by ${\rm det}(f_\alpha)=1$, namely 
${\rm det}(h_{\alpha})={\rm det}(h'_{\alpha})$. In this sense the local trivializations of ${\rm det}(\bb{E})$ are canonical and we write ${\rm det}_{\f{U}_\alpha}:{\rm det}(\bb{E})|_{\f{U}_\alpha}\to \f{U}_\alpha\times \C$ for the unique trivializing map.

The topology of the vector bundle $\bb{E}$ is uniquely determined
by the set of \emph{transition functions} $g_{\alpha\beta}:\f{U}_\alpha\cap\f{U}_\beta\to\n{U}({2m})$ associated with the maps
$h_\alpha\circ h_\beta^{-1}$ (which are well defined on $\f{U}_\alpha\cap\f{U}_\beta$). The equivariance of the maps $h_\alpha$ implies $g_{\alpha\beta}:\f{U}_\alpha\cap\f{U}_\beta\to\n{U}({2m})^\mu$. The determinant line bundle ${\rm det}(\bb{E})$ is completely specified by the set of transition functions ${\rm det}(g_{\alpha\beta}):\f{U}_\alpha\cap\f{U}_\beta\to\n{U}(1)$. Since ${\rm det}(g_{\alpha\beta})=1$ for all $\alpha$ and $\beta$ the local canonical  trivializations ${\rm det}_{\f{U}_\alpha}
$ glue together to give rise to the unique  global canonical trivialization \eqref{eq:triv_001}.
\qed

\medskip
\noindent
If $X^\tau\neq \emptyset$   the restricted vector bundle $\bb{E}|_{X^\tau}\to{X^\tau}$ can be seen as a $\rr{Q}$-bundle over a space with trivial involution and Lemma \ref{lemma:R_Q_det_bun2} provides the canonical trivialization
\beql{eq:triv_2}
{\rm det}_{X^\tau}\;:\;{\rm det}(\bb{E}|_{X^\tau})\;\longrightarrow\; X^\tau\;\times \C 
\eeq
for the restricted determinant line bundle ${\rm det}(\bb{E}|_{X^\tau})$ which is \virg{a priori} different from \eqref{eq:triv_1}.

\subsection{Construction of the FKMM-invariant}
\label{subsec:constr-FKMM_inv}
In this section we construct the  \emph{FKMM-invariant}
associated to a  \virg{Quaternionic} vector bundle $(\bb{E},\Theta)$ over an involutive space $(X,\tau)$.
Although  a more general approach is possible \cite{denittis-gomi-14-gen} (\cf also Remark \ref{rk:gen_FKMM_inv})  we decided, for pedagogical reasons, to present here
a construction which is specific for a particular (albeit sufficiently large) class
of $\rr{Q}$-bundles.
\begin{definition}[$\rr{Q}$-bundles of FKMM-type]
\label{def:FKMM-type}
A $\rr{Q}$-bundle $(\bb{E},\Theta)$ over the involutive space $(X,\tau)$ is  of \emph{FKMM-type} if $X^\tau\neq\emptyset$ and if the associated \virg{Real} determinant line bundle $({\rm det}(\bb{E}),{\rm det}(\Theta))$ is $\rr{R}$-trivial. The property to be of FKMM-type is an isomorphism invariant and we use the notation
$$
{\rm Vec}_{{\rm FKMM}}^{2m}(X,\tau)\;\subseteq {\rm Vec}_{\rr{Q}}^{2m}(X,\tau)
$$
for the set of equivalence classes of FKMM $\rr{Q}$-bundles of rank $2m$.
\end{definition}

\noindent
For certain  involutive spaces $(X,\tau)$  all possible $\rr{Q}$-bundles are of FKMM-type.

\begin{proposition}
\label{prop:FKMM_type_space}
Let $(X,\tau)$ be an FKMM-space in the sense of Definition \ref{ass:2}. Then
$$
{\rm Vec}_{{\rm FKMM}}^{2m}(X,\tau)\;= {\rm Vec}_{\rr{Q}}^{2m}(X,\tau)\qquad\quad \forall\ \ m\in\N\;.
$$
\end{proposition}
\proof
Since $X^\tau\neq\emptyset$ the admissible $\rr{Q}$-bundles have even rank (\cf Proposition \ref{prop:Qv0}). This implies that the associated determinant line bundles carry  an $\rr{R}$-structure (\cf Lemma \ref{lemma:R_Q_det_bun}). Finally, the condition $H^2_{\Z_2}(X,\Z(1))=0$ and the (second) isomorphism in \eqref{eq:iso:eq_cohom} assure the $\rr{R}$-triviality of  each $\rr{R}$-line bundle over $(X,\tau)$. 
\qed

\medskip

Let $(\bb{E},\Theta)$ be a $\rr{Q}$-bundle of FKMM-type over the involutive space $(X,\tau)$ and consider the restricted determinant line bundle ${\rm det}(\bb{E}|_{X^\tau})\to X^\tau$. This line bundle is $\rr{R}$-trivial and according to Lemma \ref{lemma:R_Q_det_bun2} it admits a canonical trivialization \eqref{eq:triv_2}. On the other hand, the full determinant line bundle ${\rm det}(\bb{E})\to X$ is $\rr{R}$-trivial 
by assumption and so (as in Corollary \ref{cor:triv_det_01})
there exists a global trivialization $h_{\rm det}:{\rm det}(\bb{E})\to X\times \C$ which restricts to a trivialization for ${\rm det}(\bb{E}|_{X^\tau})$ as in \eqref{eq:triv_1}.
 If one fixes an equivariant Hermitian metric on $\bb{E}$
 the maps ${\rm det}_{X^\tau}$ and  $h_{\rm det}|_{X^\tau}$  can be chosen to be isometries with respect to the standard  Hermitian metric on the product bundle.
This implies that the difference 
$$
h_{\rm det}|_{X^\tau}\circ {\rm det}_{X^\tau}^{-1}\;:\: X^\tau\;\times\; \C\;\longrightarrow\; X^\tau\;\times\; \C
$$
identifies a  map 
$$
\omega_{\bb E}\;:\;X^\tau\;\longrightarrow\;\n{U}(1)
$$
such that $(h_{\rm det}|_{X^\tau}\circ {\rm det}_{X^\tau}^{-1})(x,\lambda)=(x,\omega_{\bb E}(x)\lambda)$
for all $(x,\lambda)\in X^\tau \times  \C$.
The equivariance property $(h_{\rm det}|_{X^\tau}\circ {\rm det}_{X^\tau}^{-1})\circ\Upsilon_0 =\Upsilon_0\circ(h_{\rm det}|_{X^\tau}\circ {\rm det}_{X^\tau}^{-1})$ implies that $\omega_{\bb E}$ is  equivariant with respect to the involution on $\n{U}(1)$ given by the complex conjugation, \ie 
$\omega_{\bb E}(\tau(x))=\overline{\omega_{\bb E}(x)}$. Since $\omega_{\bb E}$ is defined on the fixed point set $X^\tau$ and the invariant subset of $\n{U}(1)$  is $\{-1,+1\}$ one has that
\beql{eq:fkmm_inv1}
\omega_{\bb E}\;:\;X^\tau\;\longrightarrow\;\{-1,+1\}\simeq\Z_2\;,
\eeq
namely $\omega_{\bb E}\in{\rm Map}(X^\tau,\Z_2)\simeq [X^\tau,\n{U}(1)]_{\Z_2}$. Considering that the canonical trivialization ${\rm det}_{X^\tau}$ is unique,
the construction of  $\omega_{\bb E}$ 
only depends on the choice of $h_{\rm det}$. This freedom is equivalent to the choice of a global equivariant gauge transform $f:X\to \n{U}(1)$ that affects $\omega_{\bb E}$ by multiplication and restriction (as in Lemma \ref{lemma:coker}). 
Moreover, only the homotopy class $[f]\in[X,\n{U}(1)]_{\Z_2}$ is relevant.

\begin{definition}[FKMM-invariant, {\cite{furuta-kametani-matsue-minami-00}}]
\label{def:FKMM-invariant}
To each
  $\rr{Q}$-bundle $(\bb{E},\Theta)$  of {FKMM-type} is associated the class
$$
\kappa(\bb{E})\;:=\;[\omega_{\bb E}]\;\in\;[X^\tau,\n{U}(1)]_{\Z_2}/[X,\n{U}(1)]_{\Z_2}\;.
$$
We say that $\kappa(\bb{E})$ is the \emph{FKMM-invariant} of the  $\rr{Q}$-bundle $(\bb{E},\Theta)$ .
\end{definition}

\begin{remark}\label{rk:FKMM_section}{\upshape
The FKMM-invariant can be introduced also from the point of view  of sections. Since ${\rm det}(\bb{E})$ is $\rr{R}$-trivial by assumption, the set of  $\rr{R}$-sections of the circle bundle $\n{S}({\rm det}(\bb{E}))$
is non-empty and any two of such
$\rr{R}$-sections  are related by the multiplication by a $\Z_2$-equivariant map $u:X\to\n{U}(1)$.
Each $\rr{R}$-section $t:X\to \n{S}({\rm det}(\bb{E}))$ restricts to a
section $t|_{X^\tau}$
 of  
$ \n{S}({\rm det}(\bb{E}|_{X^\tau}))$. We can compare this section with 
the unique \emph{canonical} section $s_{X^\tau}$ constructed in Lemma \ref{lemma:R_Q_det_bun2}. The difference between  them
$t|_{X^\tau}=\omega_{\bb{E}}\cdot s_{X^\tau}$,
is specified by a $\Z_2$-equivariant map
$\omega_{\bb{E}}:X^\tau\to\n{U}(1)$ that coincides with the one introduced in Definition \ref{def:FKMM-invariant}.
 This equivalent description helps us to understand the meaning of the FKMM-invariant. Indeed, each equivariant section  of the \virg{Real} circle bundle $\n{S}({\rm det}(\bb{E}))$ defines, by restriction, an equivariant section over $X^\tau$. On the other hand, according to the topology of $X^\tau$ (for instance when there are several disconnected components) there may exist equivariant sections over $X^\tau$  that are not obtainable as the restriction of global equivariant sections over $X$. In a certain sense, it is exactly this redundancy which is measured by the FKMM-invariant.
 }
 \hfill $\blacktriangleleft$
\end{remark}

\begin{remark}[A more general definition of the FKMM-invariant]\label{rk:gen_FKMM_inv}{\upshape
The-FKMM invariant can be formulated in a more general
setting.  The key observation is that the cohomology group $H^2_{\Z_2}(X|X^\tau, \Z(1))$ can be realized as
isomorphism classes of pairs consisting of a \virg{Real} line bundles on $(X,\tau)$ and
a nowhere vanishing section on $X^\tau$. Then, in view of Lemma \ref{lemma:R_Q_det_bun2},
a generalized version of the FKMM-invariant can be defined for \emph{every} $\rr{Q}$-bundle $(\bb{E},\Theta)$
on \emph{every} involutive space $(X,\tau)$ as the element in $H^2_{\Z_2}(X|X^\tau, \Z(1))$
represented by the pair $({\rm det} (E), s_{X^\tau})$. 
The details of this construction will be given in a
future work \cite{denittis-gomi-15}.
 }
 \hfill $\blacktriangleleft$
\end{remark}

\noindent
The main properties of the FKMM-invariant are listed in the following theorem.

\begin{theorem}
\label{theo:FKMM_propert}
The FKMM-invariant is  well defined in the sense that if 
$(X,\tau)$ is an involutive space such that $X^\tau\neq\emptyset$ and $(\bb{E}_1,\Theta_1)$ and $(\bb{E}_2,\Theta_2)$ are two isomorphic $\rr{Q}$-bundles of FKMM-type over $(X,\tau)$
then $\kappa(\bb{E}_1)=\kappa(\bb{E}_2)$. Moreover:
\begin{enumerate}
\vspace{1.3mm}
\item[(i)] The FKMM-invariant is \emph{natural}, meaning that if $(\bb{E},\Theta)$  is a $\rr{Q}$-bundle of  FKMM-type over $(X,\tau_X)$ 
and
$f:(Y,\tau_Y)\to (X,\tau_X)$ is an equivariant map such that $f(Y)\cap X^{\tau_X}\neq\emptyset$
then $\kappa(f^\ast(\bb{E}))=f^\ast(\kappa(\bb{E}))$;

\vspace{1.3mm}
\item[(ii)] If $(\bb{E},\Theta)$ is $\rr{Q}$-trivial then $\kappa(\bb{E})=+1$;
\vspace{1.3mm}
\item[(iii)] If $(\bb{E}_1,\Theta_1)$ and $(\bb{E}_2,\Theta_2)$ are  two  $\rr{Q}$-bundles of {FKMM-type} over $(X,\tau)$ then 
$$\kappa(\bb{E}_1\oplus\bb{E}_2)\;=\;\kappa(\bb{E}_1)\cdot\kappa(\bb{E}_2)\;.$$
\end{enumerate}
Finally,  if $(X,\tau)$ is a FKMM-space then

\begin{enumerate}
\vspace{1.3mm}
\item[(iv)] $
\kappa(\bb{E})\;\in\;H^{2}_{\Z_2}(X|X^\tau,\n{Z}(1))\;.
$
\end{enumerate}
\end{theorem}
\proof
If $(\bb{E}_1,\Theta_1)\simeq(\bb{E}_2,\Theta_2)$
as $\rr{Q}$-bundles then 
the argument of the proof of Lemma \ref{lemma:R_Q_det_bun2} shows that
the two canonical trivializations for 
${\rm det}(\bb{E}_1|_{X^\tau})$ and ${\rm det}(\bb{E}_2|_{X^\tau})$ coincide (up to a suitable identification). Then $\omega_{\bb{E}_1}$ and $\omega_{\bb{E}_2}$
may differ only for the multiplication by a gauge transformation $[f]\in[X,\n{U}(1)]_{\Z_2}$ which connects the two global trivializations of the isomorphic $\rr{R}$-line bundles 
${\rm det}(\bb{E}_1)\simeq{\rm det}(\bb{E}_2)$.
This implies $\kappa(\bb{E}_1)=\kappa(\bb{E}_2)$ by construction. 

(i) The condition $f(Y)\cap X^{\tau_X}\neq\emptyset$ implies $Y^{\tau_Y}\neq\emptyset$. Moreover, the global  $\rr{R}$-trivialization
$h_{\rm det}:{\rm det}(\bb{E})\to X\times\C$ 
induces the global  $\rr{R}$-trivialization
$f^*(h_{\rm det}):{\rm det}(f^*(\bb{E}))\to Y\times\C$ 
for  ${\rm det}(f^*(\bb{E}))\simeq f^*({\rm det}(\bb{E}))$.
Then, also $f^*(\bb{E})$ is a $\rr{Q}$-bundle of FKMM-type
and $\kappa(f^\ast(\bb{E}))$ is well defined. The relations 
$f^*(h_{\rm det})=h_{\rm det}\circ {\rm det}(\hat{f})$ and 
${\rm det}_{Y^{\tau_Y}}={\rm det}_{X^{\tau_X}} \circ {\rm det}(\hat{f})$ imply $\kappa(f^\ast(\bb{E}))=f^\ast(\kappa(\bb{E}))$ (here $\hat{f}$ denotes the canonical morphism between  total spaces induced by $f$). 

(ii)
In this case $\omega_{\bb{E}}$ is the constant map on $X^\tau$ with value $+1$. 

(iii) It follows by construction from the functorial relation ${\rm det}(\bb{E}_1\oplus\bb{E}_2)\simeq{\rm det}(\bb{E}_1)\otimes{\rm det}(\bb{E}_2)$.

(iv) This is a consequence of  Lemma \ref{lemma:coker}.
\qed

\section{Classification in dimension $d\leqslant 3$}
The aim of this section is to classify the sets ${\rm Vec}^{2m}_\rr{Q}({\n{S}}^d,\tau)$ and  ${\rm Vec}^{2m}_\rr{Q}({\n{T}}^d,\tau)$
of isomorphism classes of \virg{Quaternionic} vector bundles over the involutive spaces described in Definition \ref{def:AII_sys} in the low dimensional regime $d\leqslant 3$. Since we already know the classification for $d=1$ (\cf Proposition \ref{prop_(ii)}) we can restrict our attention to $d=2,3$. The case $d=4$ will be considered separately in Section \ref{sec:d=4}.

\medskip

In some generality, the analysis which provides the desired classification can be extended to involutive spaces which share many of the features of the  TR-spaces $\tilde{\n{T}}^d$ and $\tilde{\n{S}}^d$. These are 
exactly the FKMM-spaces of Definition \ref{ass:2}.
When $(X,\tau)$ is a FKMM-space of \virg{low dimension}, namely 
with  free $\Z_2$-cells of  dimension not exceeding $d=3$,  Corollary \ref{corol_(iii)} and  Theorem \ref{theo:FKMM_propert} imply that the 
 following facts hold true:
\begin{itemize}
\vspace{1.3mm}
\item ${\rm Vec}^{2m}_\rr{Q}(X,\tau)\simeq {\rm Vec}^{2}_\rr{Q}(X,\tau)$ for all $m\in\N$;
\vspace{1.3mm}
\item The FKMM-invariant associated with an element  $[\bb{E}]\in{\rm Vec}^{2m}_\rr{Q}(X,\tau)$ coincides with the FKMM-invariant associated with its non-trivial part  $[\bb{E}']\in{\rm Vec}^{2}_\rr{Q}(X,\tau)$;
\vspace{1.3mm}
\item The FKMM-invariant provides a map $\kappa:{\rm Vec}^{2}_\rr{Q}(X,\tau)\to H^{2}_{\Z_2}(X|X^\tau,\n{Z}(1))$.
\end{itemize}
The main result we want to prove is  that the map $\kappa$ is in effect an isomorphism  in many cases of interest like $\tilde{\n{T}}^d$ or $\tilde{\n{S}}^d$.

\subsection{Injectivity of the FKMM-invariant}
\label{sec:injFKMM-inv}
The first step of our analysis is to show that the map
\begin{equation}\label{eq:inv_inj}
\kappa\;:\;{\rm Vec}^{2m}_\rr{Q}(X,\tau)\;\longrightarrow\; H^{2}_{\Z_2}(X|X^\tau,\n{Z}(1))
\end{equation}
is injective if $(X,\tau)$ is an FKMM-space.  We start with an important (and quite general) technical result.

\begin{lemma}[\cite{furuta-kametani-matsue-minami-00}]\label{lemma:inv_inject1}
Let $(\bb{E}_1,\Theta_1)$ and $(\bb{E}_2,\Theta_2)$ be  rank 2 \virg{Quaternionic} vector bundles of FKMM-type over the involutive space $(X,\tau)$. Assume also the equality $\kappa(\bb{E}_1)=\kappa(\bb{E}_2)$ of the respective  FKMM-invariants  and the existence of a  $\rr{Q}$-isomorphism $F:\bb{E}_1|_{X^\tau}\to \bb{E}_2|_{X^\tau}$. Then:

\begin{enumerate}
\vspace{1.3mm}
\item[(i)] There exists an $\rr{R}$-isomorphism $f:{\rm det}(\bb{E}_1)\to{\rm det}(\bb{E}_2)$ such that $f|_{X^\tau}={\rm det}(F)$;

\vspace{1.3mm}
\item[(ii)] The set
$$
\f{SU}(\bb{E}_1,\bb{E}_2,f)\;:=\;\bigsqcup_{x\in X}\;\left\{G_x:\bb{E}_1|_x\to\bb{E}_2|_x\ \left|\ 
\begin{aligned}
&\text{vector bundle isomorphism}\\
&\text{such that}\ {\rm det}(G_x)=f|_x\\
\end{aligned}
\right. \right\}
$$
defines a locally trivial fiber bundle over $X$ with typical fiber ${\rm S}\n{U}(2)$;
\vspace{1.3mm}
\item[(iii)] There is a natural involution on $\f{SU}(\bb{E}_1,\bb{E}_2,f)$ covering $\tau$. Moreover, on the fixed points $x\in X^\tau$
this involution is identifiable  with the involutive map $\mu:U\mapsto -Q\overline{U}Q$  on ${\rm S}\n{U}(2)$ defined in Remark \ref{rk:quat_str};
\vspace{1.3mm}
\item[(iv)] The set of equivariant sections of $\f{SU}(\bb{E}_1,\bb{E}_2,f)$ is in bijection with the set of $\rr{Q}$-isomorphisms
$G:\bb{E}_1\to\bb{E}_2$  such that ${\rm det}(G)=f$.
\end{enumerate}
\end{lemma}
\proof (i) By Lemma \ref{lemma:R_Q_det_bun2} we know that there exist two unique \emph{canonical} $\rr{R}$-sections $s_{X^\tau}^{(j)}:X^\tau\to{\rm det}(\bb{E}_j)|_{X^\tau}$, $j=1,2$  and from their construction it results   evident that ${\rm det}(F)\circ s_{X^\tau}^{(1)}=s_{X^\tau}^{(2)}$. According to Remark \ref{rk:FKMM_section}, the FKMM-invariant $\kappa(\bb{E}_j)$ is representable as a $\Z_2$-equivariant map $\omega_{\bb{E}_j}:X^\tau\to\n{U}(1)$ such that $\omega_{\bb{E}_j}\cdot s_{X^\tau}^{(j)}=t^{(j)}|_{X^\tau}$ where $t^{(j)}:X\to {\rm det} (\bb{E}_j)$ is any arbitrarily chosen global $\rr{R}$-section. The assumption    
$\kappa(\bb{E}_1)=\kappa(\bb{E}_2)$ implies the existence of a $\Z_2$-equivariant map $u:X\to\n{U}(1)$ such that $\omega_{\bb{E}_1}=u|_{X^\tau}\cdot\omega_{\bb{E}_2}$. Moreover, since the $\rr{R}$-line bundles ${\rm det}(\bb{E}_j)\to X$ are both trivial there is an $\rr{R}$-isomorphism $g:{\rm det}(\bb{E}_1)\to {\rm det}(\bb{E}_2)$ 
 such that $g\circ t^{(1)}=t^{(2)}$. Then, by combining these facts
  one deduces the existence of an $\rr{R}$-isomorphism $f:{\rm det}(\bb{E}_1)\to {\rm det}(\bb{E}_2)$  given by $f(p)=u(x)\, g(p)$ for all $p\in{\rm det}(\bb{E}_1)|_x$ such that  $f|_{X^\tau}\circ s_{X^\tau}^{(1)}= s_{X^\tau}^{(2)}$. The uniqueness of the  canonical sections implies $f|_{X^\tau}={\rm det}(F)$.
 
 (ii) Let ${\rm Hom}(\bb{E}_1,\bb{E}_2)\to X$ be the \emph{homomorphism vector bundle} (\cf \cite[Chapter 5, Section 6]{husemoller-94}). The inclusion $\f{SU}(\bb{E}_1,\bb{E}_2,f)\subset {\rm Hom}(\bb{E}_1,\bb{E}_2)$ provides a topology for $\f{SU}(\bb{E}_1,\bb{E}_2,f)$. Moreover, if $\f{U}$ is a trivializing neighborhood such that ${\rm Hom}(\bb{E}_1,\bb{E}_2)|_{\f{U}}\simeq\f{U}\times\n{U}(2)$ (we are assuming that a Hermitian metric has been fixed), then
$\f{SU}(\bb{E}_1,\bb{E}_2,f)|_{\f{U}}\simeq\f{U}\times{\rm S}\n{U}(2)$ 
due to the requirement $\det G = f$ in the definition of $\f{SU}(\bb{E}_1,\bb{E}_2,f)$.
 
(iii) A natural involution $\hat{\Theta}$ on  $\f{SU}(\bb{E}_1,\bb{E}_2,f)$ which relates $\tau$-conjugated fibers   is given by the collection of maps
$\f{SU}(\bb{E}_1,\bb{E}_2,f)|_x\ni G_x\mapsto -\Theta_2|_x\circ G_x\circ\Theta_1|_{\tau(x)}\in \f{SU}(\bb{E}_1,\bb{E}_2,f)|_{\tau(x)}$.

(iv) Follows from (ii) and (iii).
\qed

\medskip

\noindent
Though $\f{SU}(\bb{E}_1,\bb{E}_2,f)\to X$ is a fiber bundle with typical fiber ${\rm S}\n{U}(2)$, it is not a \emph{principal} ${\rm S}\n{U}(2)$-bundle. This object plays an essential role in the proof of  one of our main results:

\proof[
of Theorem \ref{theorem:inv_inject2}]
Due to the low dimension assumption (\cf Corollary \ref{corol_(iii)}) and the properties of the FKMM-invariant (\cf Theorem \ref{theo:FKMM_propert}) it is enough to prove the claim only for the case of rank 2 $\rr{Q}$-bundles (\ie $m=1$).

Let $(\bb{E}_1,\Theta_1)$ and $(\bb{E}_2,\Theta_2)$ be  rank 2 $\rr{Q}$-bundles  over the FKMM-space $(X,\tau)$. 
(We notice that as a consequence of Proposition \ref{prop:FKMM_type_space} the two $\rr{Q}$-bundles  are of FKMM-type.
Since $X^\tau$ is a finite collection of points, the restricted bundles  $\bb{E}_j|_{X^\tau}$, $j=1,2$ are both $\rr{Q}$-trivial   and so we can set a $\rr{Q}$-isomorphism
$F:\bb{E}_1|_{X^\tau}\to\bb{E}_2|_{X^\tau}$ induced by these trivializations.  If  $\kappa(\bb{E}_1)=\kappa(\bb{E}_2)$,  Lemma \ref{lemma:inv_inject1} allows us to introduce  the fiber bundle $\f{SU}(\bb{E}_1,\bb{E}_2,f)\to X$. To complete the proof of the injectivity, it is enough to prove that there exists a global $\Z_2$-equivariant section of  $\f{SU}(\bb{E}_1,\bb{E}_2,f)$. This fact can be proved exactly as in 
Proposition \ref{prop:stab_ran_Q}
by using the fact that the fibers of $\f{SU}(\bb{E}_1,\bb{E}_2,f)$ are 2-connected as a consequence of 
 $\pi_k({\rm S}\n{U}(2))=0$ for $k=0,1,2$.

Let us  describe first the group structure of ${\rm Vec}^{2}_\rr{Q}(X,\tau)$.
For isomorphism classes $[\bb{E}_1], [\bb{E}_2]\in {\rm Vec}^{2}_\rr{Q}(X,\tau)$, we define
the addition by $[\bb{E}_1] + [\bb{E}_2] := [\bb{E}]$, where $\bb{E}$ is a rank 2 $\rr{Q}$-bundle such that
$\bb{E}_1 \oplus \bb{E}_2 \simeq \bb{E} \oplus (X\times \C^2)$ and $X\times \C^2$ stands for the trivial  product $\rr{Q}$-bundle. The existence of $\bb{E}$ is guaranteed by Theorem \ref{theo:stab_ran_Q}. The class $[X\times \C^2]$ plays the role of  the neutral element and,
in view of Theorem \ref{theo:FKMM_propert},
 the FKMM-invariant
acts as a (multiplicative) morphism $\kappa([\bb{E}])=\kappa([\bb{E}_1]) \cdot\kappa([\bb{E}_2])$. This morphism is well defined since the value of the FKMM-invariant depends only on the equivalence class 
of a $\rr{Q}$-bundle and
$$
\kappa(\bb{E})\;=\;\kappa(\bb{E}\oplus (X\times \C^2))\;=\;\kappa(\bb{E}_1 \oplus \bb{E}_2 )\;=\;
\kappa(\bb{E}_1) \cdot\kappa( \bb{E}_2 )\;.
$$
Moreover, if  $\bb{E}'$ is a second rank 2 $\rr{Q}$-bundle such that
$\bb{E}_1 \oplus \bb{E}_2 \simeq \bb{E}' \oplus (X\times \C^2)$ the above computation shows that 
$\kappa(\bb{E})=\kappa(\bb{E}')$ and the injectivity of $\kappa$ implies that $\bb{E}\simeq \bb{E}'$ are in the same equivalence class. 
To define the inverse element in ${\rm Vec}^{2}_\rr{Q}(X,\tau)$ we observe that under the hypothesis above 
 for each $\rr{Q}$-bundle $(\bb{E},\Theta)$ there exists a  $\rr{Q}$-bundle $(\bb{E}',\Theta')$ such that $\bb{E}\oplus \bb{E}'\simeq X\times\C^4$. In this case the injectivity of $\kappa$ allows us to define $[\bb{E}']=-[\bb{E}]$. The construction of  $\bb{E}'$ is quite explicit: The complex vector bundle $\bb{E}\to X$ is trivial (\cf Proposition \ref{propos:Q_struc_lod_d}) and so it admits a global frame $\{t_1,t_2\}\in\Gamma(\bb{E})$
subjected to a $\Theta$-action of type \eqref{eq:mat_W010}. Let $\bb{E}'\to X$ be the rank 2 $\rr{Q}$-bundle generated by the global frame $t'_j:=\tau_\Theta(t_j)=\Theta\circ t_j\circ\tau$, $j=1,2$ and endowed with the $\rr{Q}$-structure induced by $\Theta$. The sum  $\bb{E}\oplus \bb{E}'$ is $\rr{Q}$-trivial since the collection $\{s_1=t_1,s_2=t'_1,s_3=t_2, s_4=t'_2\}$ provides a global $\rr{Q}$-frame. Finally, a group structure on ${\rm Vec}^{2m}_\rr{Q}(X,\tau)$ can be induced from the group structure of
${\rm Vec}^{2}_\rr{Q}(X,\tau)$  and the isomorphism in Corollary \ref{corol_(iii)}. 
\qed

\begin{remark}\label{rk:obstr1}
{\upshape
Under the hypothesis of Theorem \ref{theorem:inv_inject2} the FKMM-invariant suffices to establish the non-triviality of a given \virg{Quaternionic} vector bundle. Hence,  as a consequence of Theorem \ref{theo:triv_glob},
we deduce that the FKMM-invariant can be interpreted as the (first) topological obstruction for the existence of a 
global $\rr{Q}$-frame. This is exactly the point of view explored in \cite{porta-graf-13} in the particular case of the involutive space $\tilde{\n{T}}^2$.
}
 \hfill $\blacktriangleleft$
\end{remark}

\subsection{General structure of low-dimensional $\rr{Q}$-bundles}
\label{sec:gen_struct_low_dim}
 The low dimension assumption $d\leqslant 3$ provides a certain number of simplifications in the description of \virg{Quaternionic} vector bundles. 

 \begin{proposition}
 \label{propos:Q_struc_lod_d}
Let $(X,\tau)$ be an involutive space  such that  $X$  verifies (at least) condition 0 of Definition \ref{ass:2}
and has cells of dimension not bigger than $d=3$. Assume also  
$X^\tau\neq\emptyset$. If $(\bb{E},\Theta)$ is a (even rank) $\rr{Q}$-bundle over  $(X,\tau)$ of FKMM-type, then the underlying  vector bundle $\bb{E}\to X$ is trivial in the category of complex vector bundles.
 \end{proposition}
\proof[
sketch of] 
Due to the low dimension assumption $\bb{E}\to X$ is trivial (as a complex vector bundle) if and only if its first Chern class $c_1(\bb{E})\in H^2(X,\Z)$ is trivial (\cf \cite[Section 3.3]{denittis-gomi-14}). Moreover, one of the main properties of the determinant construction is $c_1(\bb{E})=c_1({\rm det}(\bb{E}))$. Since $(\bb{E},\Theta)$ is of FKMM-type, it follows that ${\rm det}(\bb{E})\to X$ is an $\rr{R}$-trivial line bundle
and so the first \virg{Real} Chern class is trivial,  $\tilde{c}_1({\rm det}(\bb{E}))=0$  \cite[Theorem 5.1]{denittis-gomi-14}. Finally, under the map that  forgets the \virg{Real}-structure one has that $\tilde{c}_1({\rm det}(\bb{E}))=0$ implies ${c_1}({\rm det}(\bb{E}))=0$ (\cf \cite[Section 5.6]{denittis-gomi-14}).
\qed

\medskip

As a consequence of this result each $\rr{Q}$-bundle $(\bb{E},\Theta)$ of FKMM-type 
over a base space of dimension not bigger than 3
can be endowed with a {global frame} of  sections $\{t_1,\ldots,t_{2m}\}\subset\Gamma(\bb{E})$. However, this frame is not in general a $\rr{Q}$-frame in the sense of Definition \ref{def:Q_pair} and the possibility to deform $\{t_1,\ldots,t_{2m}\}$ into a $\rr{Q}$-frame is linked to the 
$\rr{Q}$-triviality of $(\bb{E},\Theta)$ as a \virg{Quaternionic} vector bundle. The $\rr{Q}$-structure acts on the frame $\{t_1,\ldots,t_{2m}\}$ via the map $\tau_\Theta$ (\cf Section \ref{sect:Qs-vb_sect})
\beql{eq:mat_W010}
\tau_\Theta(t_i)(x)\;=\;(\Theta \circ t_i)(\tau(x))\;=\;\sum_{j=1}^{2m}w_{ji}(\tau(x))\; t_j(x)\; \qquad\quad i=1,\ldots,2m\;,\ \ \ x\in X\;,
\eeq
where $w_{ji}(x)\in\C$ are   components of a matrix-valued map $w:X\to{\rm Mat}_{2m}(\C)$. The relation $\tau_\Theta^2=-{\rm Id}$ implies 
\beql{eq:mat_W1}
\sum_{k=1}^{2m}\overline{w_{ki}(\tau(x))}\;w_{jk}(x)\;=\;-\delta_{i,j}\;,\qquad\quad \forall\ x\in X\;.
\eeq
Without loss of generality, one can assume that   $(\bb{E},\Theta)$ is endowed with an equivariant Hermitian  metric $\rr{m}$ (\cf Proposition \ref{rk:eq_metric}) with respect to which the frame $\{t_1,\ldots,t_{2m}\}$ is orthonormal. In this case the $\rr{Q}$-structure of $(\bb{E},\Theta)$ is encoded
 in a map $w:X\to\n{U}(2m)$ with components
\beql{eq:mat_W010bis}
w_{ji}(x)\;:=\;\rr{m}\big(t_j(\tau(x)),\tau_\Theta(t_i)(\tau(x))\big)\;=\;\rr{m}\big(t_j(\tau(x)),\Theta\circ t_i(x)\big)
\eeq
and the relation \eqref{eq:mat_W1} reads 
\beql{eq:mat_W2}
w(\tau(x))\;=\;-w^t(x),\qquad\quad \forall\ x\in X\;
\eeq
where $w^t$ is the \emph{transpose} of the matrix $w$. In particular, \eqref{eq:mat_W2} implies the \emph{skew}-symmetric relation $w(x)=-w^t(x)$
on the fixed point set $x\in X^\tau$.

\medskip

If one uses the  (orthonormal) frame  $\{t_1,\ldots,t_{2m}\}$ as a basis for the trivialization of the vector bundle 
one ends with the identification
 $\bb{E}\simeq\ X\times\C^{2m}$ and the  $\rr{Q}$-structure $\Theta: X\times\C^{2m}\to X\times\C^{2m}$ is 
fixed   by 
\beql{eq:mat_W3}
\Theta\;:\; (x,{\rm v})\;\longrightarrow \;(\tau(x),w(x)\;\overline{\rm v})\;.
\eeq
Therefore,
 in view of Proposition \ref{propos:Q_struc_lod_d}, the most general \virg{Quaternionic} vector bundle of FKMM-type over an involutive base space $(X,\tau)$ of dimension not bigger than $3$ is represented by a product  $X\times\C^{2m}$ endowed with a map \eqref{eq:mat_W3} associated with a matrix-valued continuous function $w:X\to\n{U}(2m)$ which verifies  \eqref{eq:mat_W2}. The topology of such a $\rr{Q}$-bundle is 
 entirely contained in the map $w$. For instance the question about the  $\rr{Q}$-triviality can be rephrased in the equivalent question about  the existence of a map $u:X\to \n{U}(2m) $
such that $u^*(x)\cdot w(\tau(x))\cdot \overline{u(\tau(x))}=Q$ for all $x\in X$, where the matrix $Q$ is the one defined in \eqref{eq:Q-mat}.

\medskip

Also the FKMM-invariant can be reconstructed from $w$. This is linked to the fact that the induced \virg{Real} structure on the determinant line bundle ${\rm det}(\bb{E})\simeq\ X\times\C^{1}$ is specified by the map 
\beql{eq:mat_W3'}
{\rm det}(\Theta)\;:\; (x,\lambda)\;\longrightarrow \;(\tau(x),{\rm det}(w)(x)\;\overline{\lambda})\;.
\eeq
where ${\rm det}(w):X\to\n{U}(1)$ inherits from \eqref{eq:mat_W2} the property ${\rm det}(w)(x)={\rm det}(w)(\tau(x))$ for all $x\in X$.
Moreover, since $({\rm det}(\bb{E}), {\rm det}(\Theta))$ is $\rr{R}$-trivial by assumption there  exists a map $q_w:X\to \n{U}(1)$
such that ${\rm det}(w)(x)=q_w(x)q_w(\tau(x))$ for all $x\in X$. The choice of $q_w$ is not unique. In fact, if $\varepsilon:X\to\n{U}(1)$
is a $\Z_2$-equivariant map $\varepsilon(\tau(x))=\overline{\varepsilon(x)}$ then $q'_w(x):=\varepsilon(x)q(x)$ also verifies ${\rm det}(w)(x)=q'_w(x)q'_w(\tau(x))$. Moreover, all the possible choices for $q_w$ are related by a gauge transformation of this type.
On the fixed points  $x\in X^\tau$ the matrix $w(x)$ is skew-symmetric and so ${\rm det}(w)(x)={\rm Pf}[w(x)]^2$
 where the symbol ${\rm Pf}[A]$ is used for the \emph{Pfaffian} of the skew-symmetric matrix $A$ (\cf \cite[Appendix C, Lemma 9]{milnor-stasheff-74}).   Therefore,  the map $w:X\to\n{U}(2m)$ and the notion of Pfaffian fix  a well defined map  
${\rm Pf}_w: X^\tau\to\n{U}(1)$ given by ${\rm Pf}_w(x):={\rm Pf}[w(x)]$. Observing that on $X^\tau$ the  maps
$q_w$ and ${\rm Pf}_w$ can differ only by a sign, we can set a \emph{sign map}
 $\rr{d}_{w}:X^\tau\to\{\pm 1\}$  by 
\beql{eq:sign_map_fkm}
 \rr{d}_w(x)\;:=\;\frac{q_w(x)}{{\rm Pf}_w(x)}\;\qquad\qquad x\in X^\tau\;.
\eeq
By construction,  the map $ \rr{d}_w$ is defined only up to the choice of a gauge transformation $\varepsilon$ and so it can be identified with  a representative for a class $[\rr{d}_w]$ in the cokernel $[X^\tau,\n{U}(1)]_{\Z_2}/[X,\n{U}(1)]_{\Z_2}$.

\begin{proposition}[FKMM-invariant via the sign map]
\label{propo:FKMM-sign}
Let $(X,\tau)$ be an involutive space  
such that  $X$  verifies (at least) condition 0 of Definition \ref{ass:2}
and has cells of dimension not bigger than $d=3$.
 Assume also  
$X^\tau\neq\emptyset$. Let $\bb{E}\simeq X\times\C^{2m}$ endowed with a $\rr{Q}$-structure $\Theta$ of type \eqref{eq:mat_W3}
associated to a map $w:X\to\n{U}(2m)$ which verifies \eqref{eq:mat_W2}. Then
$$
\kappa(\bb{E})\;=\;[\rr{d}_w]
$$
where the sign map  $\rr{d}_w$ is the one defined by \eqref{eq:sign_map_fkm}.
\end{proposition}
\proof
Let us consider the section $s_{X^\tau}:X^\tau\to X^\tau\times\C^1$ defined by $s_{X^\tau}(x):=(x, {\rm Pf}_w(x))$. It is not difficult to check that this section is \virg{Real}. Indeed
$$
({\rm det}(\Theta)\circ s_{X^\tau}\circ\tau)(x)\;=\;{\rm det}(\Theta)(x, {\rm Pf}_w(x))\;=\;(x, {\rm det}(w)(x)\;{\rm Pf}_w(x)^{-1})\;=\; s_{X^\tau}(x)\;.
$$
This section agrees with the \emph{canonical} section for the restricted $\rr{R}$-bundle $X^\tau\times\C^1$ described in Lemma \ref{lemma:R_Q_det_bun2}. On the other side, the map $t:X\to  X\times \C^1$ given by $t(x):=(x, q_w(x))$ provides a global  \virg{Real}-section
 since
 $$
({\rm det}(\Theta)\circ t \circ\tau)(x)\;=\;{\rm det}(\Theta)(\tau(x), q_w(\tau(x)))\;=\;(x, {\rm det}(w)(x)\;q_w(\tau(x))^{-1})\;=\; t(x)\;.
$$
Evidently, the difference between the  {canonical} section  $s_{X^\tau}$ and the restricted section $t|_{X^\tau}$ is expressed exactly by the sign map 
$\rr{d}_w$. This implies, in view of Remark \ref{rk:FKMM_section}, that   $\rr{d}_w$ provides a representative for the FKMM-invariant 
of   $(\bb{E},\Theta)$, namely 
$\kappa(\bb{E})=[\rr{d}_w]$.
\qed

\begin{remark}[The sign map by Fu, Kane \& Mele]\label{rk:FKMM_w_det}{\upshape
The sign map \eqref{eq:sign_map_fkm} was introduced for the first time by L. Fu, C. L. Kane and E. J. Mele in a series of works concerning \virg{Quaternionic} vector bundles over TR-tori $\tilde{\n{T}}^d$ with $d=2$ \cite{kane-mele-05,fu-kane-06}  and $d=3$ \cite{fu-kane-mele-07}.
In these works the topology of the $\rr{Q}$-bundle is investigated through the properties of a matrix-valued map $w$ defined similarly to \eqref{eq:mat_W010bis} and a related  topological invariant that agrees (morally) with our sign map \eqref{eq:sign_map_fkm}. However, 
there is a difference:  the map $q_w$ (that expresses the $\rr{R}$-triviality of the determinat line bundle) is replaced by a less natural and more ambiguous \emph{branched} function $\sqrt{{\rm det}(w)}$. Let us point out that, albeit apparently similar, our definition of the sign map is much more general and it applies to all $\rr{Q}$-bundles of FKMM-type over base spaces of dimension not bigger that 3 (which in turns implies that the underlying complex
vector bundles are trivial).}
 \hfill $\blacktriangleleft$
\end{remark}

\subsection{Classification for TR-spheres}

This section is devoted to the justification of the following fact:
\begin{proposition}
\label{propos_classif_low_S}
Let $\tilde{\n{S}}^d\equiv({\n{S}}^d,\tau)$ 
be the  TR-sphere of Definition \ref{def:AII_sys}. Then,  for all $m\in\N$ there is a group isomorphism
$$
{\rm Vec}^{2m}_\rr{Q}({\n{S}}^d,\tau)\;\stackrel{\kappa}{\simeq}\;H^{2}_{\Z_2}\big(\tilde{\n{S}}^d|(\tilde{\n{S}}^d)^\tau,\n{Z}(1)\big)\;\simeq\;\Z_2\;\qquad\quad\ \text{if}\ \ d=2,3\;
$$
which is established  by the FMMM-invariant $\kappa$.
\end{proposition}
\proof
Since the second isomorphism has already been established in  \eqref{eq:coker_TR_sphe}, it remains only to justify the first  isomorphism.
 Theorem \ref{theorem:inv_inject2} already assures  the injectivity of the group morphism $\kappa$. Hence, we need only to show that both
 ${\rm Vec}^{2m}_\rr{Q}({\n{S}}^2,\tau)$ and ${\rm Vec}^{2m}_\rr{Q}({\n{S}}^3,\tau)$ have a non-trivial element.  Corollary \ref{corol_(iii)} assures that it is enough to show the existence of  non-trivial  $\rr{Q}$-bundle of rank 2.  A specific realization of 
  these  non-trivial elements  is presented below.
\qed

\medskip

Let us start with the two-sphere $\tilde{\n{S}}^2$ endowed with the  TR involution $\tau(k_0,k_1,k_2)=(k_0,-k_1,-k_2)$. Following the arguments in Section \ref{sec:gen_struct_low_dim} we can represent any rank two \virg{Quaternionic} vector bundle $(\bb{E},\Theta)$ over $\tilde{\n{S}}^2$ by a product bundle $\bb{E}\equiv \n{S}^2\times\C^2$ on which the action of $\Theta$ is specified by a map $w:\n{S}^2\to\n{U}(2)$
such that $w(k_0,-k_1,-k_2)=-w^t(k_0,k_1,k_2)$ for all $k=(k_0,k_1,k_2)\in\n{S}^2$ according to the relation \eqref{eq:mat_W2}. One possible realization for $w$ is
\beql{eq:ex_d2_w}
w(k_0,k_1,k_2)\;:=\;
\left(\begin{array}{cc}
k_1+\ii k_2 & +k_0 \\
-k_0 & k_1-\ii k_2
\end{array}\right)\;.
\eeq
We can compute the FKMM-invariant of this $\rr{Q}$-bundle by using the sign map as in Proposition \ref{propo:FKMM-sign}.
Since ${\rm det}(w)$ is constantly 1 we can chose $q_w(k)=1$ for all $k\in\n{S}^2$. On the other side, on the fixed points $k_\pm:=(\pm1,0,0)$ 
a simple computation shows that ${\rm Pf}_w(k_\pm)=\pm1$. Therefore, the sign map $\rr{d}_w(k_\pm)=\pm1$ provides a representative for the non trivial element in $[(\tilde{\n{S}}^2)^\tau,\n{U}(1)]_{\Z_2}\;/\;[\tilde{\n{S}}^2,\n{U}(1)]_{\Z_2}\simeq\Z_2$.
Equivalently, $\kappa(\bb{E})=[\rr{d}_w]$ can be identified with the non trivial element $-1\in\Z_2$ showing that the $\rr{Q}$-bundle $(\bb{E},\Theta)$ associated with $w$ is non-trivial.

\medskip

For the three-sphere $\tilde{\n{S}}^3$  with  involution $\tau(k_0,k_1,k_2,k_3)=(k_0,-k_1,-k_2,-k_3)$ the argument is similar.
A non-trivial $\rr{Q}$-structure over $\bb{E}\equiv \n{S}^3\times\C^2$ can be  induced by the map $w:\n{S}^3\to\n{U}(2)$
\beql{eq:ex_d3_w}
w(k_0,k_1,k_2,k_3)\;:=\;
\left(\begin{array}{cc}
k_1+\ii k_2 & \ii k_3+k_0 \\
\ii k_3-k_0 & k_1-\ii k_2
\end{array}\right)\;.
\eeq
It is straightforward to check that this map verifies  \eqref{eq:mat_W2}.
As before, we can compute the FKMM-invariant by means of the sign map as in Proposition \ref{propo:FKMM-sign}.
Again ${\rm det}(w)(k)=1$ allows us to set $q_w(k)=1$. Moreover, on the fixed points $k_\pm:=(\pm1,0,0,0)$ 
one checks  ${\rm Pf}_w(k_\pm)=\pm1$. Therefore,   $\rr{d}_w(k_\pm)=\pm1$ provides a representative for the non trivial element in $[(\tilde{\n{S}}^3)^\tau,\n{U}(1)]_{\Z_2}\;/\;[\tilde{\n{S}}^3,\n{U}(1)]_{\Z_2}\simeq\Z_2$  showing that the associated $\rr{Q}$-bundle can not be trivial.

\medskip

For later use in Section \ref{sec:d=4}, we derive  here a consequence of Proposition \ref{propos_classif_low_S}. To state it we need   the sphere $\hat{\n{S}}^2\equiv(\n{S}^2,\vartheta)$ endowed with the \emph{antipodal} action $\vartheta:k\mapsto -k$. We know that the group $[\hat{\n{S}}^2, \n{U}(1)]_{\Z_2} \simeq H^1_{\Z_2}(\hat{\n{S}}^2, \Z(1)) \simeq \Z_2$ is generated by the constant maps with values $\{\pm 1\} \subset \n{U}(1)$
\cite{denittis-gomi-14}.

\begin{corollary}
\label{corollary:classification_on_3_sphere}
Let $\hat{\n{S}}^2$ be the sphere endowed with the antipodal action $k\mapsto -k$. Then the determinant $\det : \n{U}(2) \to \n{U}(1)$ induces an isomorphism of groups
$$
\det\; :\;  [\hat{\n{S}}^2,\n{U}(2)]_{\Z_2}\;\longrightarrow\;
[\hat{\n{S}}^2, \n{U}(1)]_{\Z_2}\; \simeq\; \Z_2
$$
where the space $\n{U}(1)$ is endowed with the involution given by the complex conjugation and $\n{U}(2)$ with the involution $\mu$ described in Remark \ref{rk:quat_str}.
\end{corollary}

\proof
Let $H_\pm := \{ (k_0, k_1, k_2, k_3) \in \tilde{\n{S}}^3 |\ \pm k_0 \ge 0 \}$. The intersection $H_+ \cap H_-$ is exactly $\hat{\n{S}}^2$. Because each $H_\pm$ is equivariantly contractible, the \emph{clutching construction} \cite[Lemma 1.4.9]{atiyah-67} adapted to $\rr{Q}$-bundles leads to a bijection ${\rm Vec}^{2}_{\rr{Q}}(\hat{\n{S}}^3) \simeq [\hat{\n{S}}^2, \n{U}(2)]_{\Z_2}$. Hence, Proposition \ref{propos_classif_low_S} implies that $[\hat{\n{S}}^2, \n{U}(2)]_{\Z_2} \cong \Z_2$, at least abstractly. To complete the proof, it suffices to verify that there exists a $\Z_2$-equivariant map $\varphi : \hat{\n{S}}^2 \to \n{U}(2)$ such that $\det\varphi :  \hat{\n{S}}^2 \to \n{U}(1)$ is non-trivial in $[\hat{\n{S}}^2, \n{U}(1)]_{\Z_2} \simeq \Z_2$. An example is provided by
$$
\varphi(k_1, k_2, k_3)\;:
= \;
\ii
\left(
\begin{array}{cc}
k_1 + \ii k_2 & - k_3 \\
k_3 & k_1 - \ii k_2
\end{array}
\right),
$$
whose determinant is evidently the constant map with value  $-1$.
\qed

\subsection{Connection with the Fu-Kane-Mele invariant}\label{subsec_Fu-Kane-Mele}
In this section we will give a deeper look to the classification of $\rr{Q}$-bundles over sufficiently \virg{nice} involutive spaces  of dimension $d=2$. With this we mean that:
 \begin{assumption}\label{ass:2.3}
Let $(X,\tau)$  be an involutive space such that:
\begin{enumerate}
\item[0.]  $X$ is a closed (compact without boundary) and  oriented $2$-dimensional manifold;\vspace{1.3mm}
\item[1.] The fixed point set $X^\tau\neq\emptyset$  consists of a finite number of points;
\vspace{1.3mm}
\item[2.]  The involution $\tau:X\to X$ is smooth and preserves the orientation.
\end{enumerate}
\end{assumption}

\noindent
We point out that both the  involutive spaces $\tilde{\n{S}}^2$ and $\tilde{\n{T}}^2$ fulfill the previous hypothesis.
Moreover, spaces of this type share with $\tilde{\n{S}}^2$ and $\tilde{\n{T}}^2$ the following general features:
\begin{proposition}
\label{propos:d2manif_FKMM}
Let $(X,\tau)$  be an involutive space which verifies Assumption \ref{ass:2.3}, then
$$
H^2_{\Z_2}(X,\Z(1))\;\simeq\; 0\;,\qquad\quad H^2_{\Z_2}(X|X^\tau,\n{Z}(1))\;\simeq\; \Z_2\;.
$$
Moreover, the number of fixed points is even, \ie $X^\tau=\{x_1,\ldots,x_{2n}\}$.
\end{proposition}

\noindent
We postpone the technical proof of this result  to Appendix \ref{sec:app_coker}. Here, we are mainly interested in the following consequences of Proposition \ref{propos:d2manif_FKMM}: if the space $(X,\tau)$ verifies Assumption \ref{ass:2.3} then it is a \emph{bona fide} FKMM-space and \virg{Quaternionic} vector bundles over it are classified by  the  FKMM-invariant which takes values in $\Z_2$.

\medskip

Let $\bb{L}\to X$ be a complex line bundle over an involutive space $(X,\tau)$. It is well known that the topology of 
$\bb{L}$   is fully specified by the first Chern class $c_1(\bb{L})\in H^2(X,\Z)$. Moreover, if $(X,\tau)$ is as in Assumption \ref{ass:2.3}, then   $c_1(\bb{L})$ is completely specified by the integer $C(\bb{L}):=\langle c_1(\bb{L});[X]\rangle\in\Z$ called \emph{(first) Chern number}. In the last equation the brackets $\langle\cdot ; \cdot\rangle$
denote the pairing between cohomology classes in $H^2(X,\Z)$ with the generator of the homology $[X]\in H_2(X)$ usually called the \emph{fundamental class}. This pairing can be also understood (when possible) as the integration  of the de Rham form which represents $c_1(\bb{L})$ over $X$. Let $\overline{\bb{L}}$ be the \emph{conjugate} line bundle of $\bb{L}$, namely 
$\overline{\bb{L}}$ and  $\bb{L}$ agree as (rank 2) real vector bundles but they have opposite complex structures.
It's well known that $c_1(\overline{\bb{L}})=-c_1(\bb{L})$ \cite[Lemma 14.9]{milnor-stasheff-74}. The involution $\tau:X\to X$ can be used to define the pullback line bundle $\tau^*(\bb{L})$. 
Since $\tau$ is an orientation-preserving involution it induces a map of degrees 1 in homology, hence $c_1(\tau^*(\bb{L}))=\tau^*c_1({\bb{L}})=c_1({\bb{L}})$. Given a line bundle $\bb{L}$ over $(X,\tau)$ we can build (via the Whitney sum) the rank 2 complex vector bundle
\beql{eq:classif_2d_summ}
\bb{E}_{\bb{L}}\;:=\; \tau^*(\bb{L})\;\oplus\;\overline{\bb{L}}\;.
\eeq
The first Chern class of $\bb{E}_{\bb{L}}$ vanishes identically since $c_1(\bb{E}_{\bb{L}})=c_1(\tau^*(\bb{L}))+c_1(\overline{\bb{L}})=c_1(\bb{L})-c_1(\bb{L})$. This agrees with the fact that $\bb{E}_{\bb{L}}$ admits a standard $\rr{Q}$-structure $\Theta$ which can be defined as follow: first of all notice that 
$\bb{E}_{\bb{L}}|_x\simeq\bb{L}|_{\tau(x)}\oplus\overline{\bb{L}}|_x$, then each point $p\in \bb{E}_{\bb{L}}|_x$ has the form $p=(l_1,\overline{l_2})$ with $l_1\in \bb{L}|_{\tau(x)}$ and $l_2\in {\bb{L}}|_x$
and the bar denotes the inversion of the complex structure (\ie the complex conjugation) in each fiber ${\bb{L}}|_x$. Accordingly, we can define the anti-linear anti-involution $\Theta$ between $\bb{E}_{\bb{L}}|_x$ and $\bb{E}_{\bb{L}}|_{\tau(x)}$ by $\Theta:(l_1,\overline{l_2})\mapsto (-l_2,\overline{l_1})$.

\begin{lemma}
\label{lemma:2d_manifA}
Let $(X,\tau)$ be an involutive space which verifies Assumption \ref{ass:2.3} and $\bb{L}\to X$ a complex line bundle with 
(first) Chern number $C=1$. Consider the  rank 2 \virg{Quaternionic} vector bundle $(\bb{E}_{\bb{L}},\Theta)$ associated with $\bb{L}$ by the construction \eqref{eq:classif_2d_summ}. For each fixed point $x_j\in X^\tau$, $j=1,\ldots,2n$, let us define a map $\phi_j:X^\tau\to\{\pm 1\}$ by $\phi_j(x_i):=1-2\delta_{i,j}$. Then, 
$$
\kappa(\bb{E}_{\bb{L}})\;=\;[\phi_1]\;=\;\ldots\;=\;[\phi_{2n}]\;\in\; [X^\tau,\n{U}(1)]_{\Z_2}/[X,\n{U}(1)]_{\Z_2}\;\simeq\;\Z_2\;,
$$
namely
all these maps  $\phi_j$ are representatives for the FKMM-invariant of $\bb{E}_{\bb{L}}$. Moreover,  $\kappa(\bb{E}_{\bb{L}})$ coincides with the non-trivial element $-1\in \Z_2$ showing that the $\rr{Q}$-bundle $(\bb{E}_{\bb{L}},\Theta)$ is non-trivial.
\end{lemma}
\proof 
Let $X^\tau:=\{x_1,\ldots,x_{2n}\}$ be the fixed point set of $(X,\tau)$. 
 As a consequence of the so-called \emph{slice Theorem} \cite[Chapter I, Section 3]{hsiang-75}  a neighborhood of each $x_j\in X^\tau$ can be identified with an open subset around the origin of $\R^2$ endowed with the involution given by the reflection $x\mapsto -x$. Let $D_j\subset X$ be a 
 \emph{disk} (under this identification) 
around $x_j$. More precisely, $D_j$ is an invariant
set $\tau(D_j)=D_j$ in which $x_j$ is the only fixed point; it is closed, $\Z_2$-contractible and with boundary $\partial D_j\simeq\n{S}^1$. Moreover, without loss of generality,
we can choose sufficiently small disks
    in such a way that $D_i\cap D_j=\emptyset$ if $i\neq j$.
 Let us set  $\f{D}:=\bigcup_{i=1}^{2n}D_i$ and $X':=X\setminus{\rm Int}(\f{D})$. By construction $X'$ is a manifold with boundary $\partial\f{D}\simeq\bigcup_{i=1}^{2n}\n{S}^1_i$ on which the involution $\tau$ acts freely.
 
Since the action $\tau:\partial D_j\to \partial D_j$ is free we can fix an  isomorphism $\tilde{\psi}_j:\partial D_j\simeq\n{S}^1\ni\theta\mapsto\expo{\ii\theta}\in\n{U}(1)$ which is \emph{antipodal}-equivariant, \ie $\tilde{\psi}_j\circ \tau=\expo{\ii\pi}\tilde{\psi}_j=-\tilde{\psi}_j$.
Each of these isomorphisms defines  a map $\psi_j:\partial\f{D}\simeq\n{S}^1\ni\theta\mapsto\expo{\ii\theta}\in\n{U}(1)$ given by  $\psi_j|_{\partial D_j}=\tilde{\psi}_j$ and $\psi_j|_{\partial D_i}\equiv1$ for all $i\neq j$. Finally, we can set maps $\Psi_j:\f{D}\to\n{U}(2)$  by
$$
\Psi_j(x)\;:=\;
\left(\begin{array}{cc}\psi_j(\tau(x)) & 0 \\0 & \overline{\psi_j(x)}\end{array}\right)
$$
which verify the equivariance condition $\Psi_j\circ\tau=\mu\circ \Psi_j$ where  the involution $\mu$ on $\n{U}(2)$
has been defined in Remark \ref{rk:quat_str}. It follows  that ${\rm det} (\Psi_j)(x)=-1$ if $x\in {\partial D_j}$ and ${\rm det} (\Psi_j)(x)=+1$ if $x\in\f{D}\setminus{\partial D_j}$.

 Associated with each $\Psi_j$ we can construct a $\rr{Q}$-bundle
 $$
 \bb{E}_j\;:=\;(X'\times\C^2)\;\cup_{\Psi_j}\;(\f{D}\times\C^2)
 $$
 by means of the equivariant version of the \emph{clutching construction} \cite[Lemma 1.4.9]{atiyah-67} based on the equivariant map ${\Psi_j}$. By construction
$\bb{E}_j$ turns out to be
 isomorphic to the $\rr{Q}$-bundle $\bb{E}_{\bb{L}_j}:=\tau^*(\bb{L}_j)\oplus\overline{\bb{L}_j}$ where the line bundle $\bb{L}_j=(X'\times\C)\;\cup_{\psi_j}\;(\f{D}\times\C)
$ is realized with the clutching construction based on the map $\psi_j$. Let $c_1(\bb{L}_j)$ be the Chern class of the line bundle  
 $\bb{L}_j\to X$. Since the clutching around $x_j$ (and only around this point)  is done with a phase-function of the type $D_j\simeq\n{S}\ni\theta\to\expo{\ii\theta}\in\n{U}(1)$ it follows that the first Chern number obtained from the integration of $c_1(\bb{L}_j)$ is $C_j=1$. In particular, it follows that   $\bb{L}_j\simeq \bb{L}$ for each $j=1,\ldots,2n$. This also implies $\bb{E}_j\simeq \bb{E}_{\bb{L}}$ for each $j=1,\ldots,2n$.
 
 To finish the proof we need only to show that   the FKMM-invariant of $\bb{E}_j$ is represented by $\phi_j$.
 To see this we choose an  equivariant nowhere vanishing  section of $t:X\to{\rm det}(\bb{E}_j)$. Sections of this type are in one-to-one correspondence with a pair of equivariant maps $u_{X'}:X'\to\n{U}(1)$ and $u_{\f{D}}:\f{D}\to\n{U}(1)$ such that $u_{X'}={\rm det} (\Psi_j)\cdot u_{\f{D}}$ on $X'\cap\f{D}$. As a specific choice, we can take 
 $u_{X'}\equiv 1$ and $u_{\f{D}}(x)=-1$ if $x\in {\partial D_j}$ and $u_{\f{D}}(x)=+1$ if $x\in\f{D}\setminus{\partial D_j}$. This choice shows that $\phi_j$ expresses the difference between the restricted section $t|_{X^\tau}$ and the (constant) canonical section $s_{X^\tau}$ associated with ${\rm det}(\bb{E}_j|_{X^\tau})$.
\qed

\begin{corollary}
\label{coroll:d2manif_FKMM}
Let $(X,\tau)$ be an involutive space which verifies Assumption \ref{ass:2.3}. Then ${\rm Vec}^{2m}_{\rr{Q}}(X,\tau)\simeq\Z_2$.
\end{corollary}
\proof
It follows from the injectivity of $\kappa$ proved in Theorem \ref{theorem:inv_inject2}.\qed

\begin{corollary}
\label{corol:Fu-Kane-Mele_formula}
Let $(X,\tau)$ be an involutive space which verifies Assumption \ref{ass:2.3} and $\bb{L}\to X$ a complex line bundle with 
(first) Chern number $C\in\Z$. Consider the  rank 2 \virg{Quaternionic} vector bundle $(\bb{E}_{\bb{L}},\Theta)$ associated with $\bb{L}$ by the construction \eqref{eq:classif_2d_summ}.
The topology of   $\bb{E}_{\bb{L}}$  is completely specified by the parity of $C$
through the formula 
$$
\kappa(\bb{E}_{\bb{L}})\;:=\;(-1)^C
$$
which relates the (image of the) FKMM-invariant $\kappa(\bb{E})\in\Z_2$ with the Chern class of $\bb{L}$.
\end{corollary}
\proof 
We can repeat the same proof of Lemma \ref{lemma:2d_manifA} with respect to a generalized isomorphism  $\tilde{\psi}_j:\partial D_j\simeq\n{S}^1\ni\theta\mapsto\expo{\ii C\theta}\in\n{U}(1)$ such that $\tilde{\psi}_j\circ \tau=\expo{\ii C\pi}\tilde{\psi}_j=(-1)^C \tilde{\psi}_j$. With this choice the line bundle $\bb{L}_j=(X'\times\C)\;\cup_{\psi_j}\;(\f{D}\times\C)$ has Chern number $C$ and the FKMM-invariant of $\bb{E}_{\bb{L}}\simeq \bb{E}_{\bb{L}_j}\simeq \bb{E}_{j}$ is represented by a function $\phi_j:X^\tau\to\{\pm1\}$ such that 
$\phi_j(x_j)=(-1)^C$ and
$\phi_j(x_i)=1$ if $i\neq j$.
\qed

\begin{remark}{\upshape
Corollary \ref{corol:Fu-Kane-Mele_formula} can be considered as the abstract version of the justification of the Quantum Spin Hall Effect given in \cite{kane-mele-05'} (\cf also \cite[eq. (3.26)]{fu-kane-06}). The Chern numbers associated with the two line bundles which define  $\bb{E}_{\bb{L}}$ are opposite in sign and, therefore  define opposite traveling currents. Since these currents carry opposite spins they sum up and produce a non trivial effect which is quantified by the parity of the absolute value of the Chern number carried by each band.
}
 \hfill $\blacktriangleleft$
\end{remark}

A less obvious interesting consequence of Lemma \ref{lemma:2d_manifA} is explored in the following  proposition.

\begin{proposition}
\label{propos:d2manif_FKMM_02}
Let $(X,\tau)$ be an involutive space which verifies Assumption \ref{ass:2.3}. The  isomorphism
\beql{eq:iso_manif_2d}
[X^\tau,\n{U}(1)]_{\Z_2}/[X,\n{U}(1)]_{\Z_2} \;\simeq\;\Z_2
\eeq
 is induced by the map
$$
\Pi\;:\;[X^\tau,\n{U}(1)]_{\Z_2}\;\to\;\Z_2\;,\qquad\quad [\omega]\;\stackrel{\Pi}{\longmapsto}\;\prod_{i=1}^{2n}\omega(x_i)\;.
$$
\end{proposition}
\proof
The isomorphism \eqref{eq:iso_manif_2d} is a consequence of the identification \eqref{eq:coker_iso1} and Proposition \ref{propos:d2manif_FKMM}. Let $\hat{\Phi}:[X^\tau,\n{U}(1)]_{\Z_2}/[X,\n{U}(1)]_{\Z_2} \to\Z_2$
be such an isomorphism and $\Phi:[X^\tau,\n{U}(1)]_{\Z_2}\to\Z_2$ a morphism which generates $\hat{\Phi}$.
Evidently, the action of $[X,\n{U}(1)]_{\Z_2}$ on $[X^\tau,\n{U}(1)]_{\Z_2}$ by restriction is given by elements which are in the kernel of $\Phi$.
 Since
$
[X^\tau,\n{U}(1)]_{\Z_2}\simeq{\rm Map}(X^\tau,\Z_2)\simeq\Z_2^{2n}
$,
where $2n$ is the number of fixed points in $X^\tau$, one has that $\Phi$ can be uniquely represented  as a map
$$
{\Phi}\;:\;\Z_2^{2n}\;\to\;\Z_2\;,\qquad\quad {\Phi}(\epsilon_1,\ldots,\epsilon_{2n})\;:=\;\prod_{i=1}^{2n}(\epsilon_i)^{\sigma_i}\;
$$
where for $\omega\in {\rm Map}(X^\tau,\Z_2)$ one defines 
$\epsilon_i:=\omega(x_i)\in\{\pm1\}$ and $\sigma_i\in\{0,1\}$. As a consequence of Lemma \ref{lemma:2d_manifA}
one has that the value of $\Phi([\phi_i])=\Phi(1,\ldots,1,-1,1,\ldots,1)\in\Z_2$ has to be independent of $i=1,\ldots,2n$ or, equivalently, has to be independent of the position of the unique negative entry $-1$ in the array $(1,\ldots,1,-1,1,\ldots,1)$. This implies that $\sigma_{1}=\ldots=\sigma_{2n}$. Finally, the requirement that $\hat{\Phi}$ has to be an isomorphism fixes $\sigma_i=1$ for all $i=1,\ldots,2n$.
\qed

\medskip

\noindent
The content of Proposition  \ref{propos:d2manif_FKMM_02} is quite relevant since, in combination with Proposition \ref{propo:FKMM-sign}, it provides a complete justification, as well as a generalization,  of the Fu-Kane-Mele formula
for the classification of \virg{Quaternionic} vector bundles over the TR-torus $\tilde{\n{T}}^2$ \cite{fu-kane-06,fu-kane-mele-07}.

\begin{theorem}[Fu-Kane-Mele formula]
\label{theo:Fu-Kane-Mele_formula}
Let $(X,\tau)$ be an involutive space which verifies Assumption \ref{ass:2.3}. The topology of a \virg{Quaternionic} vector bundle $(\bb{E},\Theta)$ over $(X,\tau)$ is completely specified by the (image of the) FKMM-invariant $\kappa(\bb{E})\in\Z_2$ given by the \emph{Fu-Kane-Mele formula}
\beql{eq:Fu-Kane-Mele_formula}
\kappa(\bb{E})\;:=\;\prod_{i=1}^{2n}\rr{d}_w(x_i)
\eeq
where $x_i\in X^\tau$ are the fixed points of $X$ and $\rr{d}_w\in {\rm Map}(X^\tau,\Z_2)$ is the sign map defined by \eqref{eq:sign_map_fkm}.
\end{theorem}
\proof
Corollary \ref{coroll:d2manif_FKMM} says that the topology of an element of ${\rm Vec}^{2m}_{\rr{Q}}(X,\tau)\simeq\Z_2$ is completely specified by  FKMM-invariant $\kappa(\bb{E})\in[X^\tau,\n{U}(1)]_{\Z_2}/[X,\n{U}(1)]_{\Z_2}
$. In view of  Proposition \ref{propo:FKMM-sign}, the class $\kappa(\bb{E})$ is represented by the sign map  $\rr{d}_w$ and the isomorphism $\Pi:[X^\tau,\n{U}(1)]_{\Z_2}/[X,\n{U}(1)]_{\Z_2}\to\Z_2$
described in Proposition \ref{propos:d2manif_FKMM_02} applied to $\kappa(\bb{E})=[\rr{d}_w]$
gives rise to  formula \eqref{eq:Fu-Kane-Mele_formula}.
\qed

\medskip

\noindent
We point out  that we made a slight abuse of notation in equation \eqref{eq:Fu-Kane-Mele_formula} where more correctly we should write $\Pi(\kappa(\bb{E}))$ instead $\kappa(\bb{E})$. Nevertheless, since there is no risk of confusion,
we prefer to write the Fu-Kane-Mele formula in the simplest and more evocative form \eqref{eq:Fu-Kane-Mele_formula}.

\subsection{Classification for TR-tori}

Let us start with the case $d=2$.

\begin{proposition}
\label{propos_classif_T_2d}
Let  $\tilde{\n{T}}^2\equiv({\n{T}}^2,\tau)$
be the  TR-involutive space described in Definition \ref{def:AII_sys}. Then, for all $m\in\N$ there is a group isomorphism
$$
{\rm Vec}^{2m}_\rr{Q}({\n{T}}^2,\tau)\;\stackrel{\kappa}{\simeq}\;H^{2}_{\Z_2}\big(\tilde{\n{T}}^2|(\tilde{\n{T}}^2)^\tau,\n{Z}(1)\big)\;\simeq\;\Z_2
$$
which is provided by the FMMM-invariant $\kappa$.
\end{proposition}

\medskip

\noindent
Although this result is only a special case of Corollary \ref{coroll:d2manif_FKMM}, we find instructive to show
an explicit realization of a non-trivial rank 2 $\rr{Q}$-bundle over ${\n{T}}^2$. The existence of a non-trivial element and the injectivity of $\kappa$ (\cf Theorem \ref{theorem:inv_inject2}) provide  a complete justification of Proposition \ref{propos_classif_T_2d} along the same line of the proof of 
Proposition \ref
{propos_classif_low_S}. A direct way to produce a non trivial $\rr{Q}$-bundle of rank 2 is to start with a line bundle over $\n{T}^2$ with (first) Chern number $C=1$. Then, the construction  \eqref{eq:classif_2d_summ}  provides the required result in view of Corollary \ref{corol:Fu-Kane-Mele_formula}.
A different way is to start with $\bb{E}'\equiv\n{T}^2\times\C^2$ and to introduce a $\rr{Q}$-structure $\Theta$
by means of a map $w':\n{T}^2\to\n{U}(2)$
which verifies the relation \eqref{eq:mat_W2}.  As usual we parametrize points of $\n{T}^2\simeq\R^2/(2\pi\Z)^2$ with pairs $z:=(\theta_1,\theta_2)\in[-\pi,\pi]^2$. With this choice the involution $\tau$ acts as $\tau(\theta_1,\theta_2)=(-\theta_1,-\theta_2)$ and the four (distinct) fixed points are $z_1:=(0,0)$, $z_2:=(0,\pi)$, $z_3:=(\pi,0)$ and $z_4:=(\pi,\pi)$. A simple way to introduce a non-trivial $\rr{Q}$-structure   on $\bb{E}'$
is to construct an equivariant map $\pi:\tilde{\n{T}}^2\to \tilde{\n{S}}^2$ such that $\pi(z_1)=\pi(z_2)=\pi(z_3)=(+1,0,0)$ and $\pi(z_4)=(-1,0,0)$ and to identify $\bb{E}'$ with the pullback $\pi^*\bb{E}$ where 
$\bb{E}\equiv\n{S}^2\times\C^2$ is the non-trivial $\rr{Q}$-bundle over $\tilde{\n{S}}^2$ with $\rr{Q}$-structure $w$ given by \eqref{eq:ex_d2_w}. In this case the $\rr{Q}$-structure  on $\bb{E}'$ is simply given by $w'=\pi^*w=w\circ \pi$
and 
we can compute the FKMM-invariant of this $\rr{Q}$-bundle by using the Fu-Kane-Mele formula \eqref{eq:Fu-Kane-Mele_formula}, \ie
$$
\kappa(\bb{E}')\;:=\;\prod_{i=1}^{4}\rr{d}_{w'}(z_i)\;=\;\prod_{i=1}^{4}\rr{d}_{w}(\pi(z_i))\;=\;-1.
$$

\begin{remark}[Smash product construction]\label{rk:smash_prod2d}{\upshape
A concrete realization of an equivariant map $\pi:\tilde{\n{T}}^2\ni(\theta_1,\theta_2)\mapsto (k_0,k_1,k_2)\in\tilde{\n{S}}^2$ which 
verifies the properties  required above  is given by
\beql{eq:mapT2toS2}
\begin{aligned}
k_0(\theta_1,\theta_2)\;&:=\; \frac{7+\cos(\theta_1)+\cos(\theta_2)-9\cos(\theta_1)\cos(\theta_2)}{9-\cos(\theta_1)-\cos(\theta_2)-7\cos(\theta_1)\cos(\theta_2)}\\
k_1(\theta_1,\theta_2)\;&:=\; \frac{4\sin(\theta_1)(1-\cos(\theta_2))}{9-\cos(\theta_1)-\cos(\theta_2)-7\cos(\theta_1)\cos(\theta_2)}\\
k_2(\theta_1,\theta_2)\;&:=\; \frac{4\sin(\theta_2)(1-\cos(\theta_1))}{9-\cos(\theta_1)-\cos(\theta_2)-7\cos(\theta_1)\cos(\theta_2)}\;.
\end{aligned}
\eeq
The equivariance of $\pi$ is evident from $k_0(-\theta_1,-\theta_2)=k_0(\theta_1,\theta_2)$ and $k_j(-\theta_1,-\theta_2)=-k_j(\theta_1,\theta_2)$, $j=1,2$. The common denominator in equations \eqref{eq:mapT2toS2} is well defined for all $(\theta_1,\theta_2)\neq(0,0)=z_1$ and on the three fixed points $z_2,z_3,z_4$ one verifies that $k(z_2)=k(z_3)=(+1,0,0)$ and $k(z_4)=(-1,0,0)$. Moreover, with the help of a Taylor expansion one can check 
that the  \eqref{eq:mapT2toS2} are continuous in $z_1$ with value $k(z_1)=(+1,0,0)$. Let $\s{T}_2:=\{(\theta_1,0),(0,\theta_2)\ |\ \theta_1,\theta_2\in[-\pi,\pi]\}\subset \n{T}^2$ be the one-dimensional subcomplex of  
$\n{T}^2$ consisting of two copies of $\n{S}^1$ joined together on the fixed point $z_1=(0,0)$. In the jargon of topology one says that $\s{T}_2=\n{S}^1\vee\n{S}^1$ is a \emph{wedge} sum of two circles. From \eqref{eq:mapT2toS2} it follows that $\pi(\s{T}_2)=(+1,0,0)$, hence $\pi$ corresponds to the (equivariant) projection
$$
\pi\;:\; \n{T}^2\;\longrightarrow\; \n{T}^2/(\n{S}^1\vee\n{S}^1)\;:=\; \n{S}^1\wedge\n{S}^1 \;\simeq\;\n{S}^2
$$
where the symbol $\wedge$ denotes the \emph{smash product} of two topological spaces \cite[Chapter 0]{hatcher-02}.
}
 \hfill $\blacktriangleleft$
\end{remark}

The case $d=3$ is a little more involved.

\begin{proposition}
\label{propos_classif_T_3d}
Let  $\tilde{\n{T}}^3\equiv({\n{T}}^3,\tau)$
be the  TR-involutive space described in Definition \ref{def:AII_sys}. Then, for all $m\in\N$ there is a group isomorphism
$$
{\rm Vec}^{2m}_\rr{Q}({\n{T}}^3,\tau)\;\stackrel{\kappa}{\simeq}\;H^{2}_{\Z_2}\big(\tilde{\n{T}}^3|(\tilde{\n{T}}^3)^\tau,\n{Z}(1)\big)\;\simeq\;\Z_2^4
$$
which is provided by the FMMM-invariant $\kappa$.
\end{proposition}
\proof
Theorem \ref{theorem:inv_inject2} establishes the injectivity of the group morphism $\kappa$
and we know from \eqref{eq:coker_TR_tori} the isomorphisms
$$
H^{2}_{\Z_2}\big(\tilde{\n{T}}^3|(\tilde{\n{T}}^3)^\tau,\n{Z}(1)\big)\;\simeq\; H^1_{\Z_2}((\tilde{\n{T}}^3)^\tau,\Z(1))\;/\;r\big(H^1_{\Z_2}(\tilde{\n{T}}^3,\Z(1))\big)\;\stackrel{\hat{\Phi}}{\simeq}\;\Z_2^{4}\;.
$$
We can realize  $\hat{\Phi}$ starting from a morphism $\Phi:H^1_{\Z_2}((\tilde{\n{T}}^3)^\tau,\Z(1))\to \Z_2^{4}$ which acts trivially on $r\big(H^1_{\Z_2}(\tilde{\n{T}}^3,\Z(1))\big)$. Let us build such a morphism $\Phi$.
The fixed point set $(\tilde{\n{T}}^3)^\tau\simeq(\tilde{\n{S}}^1)^\tau\times(\tilde{\n{S}}^1)^\tau\times(\tilde{\n{S}}^1)^\tau$ has eight distinct points which can be labelled with eight vectors ${\rm v}_k:=(v^1_k,v^2_k,v^3_k)\in\Z_2^3$; explicitly 
\beql{eq:fix_point_conv}
\begin{aligned}
{\rm v}_1=(+1,+1,+1)\;,&&{\rm v}_2=(+1,+1,-1)\;,&&{\rm v}_3=(+1,-1,+1)\;,&&{\rm v}_4=(-1,+1,+1)\;,\\
{\rm v}_5=(+1,-1,-1)\;,&&{\rm v}_6=(-1,+1,-1)\;,&&{\rm v}_7=(-1,-1,+1)\;,&&{\rm v}_8=(-1,-1,-1)\;.
\end{aligned}
\eeq
The presence of eight fixed points implies $H^1_{\Z_2}((\tilde{\n{T}}^3)^\tau,\Z(1))\simeq\Z_2^8$ 
\cite[eq. (5.7)]{denittis-gomi-14} and with the help of the  recursive relations \cite[eq. (5.9)]{denittis-gomi-14}
one obtains ${H}^1_{\Z_2}\big(\tilde{\n{T}}^3,\Z(1)\big)\simeq \Z_2\oplus\Z^3$.
The isomorphism ${H}^1_{\Z_2}\big(\tilde{\n{T}}^3,\Z(1)\big)\simeq[\tilde{\n{T}}^3,\n{U}(1)]_{\Z_2}$
shows that the $\Z_2$-summand  is generated by the constant map $\epsilon:\tilde{\n{T}}^3\to - 1$ while the $\Z^3$-summand is spanned by the  three (canonical) projections 
$
\pi_j:\tilde{\n{T}}^3\to \tilde{\n{S}}^1\simeq\n{U}(1)
$.
More precisely, the equality $\tilde{\n{T}}^3=\tilde{\n{S}}^1\times\tilde{\n{S}}^1\times\tilde{\n{S}}^1$ allows us
to label each $k\in \tilde{\n{T}}^3$ as $k:=(z_1,z_2,z_3)$ with $z_j\in {\n{U}(1)}$. Let $z_+=+1$ be one of the two invariant points of $\n{U}(1)\simeq\tilde{\n{S}}^1$ (with respect to the involution given by the complex conjugation). With this notation the  projections $\pi_j$ act as 
$$
\begin{aligned}
\tilde{\n{T}}^3\;\ni\;(z_1,z_2,z_3)\;\stackrel{\pi_1}{\longrightarrow}\;(z_1,+1,+1)\;\equiv\;z_1\;\in\;\tilde{\n{S}}^1\;,\\
\tilde{\n{T}}^3\;\ni\;(z_1,z_2,z_3)\;\stackrel{\pi_2}{\longrightarrow}\;(+1,z_2,+1)\;\equiv\;z_2\;\in\;\tilde{\n{S}}^1\;,\\
\tilde{\n{T}}^3\;\ni\;(z_1,z_2,z_3)\;\stackrel{\pi_3}{\longrightarrow}\;(+1,+1,z_3)\;\equiv\;z_3\;\in\;\tilde{\n{S}}^1\;.\\
\end{aligned}
$$
The map $r:H^1_{\Z_2}(\tilde{\n{T}}^3,\Z(1))\to H^1_{\Z_2}((\tilde{\n{T}}^3)^\tau,\Z(1))\simeq\Z_2^8$ coincides with the evaluations on the fixed point set, hence we have that the image
$r\big(H^1_{\Z_2}(\tilde{\n{T}}^3,\Z(1))\big)$ is generated by the four linearly independent vectors $\phi_0:=r(\epsilon)=(-1,\ldots,-1)$ and $\phi_j:=r(\pi_j):=(\pi_j({\rm v}_1),\ldots,\pi_j({\rm v}_{8}))$, $j=1,2,3$. 

Let $\varphi\in H^1_{\Z_2}((\tilde{\n{T}}^3)^\tau,\Z(1))$ be represented as a map $\varphi:\{{\rm v}_{1},\ldots,{\rm v}_{8}\}\to\Z_2$ and consider the morphism $\Phi: H^1_{\Z_2}((\tilde{\n{T}}^3)^\tau,\Z(1))\to\Z_2^4$ defined  by $\Phi:\varphi\mapsto(\Phi_0(\varphi),\Phi_1(\varphi),\Phi_2(\varphi),\Phi_3(\varphi))$ 
$$
\begin{aligned}
\Phi_1(\varphi)\;:=\;\varphi({\rm v}_1)\; \varphi({\rm v}_2)\; \varphi({\rm v}_3)\; \varphi({\rm v}_5)\;,&\qquad& \Phi_2(\varphi)\;:=\;\varphi({\rm v}_1)\; \varphi({\rm v}_2)\; \varphi({\rm v}_4)\; \varphi({\rm v}_6)\;, \\
\Phi_3(\varphi)\;:=\;\varphi({\rm v}_1)\; \varphi({\rm v}_3)\; \varphi({\rm v}_4)\; \varphi({\rm v}_7)\;,&\qquad&\Phi_0(\varphi)\;:=\;\varphi({\rm v}_2)\; \varphi({\rm v}_3)\; \varphi({\rm v}_4)\; \varphi({\rm v}_8)\;.  \\
\end{aligned}
$$
By construction $\Phi:\phi_j\mapsto (+1,+1,+1,+1)$ for all $j=0,\ldots,3$, hence $r\big(H^1_{\Z_2}(\tilde{\n{T}}^3,\Z(1))\big)$ is in the kernel of $\Phi$. Moreover the images of 
$$
\begin{aligned}
\psi_1\;:=\;(-1,+1,+1,-1,+1,+1,+1,+1)\;,&\qquad& \psi_2\;:=\;(-1,+1,-1,+1,+1,+1,+1,+1)\;,\\
\psi_3\;:=\;(-1,-1,+1,+1,+1,+1,+1,+1)\;,&\qquad& \psi_0\;:=\;(+1,+1,+1,+1,+1,+1,+1,-1)\;,
\end{aligned}
$$
under the map $\Phi$ are 
$$
\begin{aligned}
\Phi(\psi_1)\;=\;(-1,-1,+1,+1)\;,&\qquad&\Phi(\psi_2)\;=\;(-1,+1,-1,+1)\;,\\
\Phi(\psi_3)\;=\;(-1,+1,+1,-1)\;,&\qquad&\Phi(\psi_0)\;=\;(-1,+1,+1,+1)\;,
\end{aligned}
$$
showing that the morphism $\Phi$ induces the isomorphism $\hat{\Phi}$.

 In order to finish the proof, we need only to show the existence of $\rr{Q}$-bundles $\bb{E}_j$ over $\tilde{\T}^3$ such that $\kappa(\bb{E}_j)$ is represented by $\psi_j$ for $j=0,1,\ldots,3$. Let $\bb{E}\to \tilde{\T}^2$ be a non-trivial $\rr{Q}$-bundle classified by $\kappa(\bb{E})=[\rr{d}_{\bb{E}}]$ where the sign function $\rr{d}_{\bb{E}}:(\tilde{\T}^2)^\tau\to\Z_2$  takes values $\rr{d}_{\bb{E}}({\rm w}_1)=-1$, $\rr{d}_{\bb{E}}({\rm w}_s)=+1$, $s=2,3,4$ on the four fixed points  
$$
\begin{aligned}
{\rm w}_1=(+1,+1)\;,&&{\rm w}_2=(+1,-1)\;,&&{\rm w}_3=(-1,+1)\;,&&{\rm w}_4=(-1,-1)\;,
\end{aligned}
$$
of the set $(\tilde{\n{T}}^2)^\tau\simeq(\tilde{\n{S}}^1)^\tau\times(\tilde{\n{S}}^1)^\tau$. Let us consider the three projections
$$
\begin{aligned}
\tilde{\n{T}}^3\;\ni\;(z_1,z_2,z_3)\;\stackrel{\pi_{23}}{\longrightarrow}\;(+1,z_2,z_3)\;\equiv\;(z_2,z_3)\;\in\;\tilde{\n{T}}^2\;,\\
\tilde{\n{T}}^3\;\ni\;(z_1,z_2,z_3)\;\stackrel{\pi_{13}}{\longrightarrow}\;(z_1,+1,z_3)\;\equiv\;(z_1,z_3)\;\in\;\tilde{\n{T}}^2\;,\\
\tilde{\n{T}}^3\;\ni\;(z_1,z_2,z_3)\;\stackrel{\pi_{12}}{\longrightarrow}\;(z_1,z_2,+1)\;\equiv\;(z_1,z_2)\;\in\;\tilde{\n{T}}^2\;.\\
\end{aligned}
$$
The $\rr{Q}$-bundle $\pi_{23}^\ast\bb{E}\to \tilde{\n{T}}^3$ has the FKMM-invariant $\kappa(\pi_{23}^\ast\bb{E})$ which is represented by the map $\pi_{23}^\ast\rr{d}_{\bb{E}}:=\rr{d}_{\bb{E}}\circ \pi_{23}$. Observing that 
$$
\begin{aligned}
\pi_{23}({\rm v}_1)=\pi_{23}({\rm v}_4)={\rm w}_1\;,&\qquad&\pi_{23}({\rm v}_2)=\pi_{23}({\rm v}_6)={\rm w}_2\;,\\
\pi_{23}({\rm v}_3)=\pi_{23}({\rm v}_7)={\rm w}_3\;,&\qquad&\pi_{23}({\rm v}_5)=\pi_{23}({\rm v}_8)={\rm w}_4\;,
\end{aligned}
$$
one verifies that $\pi_{23}^\ast\rr{d}_{\bb{E}}=\psi_1$ and so $\pi_{23}^\ast\bb{E}$ is a representative for a $\rr{Q}$-bundle $\bb{E}_1$ with FKMM-invariant $[\psi_1]$.  In the same way, one checks that
$\pi_{13}^\ast\rr{d}_{\bb{E}}=\psi_2$ and $\pi_{12}^\ast\rr{d}_{\bb{E}}=\psi_3$ in such a way that 
$\pi_{13}^\ast\bb{E}$ is a representative for a $\rr{Q}$-bundle $\bb{E}_2$ with FKMM-invariant $[\psi_2]$
and $\pi_{12}^\ast\bb{E}$ is a representative for a $\rr{Q}$-bundle $\bb{E}_3$ with FKMM-invariant $[\psi_3]$.
To construct the $\rr{Q}$-bundle $\bb{E}_0$ classified by $[\psi_0]$
we  consider the projection $\pi_0:\tilde{\n{T}}^3\to \tilde{\n{S}}^3$ defined by the standard \emph{smash product} construction \cite[Chapter 0]{hatcher-02}
$$
\tilde{\n{S}}^3\;\simeq\;\tilde{\n{T}}^3/\s{T}_3
$$
which produces the sphere $\tilde{\n{S}}^3$ from the torus $\tilde{\n{T}}^3$ collapsing the subcomplex $\s{T}_3\subset \tilde{\n{T}}^3$ to the fixed point $k_\ast\equiv(+1,+1,+1)\in\s{T}_3$ (this is the same construction described in Remark \ref{rk:smash_prod2d} for the case of a two-torus). More precisely,  $\s{T}_3:=\pi_{23}(\tilde{\n{T}}^3)\cup\pi_{13}(\tilde{\n{T}}^3)\cup\pi_{12}(\tilde{\n{T}}^3)$
is a collection of three 2-tori  such that their common intersection is the fixed point $k_\ast$ and each two of them intersect along a circle. This  construction is compatible  with the definition of the involution $\tau$. Let $p_\pm:=(\pm1,0,0,0)$ be the two fixed points of   $\tilde{\n{S}}^3$ and $\bb{E}'\to \tilde{\n{S}}^3$ the non-trivial $\rr{Q}$-bundle classified by a sign map $\rr{d}_{\bb{E}'}(p_\pm)=\pm1$. By construction  $\pi_0({\rm v}_j)=p_+$ for all $j=1,\ldots,7$ but $\pi_0({\rm v}_8)=p_-$ and so $\pi_{0}^\ast\rr{d}_{\bb{E}'}:=\rr{d}_{\bb{E}'}\circ \pi_{0}$ coincides with $\psi_0$. Then,  $\pi_{0}^\ast\bb{E}'$ is a model  for $\bb{E}_0$.\qed

\begin{remark}[Week and strong invariants]\label{rk:FKMM_week_strong}{\upshape
According to the Proposition \ref{propos_classif_T_3d}, each \virg{Quaternionic} vector bundle $(\bb{E},\Theta)$ over $\tilde{\T}^3$ is specified by a  quadruple $(\kappa_0(\bb{E}),\kappa_1(\bb{E}),\kappa_2(\bb{E}),\kappa_3(\bb{E}))\in\Z^4_2$ which provides a representative for the FKMM-invariant $\kappa(\bb{E})$.
The proof of  Proposition \ref{propos_classif_T_3d}, together with the group structure of ${\rm Vec}^{2m}_{\rr{Q}}(\n{T}^3,\tau)$ described in Theorem \ref{theorem:inv_inject2},
gives us also a recipe to compute these numbers using the sign function $\rr{d}_{\bb{E}}:(\tilde{\T}^3)^\tau\to\Z_2$ associated with the $\rr{Q}$-bundle $\bb{E}$.
Let us use the convention \eqref{eq:fix_point_conv} for the fixed point of $\tilde{\T}^3$.
The first three invariants
$$
\kappa_1(\bb{E})\;:=\;\prod_{j\in\{1,2,3,5\}}\;\rr{d}_{\bb{E}}({\rm v}_j)\;,\qquad \kappa_2(\bb{E})\;:=\;\prod_{j\in\{1,2,4,6\}}\;\rr{d}_{\bb{E}}({\rm v}_j)\;,\qquad\kappa_3(\bb{E})\;:=\;\prod_{j\in\{1,3,4,7\}}\;\rr{d}_{\bb{E}}({\rm v}_j)
$$
can be understood as follows: if one considers, for instance, the restricted bundle $\bb{E}|_{\pi_{23}(\tilde{\T}^3)}$
over the two-dimensional torus $\pi_{23}(\tilde{\T}^3)\subset \tilde{\T}^3$ one has that $\kappa_1(\bb{E})$
coincides with the FKMM-invariant of $\bb{E}|_{\pi_{23}(\tilde{\T}^3)}$. Similarly, $\kappa_2(\bb{E})$
and $\kappa_3(\bb{E})$ are the FKMM-invariants of the restrictions $\bb{E}|_{\pi_{13}(\tilde{\T}^3)}$ and $\bb{E}|_{\pi_{23}(\tilde{\T}^3)}$, respectively. These three numbers are called \emph{week} invariants in the jargon of \cite{fu-kane-mele-07} (\cf equation (2), in particular) since they express only two-dimensional properties of the system. From a topological point of view, these three invariants describe the obstruction to extending a $\rr{Q}$-frame from the (equivariant) 1-skeleton of $\tilde{\T}^3$ to its 2-skeleton.
The fourth invariant
$$
\kappa_0(\bb{E})\;:=\;\prod_{j=1}^8\;\rr{d}_{\bb{E}}({\rm v}_j)
$$
is called \emph{strong} \cite[eq. (1)]{fu-kane-mele-07} since it expresses a genuine three-dimensional property of the system. This number can be understood as follows : if the weak invariants are trivial the   $\rr{Q}$-bundle $\bb{E}$ is trivial when restricted to the subcomplex  $\s{T}_3$. Hence $\bb{E}$ is $\rr{Q}$-isomorphic to a $\rr{Q}$-bundle over the sphere $\tilde{\n{S}}^3\simeq\tilde{\n{T}}^3/\s{T}_3$ \cite[Lemma 1.47]{atiyah-67}
with FKMM-invariant $\kappa_0(\bb{E})$. Then, the strong invariant provides the obstruction to extending a  $\rr{Q}$-frame from the subcomplex  $\s{T}_3$ to the full torus $\tilde{\n{T}}^3$.}
 \hfill $\blacktriangleleft$
\end{remark}

\section{Classification in dimension $d=4$}
\label{sec:d=4}
In this section we provide a classification for \virg{Quaternionic} vector bundles over  $\tilde{\n{S}}^4$ and $\tilde{\n{T}}^4$. We prove some general preliminary results  under the following rather restrictive hypothesis:
\begin{assumption}\label{ass:3.1}
Let $(X,\tau)$  be an involutive space such that:
\begin{enumerate}
\item[0.]  $X$ is a closed (compact without boundary) and  oriented $4$-dimensional manifold;\vspace{1.3mm}
\item[1.] The fixed point set $X^\tau=\{x_1,\ldots,x_N\}$  consists of a finite number of points with $N\geqslant 2$;
\vspace{1.3mm}
\item[2.]  The involution $\tau:X\to X$ is smooth and preserves the orientation;\vspace{1.3mm}
\item[3.]  $H_{\Z_2}^2(X,\Z(1))=0$;\vspace{1.3mm}
\item[4.]  There exists a fixed point $x_\ast\in X^\tau$
such that
the space $X_\ast:=X\setminus\{x_\ast\}$ is $\Z_2$-homotopic to a $\Z_2$-CW complex whose cells are of dimension less than 4.
\end{enumerate}
\end{assumption}
\noindent
 First of all we notice that under the above conditions $(X,\tau)$ turns out to be an FKMM-space as in  Definition \ref{ass:2}. Second, the reduced space $X_\ast=X\setminus\{x_\ast\}$ is also an FKMM-space which verifies the conditions of Theorem \ref{theorem:inv_inject2}, and so the map
$$
\kappa:{\rm Vec}^{2m}_\rr{Q}(X_\ast,\tau)\;\longrightarrow\; H^{2}_{\Z_2}(X_\ast|X_\ast^\tau,\n{Z}(1))\;\simeq\;[X_\ast^\tau,\n{U}(1)]_{\Z_2}/[X_\ast,\n{U}(1)]_{\Z_2}
$$
is injective. The fact that  $X_\ast$ is an FKMM-space depends on the following lemma:
\begin{lemma}
\label{lemma:D4_01}
Under Assumption \ref{ass:3.1}
$$
H_{\Z_2}^2(X_\ast,\Z(1))=0\;.
$$
\end{lemma}
\proof[
sketch of]
As a consequence of the so-called \emph{slice Theorem} \cite[Chapter I, Section 3]{hsiang-75}  a neighborhood of $x_\ast$ can be identified with an open subset around the origin of $\R^4$ endowed with the involution given by the reflection $x\mapsto -x$. Let $D$ be the unit (closed) ball under this identification. Then $D$ is an invariant
set in which $x_\ast$ is the only fixed point.
The space $X':=X\setminus{\rm Int}(D)$ is $\Z_2$-homotopy equivalent to $X_\ast$. Moreover, $D$ is $\Z_2$-contractible and so with the help of the  Meyer-Vietoris exact sequence for $\{ D, X' \}$ one can prove the isomorphism $H_{\Z_2}^2(X_\ast,\Z(1))\simeq H_{\Z_2}^2(X,\Z(1))$
which concludes the proof.
\qed

\medskip

The following technical lemma will provide us  important information which  turn out to be equivalent to the classification of $\rr{Q}$-bundles over $\tilde{\n{S}}^4$. We recall, following the same notation  of Corollary \ref{corollary:classification_on_3_sphere}, that  $\hat{\n{S}}^3\equiv(\n{S}^3,\vartheta)$ is the three-sphere endowed with the antipodal action $\vartheta:k\mapsto -k$. Moreover, we need also the well-known isomorphism $[\hat{\n{S}}^3,\n{U}(2)] \;\simeq\; \pi_3(\n{U}(2))\;\simeq\;\Z$ given by the \emph{topological degree}.
\begin{lemma}
\label{lemma:D4_05}
Let $\hat{\n{S}}^3$ be the sphere endowed with the antipodal action $k\mapsto -k$ and $[\hat{\n{S}}^3,\n{U}(2)]_{\Z_2}$ the set of $\Z_2$ homotopy equivalent maps with respect to the involution $\mu$ on $\n{U}(2)$ described in Remark \ref{rk:quat_str}.
 Then the  map  $[\hat{\n{S}}^3,\n{U}(2)]_{\Z_2} {\to} [\hat{\n{S}}^3,\n{U}(2)]$ defined by \virg{forgetting} the
involutive structures is an isomorphism. In particular, this leads to a group isomorphism 
$$
\deg \;:\;  [\hat{\n{S}}^3,\n{U}(2)]_{\Z_2}\;\longrightarrow\;
\Z
$$
induced by the topological degree.
\end{lemma}

The  proof of this Lemma is based on the following idea: 
{The determinant $\det : \n{U}(2) \to \n{U}(1)$ induces a homomorphism
\begin{equation}\label{eq:det_S3-S2}
\det : \ [\hat{\n{S}}^3, \n{U}(2)]_{\Z_2} \to [\hat{\n{S}}^3, \n{U}(1)]_{\Z_2}
\simeq H^1_{\Z_2}(\hat{\n{S}}^3, \Z(1)) \simeq \Z_2, 
\end{equation}
where $H^1_{\Z_2}(\hat{\n{S}}^3, \Z(1)) \simeq \Z_2$ has been proved in \cite{denittis-gomi-14}. This group is generated by the constant maps with values $\{\pm 1\}  \subset \n{U}(1)$.
One has the group morphism}
\begin{equation}\label{eq:S3U2_bis}
(\det, \deg) : \ [\hat{\n{S}}^3, \n{U}(2)]_{\Z_2} \longrightarrow
\Z_2 \oplus \Z
\end{equation}
which associates to each $[\varphi]\in [\hat{\n{S}}^3, \n{U}(2)]_{\Z_2}$ the pair $(\epsilon,n)\in \Z_2 \oplus \Z$ where $\epsilon=[\det\varphi]$ and $n=\deg \varphi$. This morphism  turns out to be injective and so it defines an isomorphism between $[\hat{\n{S}}^3, \n{U}(2)]_{\Z_2}$ and a subgroup of $\Z_2 \oplus \Z$ described by
\begin{equation}\label{eq:S3U2}
[\hat{\n{S}}^3, \n{U}(2)]_{\Z_2}\;\stackrel{(\det, \deg)}{\simeq}\{ (\epsilon, n) \in \Z_2 \oplus \Z\ |\ \epsilon = (-1)^n \}\; \simeq\; \Z\;.
\end{equation}
From \eqref{eq:S3U2}  it results evident that the degree completely classifies classes in $[\hat{\n{S}}^3, \n{U}(2)]_{\Z_2}$.

\proof[
of Lemma \ref{lemma:D4_05}]
Let us start by proving that \eqref{eq:S3U2_bis} is injective.
 To this end, we recall the $\Z_2$-CW decomposition of $\hat{\n{S}}^3$ (\cf  \cite[Section 4.5]{denittis-gomi-14}):
$$
\hat{\n{S}}^3 = \tilde{{\bf e} }_0 \cup \tilde{{\bf e}}_1 \cup \tilde{{\bf e}}_2 \cup \tilde{{\bf e}}_3,
$$
where $\tilde{{\bf e}}_d = \Z_2 \times {\bf e}_d = {\bf e}_d^+ \sqcup {\bf e}_d^-$ is a free $\Z_2$-cell, on which the involution acts by exchanging the usual $d$-dimensional cells ${\bf e}_d^\pm$. Let $X_2 = \tilde{{\bf e}}_0 \cup \tilde{{\bf e}}_1 \cup \tilde{{\bf e}}_2$ be the $2$-skeleton of $\hat{\n{S}}^3$, which provides also the  $\Z_2$-CW decomposition of the  sphere $\hat{\n{S}}^2$ with the antipodal free involution. Now, let $\varphi : \hat{\n{S}}^3 \to \n{U}(2)$ be  a $\Z_2$-equivariant map   such that $\epsilon=[\det \varphi] = 1$ in $[\hat{\n{S}}^3, \n{U}(1)]_{\Z_2} \simeq \Z_2$ and $n=\deg \varphi = 0$ in $[\hat{\n{S}}^3, \n{U}(2)] \simeq \Z$. The first assumption leads to $[\det \varphi|_{X_2}] = 1$ in $[X_2, \n{U}(1)]_{\Z_2}$, hence Corollary \ref{corollary:classification_on_3_sphere} 
assures the existence of  an equivariant homotopy between $\varphi|_{X_2}$ and the constant map at   $\n{1}_2 \in \n{U}(2)$. By applying the equivariant homotopy extension property \cite{matumoto-71} to the subcomplex $X_2 \subset \hat{\n{S}}^3$, we get an equivariant homotopy from $\varphi$ to an equivariant map $\varphi' : \hat{\n{S}}^3 \to \n{U}(2)$ such that $\varphi'|_{X_2} = \n{1}_2$. Then, the restriction $\varphi'|_{{\bf e}_3^+}$ 
is identifiable with an element $[\varphi'|_{{\bf e}_3^+}] \in \pi_3(\n{U}(2))$. The involution on $\hat{\n{S}}^3$ preserves the orientation, and the involution $\mu$ on $\n{U}(2)$ induces the identity on $\pi_3(\n{U}(2))$. As a result, $[\varphi'|_{{\bf e}_3^+}] = [\varphi'|_{{\bf e}_3^-}]$ holds true and $[\varphi] \in \pi_3(\n{U}(2))$ is expressed as $[\varphi] = 2 [\varphi'|_{{\bf e}_3^+}]$. Because $\pi_3(\n{U}(2)) \simeq \Z$ has no torsion $[\varphi]=0$ (which is the same of $\deg \varphi = 0$) implies $[\varphi'|_{e_3^+}] = 0$ in $\pi_3(\n{U}(2)) $. Then, there exists   a homotopy from $\varphi'|_{{\bf e}_3^+}$ to the constant map at $\n{1}_2$. By means of the involution on $\hat{\n{S}}^3$, this homotopy extends to an equivariant homotopy from $\varphi'|_{\tilde{\bf e}_3}$ to the constant map at $\n{1}_2$. Again by the equivariant homotopy extension property, we  get an equivariant homotopy from $\varphi'$ to the constant map at $\n{1}_2$.  In summary, we proved that if $(\det, \deg):[\varphi]\mapsto (1,0)$ then $[\varphi]$ is $\Z_2$-homotopy equivalent to the constant map at $\n{1}_2$, namely  the injectivity of $(\det, \deg)$. 
It remains to prove that the image of $(\det, \deg)$ agrees with the subgroup \eqref{eq:S3U2}. We point out that
from the previous argument it follows that
 there is no equivariant map $\varphi : \hat{\n{S}}^3 \to \n{U}(2)$ such that $\epsilon=n = 1$. In fact we already proved that  $\epsilon=[\det\varphi] = 1$ implies that $\varphi$ is of even degree, \ie $n \in 2\Z$.
Finally, the equivariant map $\varphi : \hat{\n{S}}^3 \to \n{U}(2)$ defined by
$$
\varphi(k_1, k_2, k_3, k_4)\;:
=\;
\ii
\left(
\begin{array}{cc}
k_1 + \ii k_2 & -k_3 + \ii k_4 \\
k_3 + \ii k_4 & k_1 - \ii k_2
\end{array}
\right)\;
$$
clearly has $\epsilon=\det\varphi = -1$ and $n=\deg \varphi = 1$ and so it provides a generator for \eqref{eq:S3U2}. 
\qed

\medskip

We are now in  position to prove Theorem \ref{theo:injectD4} under Assumption \ref{ass:3.1}. Before going through the technical part of the proof of this theorem, let us point out that Theorem \ref{theo:injectD4} is true under the hypothesis as stated, which is less restrictive than Assumption \ref{ass:3.1}. The generalized proof makes strong use of the \emph{obstruction theory} and this implies a considerable increasing of technical complications. 
Since the aim of this work is to provide a classification 
in the case of the spaces  $\tilde{\n{S}}^4$ and $\tilde{\n{T}}^4$
(which verify Assumption \ref{ass:3.1}) we opted here for a simplified proof, 
leaving the more general case for a future work \cite{denittis-gomi-15}.

\begin{theorem}[Injective group homomorphism: d=4]
\label{theo:injectD4_weak_version}
Let $(X,\tau)$ be as in Assumption \ref{ass:3.1}. 
Then, the FKMM-invariant $\kappa$ and the second Chern class  $c_2$ define a map
\begin{equation}
(\kappa,c_2)\;:\;{\rm Vec}^{2m}_\rr{Q}(X,\tau)\;\longrightarrow\; H^{2}_{\Z_2}(X|X^\tau,\n{Z}(1))\;\oplus\; H^4(X,\Z)\;\qquad\qquad m\in\N
\end{equation}
that is injective. Moreover, ${\rm Vec}^{2m}_\rr{Q}(X,\tau)$ can be endowed with a group structure in such a way that the pair $(\kappa,c_2)$ sets  an injective group homomorphism. 
\end{theorem}

\proof
The stable range condition implies 
$
{\rm Vec}^{2m}_{\rr{Q}}(X, \tau)\simeq {\rm Vec}^2_{\rr{Q}}(X, \tau)
$
(\cf Corollary \ref{corol_(iii)}).
Moreover, if $\bb{E}\simeq\bb{E}_0\oplus(X\times\C^{2(m-1)})$ as $\rr{Q}$-bundles we obtain from 
Theorem \ref{theo:FKMM_propert} (iii) that $\kappa(\bb{E})=\kappa(\bb{E}_0)$ and from the usual properties of Chern classes that $c_2(\bb{E})=c_2(\bb{E}_0)$. Then, without loss of generality, we can consider only the case $m=1$.

 As in the proof of Lemma \ref{lemma:D4_01}, we set
 $X':=X\setminus{\rm Int}(D)$ 
where $D$ is an invariant ball around the fixed point  $x_\ast\in X^\tau$.
The space $X'$ is $\Z_2$-homotopy equivalent to $X_\ast$ and so, by assumption, it is equivalent to a $\Z_2$-CW complex whose cells are of dimension less or equal to 3.
 By assumption also $(X')^\tau\neq\emptyset$. The inclusion map $\imath: X'\hookrightarrow X$ induces by the naturality of the FKMM-invariant the equalities $\kappa(\bb{E}_i|_{X'})=\kappa(\imath^*\bb{E}_i)=\imath^*\kappa(\bb{E}_i)$, $i=1,2$. Then the coincidence of the FKMM-invariants 
 of $\bb{E}_1$ and $\bb{E}_2$, immediately implies $\kappa(\bb{E}_1|_{X'})=\kappa(\bb{E}_2|_{X'})$ and  Theorem \ref{theorem:inv_inject2} assures the existence of a $\rr{Q}$-isomorphism
$
f': \bb{E}_1|_{X'}\to\bb{E}_2|_{X'}
$.
Since $D$ is $\Z_2$-contractible there are also two $\rr{Q}$-isomorphisms $g_{i}:\bb{E}_i|_{D}\to D\times\C^2$, $i=1,2$. Let $Y:=\partial X'=\partial D$. This space is $\Z_2$-homotopic to a three-sphere $\hat{\n{S}}^3$ endowed with the {antipodal} involution. The composition of $\rr{Q}$-isomorphisms
$$
Y\times\C^2\;\stackrel{g_{1}^{-1}|_{Y}}{\longrightarrow}\; \bb{E}_1|_{Y}\;\stackrel{f'|_{Y}}{\longrightarrow}\; \bb{E}_2|_{Y}\;\stackrel{g_{2}|_{Y}}{\longrightarrow}\; Y\times\C^2
$$
can be expressed by a mapping $(x,{\rm v})\mapsto(x,\theta(x)\cdot Q\cdot\overline{\rm v})$ where $\theta:Y\to\n{U}(2)$ is a $\mu$-equivariant map in the sense of Remark \ref{rk:quat_str}, \ie
 $\theta(\tau(x))=-Q\overline{\theta(x)}Q$.
 In view of the homotopy property for \virg{Quaternionic} vector bundles (\cf Theorem \ref{theo:equi_homot1}), the map $
f': \bb{E}_1|_{X'}\to\bb{E}_2|_{X'}
$ extends to an isomorphism $f:\bb{E}_1\to \bb{E}_2$ if and only if the $\Z_2$-homotopy class $[\theta]\in[Y,\n{U}(2)]_{\Z_2}=[\hat{\n{S}}^3,\n{U}(2)]_{\Z_2}$ is in the image of $[D,\n{U}(2)]_{\Z_2}$ under the restriction morphism.
But we know that $D$ is $\Z_2$-contractible, hence $[D,\n{U}(2)]_{\Z_2}\simeq0$ and so we get a $\rr{Q}$-isomorphism 
which extends $f'$
if and only if 
$[\theta]$ is trivial, or equivalently if and only if ${\rm deg}(\theta)=0$ as a consequence of Lemma \ref{lemma:D4_05}.
In the remaining part we prove  that $c_2(\bb{E}_1)=c_2(\bb{E}_2)$ implies ${\rm deg}(\theta)=0$
 and this  concludes the proof.

We recall that $X'$ has the structure of a CW-complex of dimension less or equal to 3. Then, any complex vector bundle over  $X'$ with associated trivial determinant line bundle is automatically trivial in the complex category (\cf Proposition \ref{propos:Q_struc_lod_d}). This implies the existence of  an isomorphism of complex vector bundles
$
h_1:\bb{E}_1|_{X'}\to X'\times\C^2
$
and we can set a second isomorphism $h_2:\bb{E}_2|_{X'}\to X'\times\C^2$ by $h_2:=h_1\circ (f')^{-1}$.
For each $i=1,2$ the composition
$$
Y\times\C^2\;\stackrel{g_{i}^{-1}|_{Y}}{\longrightarrow}\; \bb{E}_i|_{Y}\;\stackrel{h_i|_{Y}}{\longrightarrow}\; Y\times\C^2
$$
can be expressed by a mapping $(x,{\rm v})\mapsto(x,\varphi_i(x) \cdot {\rm v})$ where $\varphi_i:Y\to\n{U}(2)$. This is nothing but a clutching function of the underlying complex vector bundle $\bb{E}_i$. In view of the Meyer-Vietoris exact sequence 
$$
H^3(X',\Z)\;\oplus\; H^3(D,\Z)\;\stackrel{}{\longrightarrow}\;H^3(Y,\Z)\;\stackrel{}{\longrightarrow}\;H^4(X,\Z)\;\stackrel{}{\longrightarrow}\;0
$$
(we recall $D\simeq\{\ast\}$, $Y=X'\cap D$ and $X=X'\cup D$) we obtain that $H^3(Y,\Z)\simeq H^4(X,\Z)$
which implies that the degree of $\varphi_i$ coincides with the second Chern class of $\bb{E}_i$ {under the choice of an orientation on $D$ compatible with $X$}. By construction, one readily sees  $\theta=\varphi_2\circ\varphi_1^{-1}$ which implies ${\rm deg}(\theta)={\rm deg}(\varphi_2)-{\rm deg}(\varphi_1)=c_2(\bb{E}_2)-c_2(\bb{E}_1)$.
\qed

\medskip
\medskip

\noindent{\bf The case of $\tilde{\n{S}}^4$.}
If we remove from the TR-sphere $\tilde{\n{S}}^4$ one of the two fixed points $k_\pm:=(\pm1,0,0,0)$ we obtain a space which is $\Z_2$-homotopic to $\R^3$ endowed with the involution given by the reflection around the origin $x\mapsto-x$. This space is $\Z_2$-contractible to a point and so Assumption \ref{ass:3.1} is verified. This implies that Theorem \ref{theo:injectD4} applies to $\tilde{\n{S}}^4$, although the use of this result   is not strictly necessary for the classification of 
${\rm Vec}^{2m}_\rr{Q}({\n{S}}^4,\tau)$.

\proof[
of Theorem \ref{theo:AI_classd=4} (i)]
One has the following isomorphisms
$$
{\rm Vec}^{2m}_\rr{Q}({\n{S}}^4,\tau)\;\simeq\;{\rm Vec}^{2}_\rr{Q}({\n{S}}^4,\tau)\;\simeq\; [\hat{\n{S}}^3, \n{U}(2)]_{\Z_2}\;\simeq\;\Z
$$
where the first is a consequence of Corollary \ref{corol_(iii)}, the second follows from an equivariant adaptation of the standard {clutching construction} \cite[Lemma 1.4.9]{atiyah-67} and the last one has been proved in  Lemma \ref{lemma:D4_05}. To finish the proof, one has only to recall that if $[\bb{E}_\varphi]\in {\rm Vec}^{2}_\rr{Q}({\n{S}}^4,\tau)$ is the $\rr{Q}$-bundle associated with $[\varphi]\in [\hat{\n{S}}^3, \n{U}(2)]_{\Z_2}$ then $C=\langle c_2(\bb{E}_\varphi),[\n{S}^4]\rangle=\deg \varphi=n$ and $\kappa(\bb{E}_\varphi)=[\det\varphi]=\epsilon\in\{\pm1\}$.
While the first is a classical identity in the theory of complex vector bundles, the second deserves some comments.
Let $H_\pm := \{ (k_0, k_1, k_2, k_3,k_4) \in \tilde{\n{S}}^4 |\ \pm k_0 \ge 0 \}$ with intersection $H_+ \cap H_-=\hat{\n{S}}^3$. Then, by construction
$$
\bb{E}_\varphi\;\simeq\; (H_+\;\times\;\C^2)\;\cup_{\varphi}\;  (H_-\;\times\;\C^2)
$$
The  set $(\tilde{\n{S}}^4)^\tau=\{(\pm1, 0, 0, 0,0)\}$
has two fixed points and we can set the identity map on $(\tilde{\n{S}}^4)^\tau\times\C^2$ as the trivialization of $\bb{E}_\varphi|_{(\tilde{\n{S}}^4)^\tau}$ and the  induced canonical trivialization ${\rm det}_{(\tilde{\n{S}}^4)^\tau}:{\rm det}(\bb{E}_\varphi|_{(\tilde{\n{S}}^4)^\tau})\to (\tilde{\n{S}}^4)^\tau\times\C$ which 
 obviously agrees with  the identity map on $(\tilde{\n{S}}^4)^\tau\times\C$. On the other hand, 
 $$
{\rm det (\bb{E}_\varphi})\;\simeq\; (H_+\;\times\;\C)\;\cup_{\det\varphi}\;  (H_-\;\times\;\C)
$$
is obtained by gluing together trivial \virg{Real} line bundles over $H_\pm$ by means of the  clutching function $\det\varphi$. Then,
 a global trivialization $h_{\rm det}:{\rm det (\bb{E}_\varphi})\to \tilde{\n{S}}^4\times\C$ can be realized as a pair of $\rr{R}$-isomorphisms $f_\pm: H_\pm \times \C\to H_\pm \times \C$
 such that $f_+|_{\hat{\n{S}}^3}=\det\varphi\cdot f_-|_{\hat{\n{S}}^3}$. Due to equation \eqref{eq:det_S3-S2},
 we can assume that $\det\varphi=\pm1$ as a constant map and this implies that  $f_+|_{\hat{\n{S}}^3}=\pm  f_-|_{\hat{\n{S}}^3}$, accordingly. This difference in sign is exactly the FKMM-invariant.
 \qed

\medskip
\medskip

\noindent{\bf The case of $\tilde{\n{T}}^4$.}
Let us briefly justify why the TR-torus $\tilde{\n{T}}^4$ verifies Assumption \ref{ass:3.1}. Under the identifications $\tilde{\n{T}}^4\simeq\tilde{\n{S}}^1\times\ldots\times\tilde{\n{S}}^1$ and $\tilde{\n{S}}^1\simeq\n{U}(1)$ (endowed with the complex conjugation) we can represent each point of the torus by a quadruple
$(z_1,z_2,z_3,z_4)$ with $z_j\in \n{U}(1)$, and the fixed points are labelled by choosing $z_j\in\{\pm 1\}$. Let $k_\ast:=(+1,+1,+1,+1)$. If one removes $k_\ast$ from $\tilde{\n{T}}^4$ one easily sees that the resulting space is $\Z_2$-homotopy equivalent to the union of four TR-tori $\tilde{\n{T}}^3_j:=\{(z_1,z_2,z_3,z_4)\in \tilde{\n{T}}^4\ |\ z_j=-1\}$ with common intersection at the fixed point $(-1,-1,-1,-1)$, namely
$$
\tilde{\n{T}}^4\setminus\{k_\ast\}\;\simeq\;\s{T}_4\;:=\;\tilde{\n{T}}^3_1\;\cup\;\ldots \;\cup\;\tilde{\n{T}}^3_4
$$
We notice that   $\s{T}_4$ inherits the structure of a $\Z_2$-CW complex from that of each summand $\tilde{\n{T}}^3_j$ and the dimension of the $\Z_2$-cells does not exceed $3$. Then, we showed that  Assumption \ref{ass:3.1} is verified and 
 Theorem \ref{theo:injectD4} applies to $\tilde{\n{T}}^4$.

\proof[
of Theorem \ref{theo:AI_classd=4} (ii)]
We provide here only a sketch of the proof that,   in spirit, is very close to
the proof of Proposition
\ref{propos_classif_T_3d}. The interested reader can possibly complete the details following the scheme below. 
\begin{itemize}
\item The injectivity follows from Theorem \ref{theo:injectD4}, the isomorphism  $H^{2}_{\Z_2}\big(\tilde{\n{T}}^4|(\tilde{\n{T}}^4)^\tau,\n{Z}(1)\big)\simeq\Z_2^{11}$ in equation \eqref{eq:coker_TR_tori} and the isomorphism $H^4(\n{T}^4,\Z)\simeq\Z$ induced by the pairing with the fundamental class $[\n{T}^4]\in H_4(\n{T}^4)$.
\vspace{1.3mm}
\item To prove the isomorphism one  follows the same strategy of the proof of Proposition \ref{propos_classif_T_3d}; More precisely
one shows the existence of elements in  ${\rm Vec}^{2m}_\rr{Q}({\n{T}}^4,\tau)$ (but as usual $m=1$ is enough) which provide a set of generators for the subgroup $\Z_2^{10}\oplus\Z$. 
\vspace{1.3mm}
\item Let $\epsilon=(\epsilon_1,\ldots,\epsilon_{10},\epsilon_{11})$ be the image of $\kappa(\bb{E})$ in $\Z_2^{11}$. One can realize the first six  components $\epsilon_1,\ldots,\epsilon_{6}$ by the pullback
of the (unique) non trivial element of ${\rm Vec}^{2}_\rr{Q}({\n{T}}^2,\tau)$ by means of 
the six
 projections $\pi_{ij}:\tilde{\n{T}}^4\to\tilde{\n{T}}^2$, $i,j=1,\ldots,4$,
  onto the six sub-tori of dimension 2. We refer to these components as \emph{ultra-week}. 
  \vspace{1.3mm}
 \item The next four components $\epsilon_7,\ldots,\epsilon_{10}$
 are given by the   pullback of the  non trivial element of ${\rm Vec}^{2}_\rr{Q}({\n{T}}^3,\tau)$ described by the strong FKMM-invariant $(-1,+1,+1,+1)$ by means of  the four
 projections $\pi_{i}:\tilde{\n{T}}^4\to\tilde{\n{T}}^3$, $i=1,\ldots,4$,
  onto the four sub-tori of dimension 3.   These are the \emph{week} components of $\epsilon$.
 \vspace{1.3mm}
 \item
These 10 components $(\epsilon_1,\ldots,\epsilon_{10})$
generate
a subgroup of $\Z_2^{11}$. Moreover, by construction, all the $\rr{Q}$-bundles associated with these invariants have trivial second Chern classes. Then, we provided the generators for the first summand $\Z_2^{10}$. 
\vspace{1.3mm}
\item Let $\pi_0:\tilde{\n{T}}^4\to\tilde{\n{T}}^4/\s{T}_4\simeq
\tilde{\n{S}}^4$ be the projection obtained according to the usual collapsing construction . The pullback of the  non trivial element in ${\rm Vec}^{2m}_\rr{Q}({\n{S}}^4,\tau)$ provides a $\rr{Q}$-bundle over 
$\tilde{\n{T}}^4$ with second Chern number $C=1$ and  FKMM-invariant
 $(\epsilon_1,\ldots,\epsilon_{10},\epsilon_{11})= (0, ..., 0, 1)$.
\end{itemize}
\qed

\section{Universality of the FKMM-invariant}
\label{subsec:univ-FKMM_inv}

In this section we provide a \virg{universal} interpretation of the
 FKMM-invariant. More precisely, we  show that the FKMM-invariant associated to a  \virg{Quaternionic} vector bundle can be defined \emph{functorially} from the \emph{universal} FKMM-invariant of the classifying \virg{Quaternionic} vector bundle described in Section \ref{sect:Q-vb_homotopy}. In this sense the FKMM-invariant is a \emph{bona fide} characteristic class. 
We point out that the construction of the universal FKMM-invariant requires a generalization and an improvement of the original idea in \cite{furuta-kametani-matsue-minami-00}.

\subsection{Universal FKMM-invariant}

Theorem \ref{theo:honotopy_class_Q} says that each rank $2m$ $\rr{Q}$-bundle $(\bb{E},\Theta)$ over the involutive space $(X,\tau)$ is obtained (up to isomorphisms) as the  pullback of the 
{tautological} $2m$-plane $\rr{Q}$-bundle $(\bb{T}_{2m}^\infty,\Xi)$ over the involutive Grassmannian $\hat{G}_{2m}(\C^\infty)\equiv ({G}_{2m}(\C^\infty),\rho)$. 

\medskip

The determinat construction described in Section \ref{subsec:det_construct} applies  also to the {tautological} $\rr{Q}$-bundle $\bb{T}_{2m}^\infty$ and it defines a line bundle
\beql{eq:univ01}
\pi\;:\;{\rm det}(\bb{T}_{2m}^\infty)\;\longrightarrow\; \hat{G}_{2m}(\C^\infty)
\eeq
which is endowed with a \virg{Real} structure ${\rm det}(\Xi):{\rm det}(\bb{T}_{2m}^\infty)\to{\rm det}(\bb{T}_{2m}^\infty)$ in agreement with Lemma \ref{lemma:R_Q_det_bun}. Moreover, we recall that $\n{S}({\rm det}(\bb{T}_{2m}^\infty))\subset {\rm det}(\bb{T}_{2m}^\infty)$ denotes the circle bundle according to  the notation introduced in Remark \ref{rk:circ_bund}.
The following characterization will play an important role in the sequel.

\begin{lemma}
\label{lemma:caract_tivial}
Let $(X,\tau)$ be an involutive space such that $X$  verifies (at least) condition 0 of Definition \ref{ass:2}
 and $\varphi:X\to \hat{G}_{2m}(\C^\infty)$ a $\Z_2$-equivariant map. Then:
\begin{enumerate}
\vspace{1.3mm}
\item[(i)] 
Global \virg{Real} sections of the circle bundle $\n{S}({\rm det}(\varphi^*\bb{T}_{2m}^\infty))$ are in one-to-one correspondence with 
$\Z_2$-equivariant maps $\tilde{\varphi}:X\to \n{S}({\rm det}(\bb{T}_{2m}^\infty))$ which make the  following diagram
\beql{eq:diag1}
\begin{diagram}
& &\big(\n{S}({\rm det}(\bb{T}_{2m}^\infty)),{\rm det}(\Xi)\big)\\
&\tilde{\varphi}\;\ruTo &\dTo_\pi \\
(X,\tau) &\rTo_{\varphi} & \hat{G}_{2m}(\C^\infty) \\
\end{diagram}
\qquad\qquad (\pi\circ\tilde{\varphi}\;=\;\varphi)
\eeq
 commutative;
\vspace{1.3mm}
\item[(ii)] 
The determinant line bundle ${\rm det}(\varphi^*\bb{T}_{2m}^\infty)$ associated with $\varphi^*\bb{T}_{2m}^\infty$ is $\rr{R}$-trivial if and only if there exists a $\Z_2$-equivariant map $\tilde{\varphi}:X\to \n{S}({\rm det}(\bb{T}_{2m}^\infty))$ such that the  diagram 
 \eqref{eq:diag1} is commutative;
\end{enumerate}
\end{lemma}
\proof(i)
 Since the determinat construction is functorial one has the natural isomorphism ${\rm det}(\varphi^*\bb{T}_{2m}^\infty)\simeq\varphi^*{\rm det}(\bb{T}_{2m}^\infty)$. Let $\hat{\varphi}:\varphi^*{\rm det}(\bb{T}_{2m}^\infty)\to {\rm det}(\bb{T}_{2m}^\infty)$ be the standard morphism associated with the pullback construction; we recall that  $\hat{\varphi}$ establishes fiberwise (metric-preserving) isomorphisms $\varphi^*{\rm det}(\bb{T}_{2m}^\infty)|_x\simeq {\rm det}(\bb{T}_{2m}^\infty)|_{\varphi(x)}$.
 If $s:X\to \n{S}(\varphi^*{\rm det}(\bb{T}_{2m}^\infty))$ is a global $\rr{R}$-section then $\tilde{\varphi}:=\hat{\varphi}\circ s$ is a  $\Z_2$-equivariant map such that $\pi\circ\tilde{\varphi}=\varphi$. Vice versa,  if one has 
a $\Z_2$-equivariant map $\tilde{\varphi}:X\to \n{S}({\rm det}(\bb{T}_{2m}^\infty))$ which verifies $\pi\circ\tilde{\varphi}=\varphi$ then $s(x):=\hat{\varphi}|_{\varphi(x)}^{-1}\circ \tilde{\varphi}(x)$ defines a 
global \virg{Real} section of $\n{S}(\varphi^*{\rm det}(\bb{T}_{2m}^\infty))$. (ii) This is a consequence of the fact that the
$\rr{R}$-triviality of an   $\rr{R}$-line bundle  (endowed with an equivariant metric) is equivalent to the existence of a global \virg{Real} section of the associated circle bundle
(\cf Remark \ref{rk:circ_bund}). \qed

\begin{remark}\label{rk:ambiguity}
{\upshape
It is rather evident that if  $\tilde{\varphi}_1$ and $\tilde{\varphi}_2$ are   $\Z_2$-equivariant maps which make the diagram \eqref{eq:diag1} commutative then $\tilde{\varphi}_1=u\cdot\tilde{\varphi}_2$
for some $\Z_2$-equivariant map $u:X\to\n{U}(1)$.
}\hfill $\blacktriangleleft$
\end{remark}

In order to define a \emph{universal} FKMM-invariant we need to apply the pullback construction to the {tautological} $2m$-plane $\rr{Q}$-bundle $(\bb{T}_{2m}^\infty,\Xi)$ with respect to the map \eqref{eq:univ01} understood as an equivariant map between involutive base spaces. More precisely, we need to use the following diagram
\beql{eq:diag2}
\begin{diagram}
(\pi^*\bb{T}_{2m}^\infty,\pi^*\Xi)                             &         \rTo^{\ \ \ \ \ \ \ \ \hat{\pi}\ \ \ \ \ \ \ \ }        &(\bb{T}_{2m}^\infty,\Xi)\\
       \dTo^{\pi{''}}                                                           &                &\dTo_{\pi'} \\
\big(\n{S}({\rm det}(\bb{T}_{2m}^\infty)),{\rm det}(\Xi)\big) &\rTo_{\pi} & \hat{G}_{2m}(\C^\infty) &.\\
\end{diagram}
\eeq
Just for sake of completeness, we recall that the projection $\pi{''}$ is associated with the projection $\pi'$ by the pullback construction while the maps $\pi'$ and $\pi$ are related by the determinat construction, namely $\pi={\rm det}(\pi')$ if one prefers the functorial notation. 
The peculiar structure of the $\rr{Q}$-bundle $(\pi^*\bb{T}_{2m}^\infty,\pi^*\Xi)$ is compatible with the construction of an FKMM-invariant. First of all, we notice that the fixed point set of the base space  
$\big(\n{S}({\rm det}(\bb{T}_{2m}^\infty)),{\rm det}(\Xi)\big)$ is non-empty; More precisely we will prove in Lemma \ref{lemma:detT_topo} that 
\beql{eq:univ01.01}
\n{S}({\rm det}(\bb{T}_{2m}^\infty))^\Xi\;:=\;\{x\in\n{S}({\rm det}(\bb{T}_{2m}^\infty))\ |\ {\rm det}(\Xi)(x)=x\}\;\simeq\;{G}_{m}(\n{H}^\infty)\;\sqcup\;{G}_{m}(\n{H}^\infty)
\eeq
which shows  that $\n{S}({\rm det}(\bb{T}_{2m}^\infty))^\Xi$ decomposes as the disjoint union of two identical components that are
path-connected $\pi_0({G}_{m}(\n{H}^\infty))=0$. 
The second important ingredient is contained in the next 
result.

\begin{lemma}
\label{lemma:FKMM_type}
The rank $2m$ \virg{Quaternionic} vector bundle $(\pi^*\bb{T}_{2m}^\infty,\pi^*\Xi)$ 
over the involutive space $\big(\n{S}({\rm det}(\bb{T}_{2m}^\infty)),{\rm det}(\Xi)\big)$
is of FKMM-type according to Definition \ref{def:FKMM-type}.
\end{lemma}
\proof
The functoriality of the determinant construction implies that ${\rm det}(\pi^*\bb{T}_{2m}^\infty)\simeq \pi^*{\rm det}(\bb{T}_{2m}^\infty)$. To complete the proof we only need to show that  
$\pi^*{\rm det}(\bb{T}_{2m}^\infty)$ is $\rr{R}$-trivial and we will show this fact in two different ways.

The first proof uses Lemma \ref{lemma:caract_tivial} (ii). 
Let us consider the diagram \eqref{eq:diag1} with $\big({\n{S}({\rm det}(\bb{T}_{2m}^\infty))},{\rm det}(\Xi)\big)$ instead $(X,\tau)$
and $\varphi=\pi$. 
The identity  map ${\rm Id}:{\rm det}(\bb{T}_{2m}^\infty)\to {\rm det}(\bb{T}_{2m}^\infty)$ provides a realization for the equivariant map $\tilde{\varphi}$, hence  ${\rm det}(\pi^*\bb{T}_{2m}^\infty)$ is $\rr{R}$-trivial.

The second proof is more direct since it is based on the construction of an explicit trivialization.
The pullback construction applied to  diagram \eqref{eq:diag2} shows that
$$
\pi^*{\rm det}(\bb{T}_{2m}^\infty)\;=\;\left\{(x,y)\in\n{S}({\rm det}(\bb{T}_{2m}^\infty))\times{\rm det}(\bb{T}_{2m}^\infty)\ |\ \pi(x)\;=\;\pi(y) \right\}\;.
$$
Hence, the  \emph{diagonal section}
\beql{eq:can_diag_sec}
s_{\rm diag}\;:\;\n{S}({\rm det}(\bb{T}_{2m}^\infty))\;\longrightarrow\; \pi^*{\rm det}(\bb{T}_{2m}^\infty)\;,\qquad\qquad\quad s_{\rm diag}(x)\;:=\;(x,x)
\eeq
 provides a global equivariant section and defines a global $\rr{R}$-trivialization
$$
h_{\rm diag}\;:\;\pi^*{\rm det}(\bb{T}_{2m}^\infty)\;\longrightarrow\;\n{S}({\rm det}(\bb{T}_{2m}^\infty))\;\times\;\C\;,\qquad\qquad\quad h_{\rm diag}(x,y)\;:=\;(x,\varphi_{\rm diag}(x))
$$
where the
 equivariant map $\varphi_{\rm diag}:{\rm det}(\bb{T}_{2m}^\infty)\to \n{U}(1)$ is specified  by the condition $y=\varphi_{\rm diag}(x)\cdot x$.
\qed

\begin{remark}[Diagonal section]\label{rk:can_sec_circ_bun_univ}
{\upshape
Let us point out that 
the diagonal section $s_{\rm diag}$ in \eqref{eq:can_diag_sec} can be seen as a global $\rr{R}$-section for the circle bundle
$\n{S}(\pi^*{\rm det}(\bb{T}_{2m}^\infty))\to {\n{S}({\rm det}(\bb{T}_{2m}^\infty))}$. Moreover, it is evident that the relation
$$
s_{\rm diag}(x)\;=\;h_{\rm diag}^{-1}(x,1)\qquad\qquad\forall\ x\in \n{S}({\rm det}(\bb{T}_{2m}^\infty))
$$
holds true.
}\hfill $\blacktriangleleft$
\end{remark}

\noindent
Owing to Lemma \ref{lemma:FKMM_type},  the following definition is well posed.

\begin{definition}[Universal FKMM-invariant]
The \emph{universal FKMM-invariant} $\rr{K}_{\rm univ}$ is the FKMM-invariant of the rank $2m$ \virg{Quaternionic} vector bundle $(\pi^*\bb{T}_{2m}^\infty,\pi^*\Xi)$ defined over the involutive space $\big(
\n{S}({\rm det}(\bb{T}_{2m}^\infty)),{\rm det}(\Xi)\big)$, \ie
$$
\rr{K}_{\rm univ}\;:=\;\kappa(\pi^*\bb{T}_{2m}^\infty)\;\in\;[\n{S}({\rm det}(\bb{T}_{2m}^\infty))^\Xi,\n{U}(1)]_{\Z_2}/[\n{S}({\rm det}(\bb{T}_{2m}^\infty)),\n{U}(1)]_{\Z_2}\;.
$$ 
\end{definition}

\subsection{Naturality}

The  FKMM-invariant $\rr{K}_{\rm univ}$ is a universal \emph{characteristic class} for  ${\rm Vec}^{2m}_{{\rm FKMM}}(X,\tau)$ in the sense that it acts \emph{naturally} with respect to the homotopy classification of Theorem \ref{theo:equi_homot1}.

\begin{theorem}[Naturality of $\rr{K}_{\rm univ}$]
\label{theo:nat_univ}
Let $(\bb{E},\Theta)$ be a $\rr{Q}$-bundle of FKMM-type (\cf Definition \ref{def:FKMM-type}) over the
 involutive space $(X,\tau)$. Let $\varphi:X\to \hat{G}_{2m}(\C^\infty)$
 be the $\Z_2$-equivariant map which classifies $(\bb{E},\Theta)$ (up to isomorphisms). Then
\beql{eq:natur_K_uni}
 \kappa(\bb{E})\;=\;\tilde{\varphi}^*(\rr{K}_{\rm univ})
\eeq
where the $\Z_2$-equivariant map $\tilde{\varphi}:X\to \n{S}({\rm det}(\bb{T}_{2m}^\infty))$ verifies $\pi\circ\tilde{\varphi}=\varphi$ according to Diagram \eqref{eq:diag1}. Moreover, equality \eqref{eq:natur_K_uni}
is well defined in the sense that if $\tilde{\varphi}_1,\tilde{\varphi}_2$
are two  $\Z_2$-equivariant maps such that $\pi\circ\tilde{\varphi}_j=\varphi$, $j=1,2$ then $\tilde{\varphi}^*_1(\rr{K}_{\rm univ})=\tilde{\varphi}^*_2(\rr{K}_{\rm univ})$.
 \end{theorem}
\proof
To establish equation \eqref{eq:natur_K_uni} is quite easy, indeed
$$
\kappa(\bb{E})\;=\;\kappa(\varphi^*\bb{T}_{2m}^\infty)\;=\;\kappa(\tilde{\varphi}^*\circ\pi^*\bb{T}_{2m}^\infty)\;=\;\tilde{\varphi}^*\kappa(\pi^*\bb{T}_{2m}^\infty)\;=\;
\tilde{\varphi}^*(\rr{K}_{\rm univ})\;.
$$
In particular, in the first equality we used the invariance of $\kappa$ under isomorphisms and in the third equality we used the naturality of $\kappa$ under pullbacks (\cf Theorem \ref{theo:FKMM_propert}). Of course, the above relation also implies $\tilde{\varphi}_1^*(\rr{K}_{\rm univ})=\kappa(\bb{E})=\tilde{\varphi}_2^*(\rr{K}_{\rm univ})$.
\qed

\begin{remark}{\upshape
The \virg{well-posedness} of definition \eqref{eq:natur_K_uni} can also be established by a direct argument. By definition $\rr{K}_{\rm univ}=\kappa(\pi^*\bb{T}_{2m}^\infty)$ is represented by a $\Z_2$-equivariant map $\omega_{\rm univ}:\n{S}({\rm det}(\bb{T}_{2m}^\infty))^\Xi\to\n{U}(1)$ induced by the  diagonal section
$s_{\rm diag}$
restricted to the fixed point set $\n{S}({\rm det}(\bb{T}_{2m}^\infty))^\Xi$ 
(\cf Remark \ref{rk:FKMM_section}). 
As argued in Remark \ref{rk:ambiguity},
the two pullback sections $\tilde{\varphi}^*_j(s_{\rm diag}):=s_{\rm diag}\circ \tilde{\varphi}_j$, $j=1,2$ are related by the multiplication by a $\Z_2$-equivariant map $u:X\to\n{U}(1)$.
The pullback classes
$\tilde{\varphi}^*_j(\rr{K}_{\rm univ})$ are, by construction, represented by the maps $\tilde{\varphi}^*_j(\omega_{\rm univ}):=\omega_{\rm univ}\circ\tilde{\varphi}_j|_{X^\tau}$ induced by the restricted sections $\tilde{\varphi}^*_j(s_{\rm diag})|_{X^\tau}$ and the difference between these two maps is exactly the multiplication by $u|_{X^\tau}$. However, this ambiguity is removed by the quotient with respect to 
the action of $[X,\n{U}(1)]_{\Z_2}$ which appears in  the definition of the FKMM-invariant.
}\hfill $\blacktriangleleft$
\end{remark}

\subsection{Characterization of the universal FKMM-invariant}

In this section we provide 
 a useful  characterization of the universal invariant $\rr{K}_{\rm univ}$. First of all we need an analysis of the fixed point set of the space $\n{S}({\rm det}(\bb{T}_{2m}^\infty))$ under the involution given by the \virg{Real} structure ${\rm det}(\Xi)$.

\begin{lemma}
\label{lemma:detT_topo}
The following topological isomorphism
$$
\n{S}({\rm det}(\bb{T}_{2m}^\infty))^\Xi\;\simeq\;{G}_{2m}(\C^\infty)^\rho\;\sqcup\;{G}_{2m}(\C^\infty)^\rho
$$
holds true.
\end{lemma}
\proof
By Lemma \ref{lemma:R_Q_det_bun2} the restriction ${\rm det}(\bb{T}_{2m}^\infty)|_{{G}_{2m}(\C^\infty)^\rho
}\to {G}_{2m}(\C^\infty)^\rho
$ is $\rr{R}$-trivial and admits a canonical (metric-preserving) trivialization
$$
{\rm det}_{{G}^\rho}:{\rm det}(\bb{T}_{2m}^\infty)|_{{G}_{2m}(\C^\infty)^\rho}\to{G}_{2m}(\C^\infty)^\rho\times\C 
$$
with related canonical section
\beql{eq:can_sec_s_G}
s_{G^\rho}\;:\;{G}_{2m}(\C^\infty)^\rho\;\longrightarrow\; \n{S}({\rm det}(\bb{T}_{2m}^\infty))|_{{G}_{2m}(\C^\infty)^\rho
}\;\simeq\;{G}_{2m}(\C^\infty)^\rho\times\n{U}(1)
\eeq
defined by $s_{G^\rho}(x)={\rm det}_{{G}^\rho}^{-1}(x,1)
$. Because the bundle projection $\pi:{\rm det}(\bb{T}_{2m}^\infty)\to {G}_{2m}(\C^\infty)$ is equivariant it follows that  $\n{S}({\rm det}(\bb{T}_{2m}^\infty))^\Xi\subset\pi^{-1}({G}_{2m}(\C^\infty)^\rho)\cap\n{S}({\rm det}(\bb{T}_{2m}^\infty))=\n{S}({\rm det}(\bb{T}_{2m}^\infty))|_{{G}_{2m}(\C^\infty)^\rho}$. On the other side  the inclusion
 ${\rm det}(\bb{T}_{2m}^\infty)|_{{G}_{2m}(\C^\infty)^\rho}\subset{\rm det}(\bb{T}_{2m}^\infty)$ restricted at level of sphere-bundle 
 and the isomorphism given by ${\rm det}_{{G}^\rho}$ provide 
\beql{eq:can_sec_s_G2}
\n{S}({\rm det}(\bb{T}_{2m}^\infty))^\Xi\;=\;\big(\n{S}({\rm det}(\bb{T}_{2m}^\infty))|_{{G}_{2m}(\C^\infty)^\rho})^\Xi\;\simeq\; {G}_{2m}(\C^\infty)^\rho\times\{\pm1\}\;.
\eeq
\qed

\medskip

\noindent
Equation \eqref{eq:univ01.01} follows from Lemma \ref{lemma:detT_topo} and the isomorphism ${G}_{2m}(\C^\infty)^\rho\simeq {G}_{m}(\n{H}^\infty)$ discussed in Section \ref{sect:Q-vb_homotopy}.

\begin{theorem}
The following isomorphism
\beql{eq:univ000}
[{\rm det}(\bb{T}_{2m}^\infty)^\Xi,\n{U}(1)]_{\Z_2}/[{\rm det}(\bb{T}_{2m}^\infty),\n{U}(1)]_{\Z_2}\;\simeq\;\Z_2
\eeq
holds true and  under this identification $\rr{K}_{\rm univ}=-1$ corresponds to the non trivial element of $\Z_2$. 
\end{theorem}

\proof
First of all we notice that it is enough to show the isomorphism 
\eqref{eq:univ000} in order to conclude also $\rr{K}_{\rm univ}=-1$.
Indeed, this is a consequence of the naturality property proved in Theorem \ref{theo:nat_univ} and the existence of non-trivial $\rr{Q}$-bundles (\eg ${\rm Vec}^{2m}_{\rr{Q}}(\tilde{\n{S}}^2)\simeq\Z_2$).

As a consequence of the first isomorphism in \eqref{eq:iso:eq_cohom}, the proof of   \eqref{eq:univ000} is equivalent to show
$$
{\rm Coker}^1\big({\rm det}(\bb{T}_{2m}^\infty)|{\rm det}(\bb{T}_{2m}^\infty)^\Xi,\n{Z}(1)\big)\;\simeq\;\Z_2\;.
$$
Due to Lemma \ref{lemma:detT_topo}  we know that ${\rm det}(\bb{T}_{2m}^\infty)$ has two connected components and so
we only need  to prove $H^1({\rm det}(\bb{T}_{2m}^\infty),\Z)=0$ and to apply
 Proposition 
\ref{prop:cond_coker} (c.2). To this end, let us
consider the 
\emph{Gysin sequence} 
$$
H^k({G}_{2m}(\n{C}^\infty),\Z)\;\stackrel{\pi^*}{\longrightarrow}\;H^k({\rm det}(\bb{T}_{2m}^\infty),\Z)
\;\stackrel{\imath_k}{\longrightarrow}\;H^{k-1}({G}_{2m}(\n{C}^\infty),\Z)\;\stackrel{\cup\rr{c}_1}{\longrightarrow}\;H^{k+1}({G}_{2m}(\n{C}^\infty),\Z)\;\stackrel{}{\longrightarrow}\;\ldots
$$
for the complex line bundle $\pi:{\rm det}(\bb{T}_{2m}^\infty)\to {G}_{2m}(\n{C}^\infty)$. The map on the third arrow is justified by the equality $c_1({\rm det}(\bb{T}_{2m}^\infty))=c_1(\bb{T}_{2m}^\infty)$ and by the fact that $c_j(\bb{T}_{2m}^\infty):=\rr{c}_j\in H^{2j}({G}_{2m}(\n{C}^\infty),\Z)$ are, by definition, the generators
of the cohomology ring $H^\bullet({G}_{2m}(\n{C}^\infty),\Z)\simeq\Z[\rr{c}_1,\ldots,\rr{c}_{2m}]$.
When  $k=1$,  the multiplication by $\cup\rr{c}_1$ is an isomorphism, hence $\imath_1=0$.
Since $H^{1}({G}_{2m}(\n{C}^\infty),\Z)=0$, this implies $H^1({\rm det}(\bb{T}_{2m}^\infty),\Z)=0$. 
\qed

\medskip

Although not necessary, we find instructive to compute $\rr{K}_{\rm univ}=-1$ using the definition of the FKMM-invariant.
Let us consider first the restriction of the  diagonal section $s_{\rm diag}$  \eqref{eq:can_diag_sec} on the fixed point set $\n{S}({\rm det}(\bb{T}_{2m}^\infty))^\Xi\simeq {G}_{2m}(\C^\infty)^\rho\times\{\pm1\}$. The last isomorphism says that we can identify each point in ${\n{S}({\rm det}(\bb{T}_{2m}^\infty))}^\Xi$ with a pair $(x,\epsilon)\in {G}_{2m}(\C^\infty)^\rho\times\{\pm1\}$
and in view of the fact that $s_{\rm diag}$ takes values in the circle bundle $\n{S}(\pi^*{\rm det}(\bb{T}_{2m}^\infty))$ (\cf Remark \ref{rk:can_sec_circ_bun_univ}) we can write
\beql{eq:com_can_sec1}
{\n{S}({\rm det}(\bb{T}_{2m}^\infty))}^\Xi\;\ni\;(x,\epsilon) \;\stackrel{s_{\rm diag}}{\longrightarrow}\;\big((x,\epsilon),(x,\epsilon)\big)\;\in\; \n{S}(\pi^*{\rm det}(\bb{T}_{2m}^\infty))\;.
\eeq
On the other side the canonical equivariant section $s_{G^\rho}$ \eqref{eq:can_sec_s_G} can be represented as a map
$$
{G}_{2m}(\C^\infty)^\rho\;\ni\;x \;\stackrel{s_{G^\rho}}{\longrightarrow}\;(x,1)\;\in\; {\n{S}({\rm det}(\bb{T}_{2m}^\infty))}^\Xi
$$
in view of  the equality \eqref{eq:can_sec_s_G2}. The pullback section 
$$
\pi^*s_{G^\rho}\;:\; \n{S}({\rm det}(\bb{T}_{2m}^\infty)|_{{G}_{2m}(\C^\infty)^\rho}\;\longrightarrow\;\pi^*{\rm det}(\bb{T}_{2m}^\infty)
$$
defined by the  diagram \eqref{eq:diag2} is still isometric, and when restricted to ${\n{S}({\rm det}(\bb{T}_{2m}^\infty))}^\Xi$ (\cf equation \ref{eq:can_sec_s_G2}) it leads to
\beql{eq:com_can_sec2}
{\n{S}({\rm det}(\bb{T}_{2m}^\infty))}^\Xi\;\ni\;(x,\epsilon) \;\stackrel{\pi^*s_{G^\rho}}{\longrightarrow}\;\big((x,\epsilon),(x,1)\big)\;\in\; \n{S}(\pi^*{\rm det}(\bb{T}_{2m}^\infty))\;.
\eeq
A comparison between the diagonal section \eqref{eq:com_can_sec1} and the canonical section \eqref{eq:com_can_sec2} shows that the difference between the two is a  map $\omega_{\rm univ}:{\n{S}({\rm det}(\bb{T}_{2m}^\infty))}^\Xi\to\n{U}(1)$ such that $\omega_{\rm univ}(x,\epsilon)=\epsilon$. This map, which represents the universal FKMM-invariant, provides a representative for the non-trivial element in \eqref{eq:univ000}.


\appendix

\section{Condition for a $\Z_2$-value FKMM-invariant}
\label{sec:app_coker}
By construction, the FKMM-invariant $\kappa(\bb{E})$ associated with a $\rr{Q}$-bundle $(\bb{E},\Theta)$ over $(X,\tau)$ takes values in the \emph{cokernel}
$$
{\rm Coker}^1\big(X|X^\tau,\n{Z}(1)\big)\;:=\;H^1_{\Z_2}\big(X^\tau,\n{Z}(1)\big)\;/\;r\big(H^1_{\Z_2}(X,\n{Z}(1))\big)\;.
$$
This fact follows from a comparison between Definition \ref{def:FKMM-invariant} and Lemma \ref{lemma:coker}. It is interesting to have conditions on   $(X,\tau)$ which assure that 
${\rm Coker}^j\big(X|X^\tau,\n{Z}(1)\big)$ reduces to the simplest (non-trivial) abelian group $\Z_2$.

\begin{proposition}
\label{prop:cond_coker}
Let $(X,\tau)$ be an involutive space and assume that:
\begin{itemize}
\item[(a)] $X$ verifies 
 condition 0  of Definition \ref{ass:2};
\item[(b)] $X^\tau\neq\emptyset$ and consists of a finite number $N$ of path-connected components;
\end{itemize}
Under these assumptions a necessary condition for ${\rm Coker}^1\big(X|X^\tau,\n{Z}(1)\big)\simeq\Z_2$ is:
\begin{itemize}
\item[(c.1)] $H^1(X,\Z)\simeq\Z^b$ and $N\leqslant b+2$.
\end{itemize}
and a sufficient condition is:
\begin{itemize}
\item[(c.2)] $b=0$ and $N=2$.
\end{itemize}
\end{proposition}
\proof
First of all, let us recall the exact sequence (\cf Proposition 2.3 in \cite{gomi-13})
$$
H^0_{\Z_2}(X,\Z(1))\;\longrightarrow\;H^0(X,\Z)\;\longrightarrow\;H^0_{\Z_2}(X,\Z)\;\longrightarrow\;H^1_{\Z_2}(X,\Z(1))\;\longrightarrow\;H^1(X,\Z)\;\longrightarrow\;\ldots
$$
Since $X$ is connected and has at least one fixed point, the exact sequence above reduces to
$$
0\;\longrightarrow\;\tilde{H}^1_{\Z_2}(X,\Z(1))\;\longrightarrow\;\tilde{H}^1(X,\Z)\;=\;{H}^1(X,\Z)\;\simeq\;\Z^b
$$
where $\tilde{H}^j$ is the standard notation for the reduced cohomology groups. Therefore, $\tilde{H}^1_{\Z_2}(X,\Z(1))$ is a subgroup of $\Z^b$ and so $\tilde{H}^1_{\Z_2}(X,\Z(1))\simeq \Z^{b_-}$ for some $b_-\leqslant b$. 
Moreover, ${H}^1_{\Z_2}(X^\tau,\Z(1))\simeq \Z_2^N$ with a $\Z_2$ summand for each connected component. This is evident from the isomorphism ${H}^1_{\Z_2}(X^\tau,\Z(1))\simeq{\rm Map}(X^\tau,\Z_2)$ which is a consequence of the first of \eqref{eq:iso:eq_cohom}.
Let us now estimate the size of  ${H}^1_{\Z_2}(X,\Z(1))\simeq \Z_2\oplus \Z^{b_-}$ in ${H}^1_{\Z_2}(X^\tau,\Z(1))\simeq \Z_2^N$ under the restriction map $r$. Evidently,  $r\big(H^1_{\Z_2}(X,\n{Z}(1))\big)\simeq\Z_2^a$ for some integer $0\leqslant a\leqslant N$. On the other hand, the direct summand $H^1_{\Z_2}(\{\ast\},\n{Z}(1))\simeq\Z_2$ in  $H^1_{\Z_2}(X,\n{Z}(1))\simeq[X,\n{U}(1)]_{\Z_2}$
consists of constant maps $X\to\pm 1$. Hence the image of this $\Z_2$ summand under $r$ is a non-trivial subgroup $\Z_2$ in
${H}^1_{\Z_2}(X^\tau,\Z(1))$. The remaining summand $\Z^{b_-}$ in 
${H}^1_{\Z_2}(X,\Z(1))$ has an image under $r$ which can be at most isomorphic to $\Z_2^{b_-}$. Therefore, one has the inequality $a\leqslant 1+b_-$. The condition ${\rm Coker}^1\big(X|X^\tau,\n{Z}(1)\big)\simeq\Z_2$ is equivalent to $a=N-1$ and this equality can be verified only if $N\leqslant b_-+2\leqslant b+2$. In the case $N=2$ and $b=0$ one has ${H}^1_{\Z_2}(X,\Z(1))\simeq\Z_2$ and ${H}^1_{\Z_2}(X^\tau,\Z(1))\simeq \Z_2^2$. Both these groups are represented by 
constant maps with values $\pm 1$ and the action of ${H}^1_{\Z_2}(X,\Z(1))$ on ${H}^1_{\Z_2}(X^\tau,\Z(1))$ is diagonal. This implies that ${\rm Coker}^1\big(X|X^\tau,\n{Z}(1)\big)\simeq\Z_2$.
\qed

\medskip

\begin{table}[htp]
 \label{tab:coker}
 \centering
 \begin{tabular}{|l|c|c|c|}
\hline
\ \ \ \ \ \   $(X,\tau)$ & $N$ & $b$ & $N\leqslant b+2$\\
\hline
\hline
\rule[-2mm]{0mm}{6mm}
$\tilde{\n{S}}^1$&   $2$ & 1 & no \\

 \hline
 \rule[-2mm]{0mm}{6mm}
$\tilde{\n{S}}^d\quad (d\geqslant 2)$&   $2$ & $0$ & yes \\
\hline
\rule[-2mm]{0mm}{6mm}
$\tilde{\n{T}}^2$&   $4$ & $2$ & yes\\
\hline
\rule[-2mm]{0mm}{6mm}
$\tilde{\n{T}}^d\quad (d\geqslant 3)$&   $2^d$ & $d$ & no\\

\hline
\end{tabular}\vspace{1mm}
 \caption{
 {\footnotesize
 According to Preposition \ref{prop:cond_coker} all the TR-spheres $\tilde{\n{S}}^d$  with $d\geqslant 2$ fulfill the strong condition (c.2) which implies that the cokernel is exactly $\Z_2$.
In the case of  TR-tori only $\tilde{\n{T}}^2$ fulfills the necessary condition (c.1) but not the sufficient condition (c.2).}
 }
 \end{table}

\noindent
A formula for the cokernel of   $\tilde{\n{T}}^d$ can be explicitly derived.

\begin{proposition}\label{prop:coker_tori}
Let $\tilde{\n{T}}^d=({\n{T}}^d,\tau)$ the involutive torus described in Definition \ref{def:AII_sys}, then
$$
{\rm Coker}^1\big(\tilde{\n{T}}^d|(\tilde{\n{T}}^d)^\tau,\n{Z}(1)\big)\;\simeq\;\Z_2^{2^d-(d+1)} \qquad\quad\ \ \forall\ d\geqslant 1\;.
$$
\end{proposition}

\proof[
sketch of]
One starts with ${H}^1_{\Z_2}\big(\tilde{\n{T}}^d,\Z(1)\big)\simeq \Z_2\oplus\Z^d$ and ${H}^1_{\Z_2}\big((\tilde{\n{T}}^d)^\tau,\Z(1)\big)\simeq \Z_2^{2^d}$ (\eg one can use the recursive relations \cite[eq. (5.9)]{denittis-gomi-14}). Let $0\leqslant a\leqslant 2^d$ be an integer such that
$r\big({H}^1_{\Z_2}\big(\tilde{\n{T}}^d,\Z(1)\big)\big)\simeq\Z_2^a$ where $r$ is the restriction map in cohomology induced by the  inclusion $\imath:(\tilde{\n{T}}^d)^\tau\hookrightarrow \tilde{\n{T}}^d$. To conclude the proof is enough to show that $a=d+1$. Since $\Z_2$ is a field, we can think of $\Z_2\oplus\ldots\oplus\Z_2$ as a vector space over $\Z_2$. We recall that the $\Z_2$-summand in ${H}^1_{\Z_2}\big(\tilde{\n{T}}^d,\Z(1)\big)$ is generated by the constant map $\epsilon:\tilde{\n{T}}^d\to - 1$ and the $\Z^d$-summand by the $d$ canonical projections 
$
\pi_j:\tilde{\n{T}}^d=\tilde{\n{T}}^1\times\ldots\times\tilde{\n{T}}^1\to \tilde{\n{T}}^1\simeq\n{U}(1)
$.
The  set $(\tilde{\n{T}}^1)^\tau$ contains two fixed points which can be parametrized with $\pm1$. This leads to  a bijection between $(\tilde{\n{T}}^d)^\tau$ and $\Z_2^{d}$. Let us fix an order for the $2^{d}$ points in $\Z_2^{d}$, \ie $\Z_2^{d}=\{{\rm v}_1,\ldots,{\rm v}_{2^d}\}$
where ${\rm v}_k:=(v_k^1,\ldots,v_k^d)$ with $v_k^j\in\{\pm1\}$, ($k=1,\ldots,2^d$, $j=1,\ldots,d$).
The map $r$, which coincides with the evaluations on the fixed points,   sends $\epsilon$ to a vector $r(\epsilon)\in \Z_2^{2^d}$ represented as
$r(\epsilon)=(-1,\ldots,-1)$. Similarly the $d$ vectors $r(\pi_j)\in \Z_2^{2^d}$ are given by $r(\pi_j):=(\pi_j({\rm v}_1),\ldots,\pi_j({\rm v}_{2^d}))$
with $\pi_j({\rm v}_k):=v^j_k$. The linear independence of $\{r(\epsilon),r(\pi_j),\ldots,r(\pi_d)\}$ can be checked by the help of the Gauss elimination and this shows that $a=d+1$.
\qed

\medskip

\noindent{\bf The proof of Proposition \ref{propos:d2manif_FKMM}.}
Let $(X,\tau)$  be an involutive space which verifies Assumption \ref{ass:2.3} and  $X^\tau=\{x_1,\ldots,x_N\}$ the fixed point set for some integer  $N>0$. Let us fix some notation: For each $x_j\in X^\tau$ 
let $D_j\subset X$ be a  \emph{disk} (closed, contractible set with boundary $\partial D_j\simeq\n{S}^1$)
 which contains $x_j$. We can choose
 sufficiently small  disks $D_j$ such  that $D_i\cap D_j=\emptyset$ if $i\neq j$ and $\tau(D_j)=D_j$
(this is a consequence of the \emph{slice Theorem} \cite[Chapter I, Section 3]{hsiang-75} as in the proof of Lemma \ref{lemma:2d_manifA}).
  Let us set  $\f{D}:=\bigcup_{i=1}^{N}D_i$ and $X':=X\setminus{\rm Int}(\f{D})$. By construction $X'$ is a manifold with boundary $\partial\f{D}\simeq\bigcup_{i=1}^{N}\n{S}^1_i$ on which the involution $\tau$ acts freely. The orbit space $X/\tau$ has boundary $\partial(X/\tau)=(\partial X)/\tau$.
  
\medskip

As a first step, let us prove that $N=2n$ for some integer $n>0$. From the assumptions it follows that a choice of a Riemannian metric  on $X$  makes it into a Riemann surface (a complex manifold of dimension 1) without boundary. 
By means of the \emph{average} construction we can assume without loss of generality that the metric is $\tau$-invariant. Because $\tau:X\to X$ preserves the orientation, we can think of it as a holomorphic map. Then, the quotient $X/\tau$ also gives rise to a Riemann surface, and the projection $\pi:X\to X/\tau$ is holomorphic. In particular $\pi$ is a ramified double covering with $N$ branching points $x_1,\ldots,x_N$.
The \emph{Riemann-Hurwicz formula} \cite[Chapter 2]{griffiths-harris-78} tells us
$$
4\; g(X/\tau)\;=\; 2\; g(X)\;+\;2\;-\;N\;, 
$$
where $g$ denotes the genus of the related Riemann surface, hence $N=2n$ has to be even.
The Riemann surface $X'$ has genus $g=g(X)$ and $2n$ boundary components, hence
\beql{eq:cohom_app1}
H^k(X',\Z)\;\simeq\;\left\{
\begin{aligned}
&\Z&\quad&\text{if}\ \ k=0&\\
&\Z^{2(g+n)-1}&\quad&\text{if}\ \ k=1&\\
&0&\quad&\text{if}\ \ k\geqslant 2&\;.
\end{aligned}
\right.
\eeq
The Riemann surface $X'/\tau$ is obtained by removing $2n$ disks around the branching points, hence 
$g(X'/\tau)=\frac{1}{2}(g+1-n)$ and it has $2n$ boundary components. This implies
\beql{eq:cohom_app2}
H^k(X'/\tau,\Z)\;\simeq\;\left\{
\begin{aligned}
&\Z&\quad&\text{if}\ \ k=0&\\
&\Z^{g+n}&\quad&\text{if}\ \ k=1&\\
&0&\quad&\text{if}\ \ k\geqslant 2&\;.
\end{aligned}
\right.
\eeq

\medskip

Now, we can prove that $H_{\Z_2}^2(X,\Z(1))=0$. We start with  the exact sequence (\cf Proposition 2.3 in \cite{gomi-13})
$$
H^k(X'/\tau,\Z)\simeq H^k_{\Z_2}(X',\Z) \;\to\;H^k(X',\Z)\;\to\;H^k_{\Z_2}(X',\Z(1))\;\to\;H^{k+1}_{\Z_2}(X',\Z)\simeq H^{k+1}(X'/\tau,\Z)
$$
where we used the fact that the action of $\tau$  on $X'$ is free. From \eqref{eq:cohom_app2} it follows  that $H^k_{\Z_2}(X',\Z(1))\simeq H^k(X',\Z)$ for all $k\geqslant 2$ and so equation \eqref{eq:cohom_app1} implies  
$H^k_{\Z_2}(X',\Z(1))=0$ for all $k\geqslant 2$.  Next, we use the Meyer-Vietoris sequence for $\{X',\f{D}\}$, \ie
$$
H^1_{\Z_2}(X',\Z(1))\oplus H^1_{\Z_2}(\f{D},\Z(1))\;\stackrel{\Delta}{\to}\;H^1_{\Z_2}(X'\cap \f{D},\Z(1))\to H^2_{\Z_2}(X,\Z(1))\to H^2_{\Z_2}(X',\Z(1))\oplus H^2_{\Z_2}(\f{D},\Z(1))
$$
where the \emph{difference} homomorphism $\Delta$ is given by $\Delta(a_{X'},a_{\f{D}})=-a_{X'}|_{X'\cap \f{D}}+a_{\f{D}}|_{X'\cap \f{D}}$ for $a_{X'}\in H^1_{\Z_2}(X',\Z(1))\simeq[X',\n{U}(1)]_{\Z_2}$
and $a_{\f{D}}\in H^1_{\Z_2}(\f{D},\Z(1))\simeq[\f{D},\n{U}(1)]_{\Z_2}$ (here we used the isomorphism \eqref{eq:iso:eq_cohom}). Since each connected component $D_j$ of $\f{D}$ is equivariantly contractible
  $H^1_{\Z_2}(\f{D},\Z(1))\simeq\Z_2^{2n}$ consists of locally constant function with values in $\{\pm 1\}$.
On the other side, the action of $\tau$ is free on $X'\cap \f{D}\simeq \bigsqcup_{j=1}^{2n}\n{S}^1$
and $(X'\cap \f{D})/\tau\simeq \bigsqcup_{j=1}^{2n}\n{S}^1$ since $\tau$ acts isomorphically to the antipodal map on each component. This yields 
$$
H^k_{\Z_2}(X'\cap \f{D},\Z(1))\;\simeq\; H^k(\n{S}^1,\Z(1))^{2n}\;\simeq\;
\left\{
\begin{aligned}
&\Z_2^{2n}&\quad&\text{if}\ \ k=1&\\
&0&\quad&\text{if}\ \ k\neq 1&\;.
\end{aligned}
\right.
$$
Since also $H^1_{\Z_2}(X'\cap \f{D},\Z(1))$ can be represented  by locally constant functions from $X'\cap \f{D}$ to $\{\pm1\}$ it follows that the map $H^1_{\Z_2}(\f{D},\Z(1))\to H^1_{\Z_2}(X'\cap \f{D},\Z(1))$ is surjective.
Moreover, $H^2_{\Z_2}(X',\Z(1))\simeq 0$ and $H^2_{\Z_2}(\f{D},\Z(1))\simeq \bigoplus_{j=1}^{2n} H^2_{\Z_2}(\{\ast\},\Z(1))\simeq 0$ imply $H_{\Z_2}^{2}(X,\Z(1))=0$.

\medskip

As the last step we prove that $H^2_{\Z_2}(X|X^\tau,\n{Z}(1)) \simeq \Z_2$. By the homotopy axiom and the excision axiom for the Borel cohomology theory, we get isomorphisms
$$
H^2_{\Z_2}(X|X^\tau,\n{Z}(1))\; \simeq \;H^2_{\Z_2}(X|\f{D},\n{Z}(1)) \; \simeq \;H^2_{\Z_2}(X'|\partial X',\n{Z}(1))\;.
$$
Let us consider the exact sequence  \cite[Proposition 2.3]{gomi-13}
$$
\Z\;\simeq\;H^2_{\Z_2}(X'|\partial X',\n{Z})\;\stackrel{f}{\to}\;H^2(X'|\partial X',\n{Z})\;\stackrel{}{\to}\;H^2_{\Z_2}(X'|\partial X',\n{Z}(1))\;\stackrel{}{\to}\;H^3_{\Z_2}(X'|\partial X',\n{Z})\;=\;0
$$
associated with a map $f$ which \virg{forgets} the $\Z_2$-action. 
We recall that $\tau$ acts freely on $X'$
and $f$ can be identified with the pullback by the projection $\pi:X' \to X'/\tau$. Then, after an application of the Poincar\'{e} duality and the universal coefficient theorem to \eqref{eq:cohom_app2} we get
$$
H^k_{\Z_2}(X'|\partial X',\n{Z})\;\simeq\; H^k(X'/\tau|\partial (X'/\tau),\n{Z})\;\simeq\;
\left\{
\begin{aligned}
&\Z^{g+n}&\quad&\text{if}\ \ k=1&\\
&\Z&\quad&\text{if}\ \ k=2&\\
&0&\quad&\text{if}\ \ k\neq 1,2&\;.
\end{aligned}
\right.
$$
In the same way, one deduces from \eqref{eq:cohom_app1} that $H^2(X'|\partial X',\n{Z})\simeq\Z$. Because the projection $\pi:X'\to X'/\tau$ is a honest double covering, $f$ is a two-to-one map. This allows us to conclude that 
$H^2_{\Z_2}(X'|\partial X',\n{Z}(1))\simeq\Z_2$.

\section{Spatial parity and quaternionic vector bundles}
\label{sec:inv_symm}

Let us introduce the  \emph{spatial parity} operator 
$\hat{I}$ defined by $(\hat{I}\psi)(x)=\psi(-x)$ for vectors $\psi$ in
 $L^2(\R^d,\dd x)\otimes\C^L$ (continuous case)  or in  $\ell^2(\Z^d)\otimes\C^L$ (periodic case). A Hamiltonian  $\hat{H}$ is \emph{parity-invariant} if $\hat{I}\hat{H}\hat{I}=\hat{H}$.
 If $\hat{H}$ posses also a -TR symmetry $\hat{\Theta}$ then the combination $\hat{\Theta}':=\hat{I}\hat{\Theta}$ provides an anti-linear fiber-preserving map of the associated Bloch-bundle. If it happens that the two symmetries are independently broken but the combination $\hat{\Theta}'$
 survives as a fundamentally symmetry, then Proposition \ref{prop:Qv1.0} assures that the topological phase of $\hat{H}$ is classified by a quaternionic vector bundle.

\medskip

The classification of quaternionic vector bundles in dimension $d=1,\ldots,4$ is almost trivial; In fact \cite[Chapter 9, Theorem 1.2]{husemoller-94} assures that for each rank of the fiber the stable rank condition is verified and just by looking at  the ${K\text{\em Sp}}$-theory (see Table \ref{tab:KR3}.2 \& Table \ref{tab:KR2}.3) one obtains
$$
{\rm Vec}^m_{\n{H}}(\n{T}^d)\;\simeq\;{\rm Vec}^m_{\n{H}}(\n{S}^d)\;\simeq\;
\left\{
\begin{aligned}
&0&&\text{if}\ \ d=1,2,3\\
&\Z&&\text{if}\ \ d=4
\end{aligned}
\right.
\qquad\quad\forall\ m\in\N\;.
$$

\section{An overview to $KQ$-theory}
\label{sect:KQ-theory}

According to  \cite{dupont-69,seymour-73} we denote with  $KQ(X,\tau)$  the \emph{Grothendieck group} 
of $\rr{Q}$-vector bundles over the involutive space $(X,\tau)$. We refer also to \cite[Section 3.6]{luke-mishchenko-98} for pedagogical description of the  $KQ$-theory and its relation with the Atiyah's $KR$-theory. By a restriction to the fixed point set $X^\tau$ (hereafter assumed non empty) one has a homomorphism
$KQ(X,\tau)\to KQ(X^\tau,{\rm Id_X})\simeq K\text{\em Sp}(X^\tau)$, where $K\text{\em Sp}$  denotes the $K$-theory for  vector bundles with quaternionic fiber (see \eg \cite{husemoller-94}).

\medskip

 The \emph{reduced group} $\widetilde{KQ}(X,\tau)$ is the kernel of the homomorphism $KQ(X,\tau)\to KQ(\{\ast\})$ where $\ast\in X$ is a $\tau$-invariant base point. 
When $X$ is compact one has the usual relation
\beql{eq:KQ0}
KQ(X,\tau)\;\simeq\;\widetilde{KQ}(X,\tau)\;\oplus\; KQ(\{\ast\})\;\simeq\;\widetilde{KQ}(X,\tau)\;\oplus\; \Z
\eeq
where we used $KQ(\{\ast\})\simeq K\text{\em Sp}(\{\ast\})\simeq\Z$. Moreover,
 the isomorphism
\beql{eq:KR0X}
\tilde{KR}(X,\tau)\;\simeq\;{\rm Vec}_\rr{R}(X,\tau)\;:=\;\bigcup_{m\in\N}{\rm Vec}^m_\rr{R}(X,\tau)
\eeq
 establishes 
 the fact that $\tilde{KQ}(X,\tau)$ provides the description for $\rr{Q}$-bundles in the stable regime (\ie when the rank of the fiber is assumed to be sufficiently large).

\medskip

As for  the $KR$-theory
also the $KQ$-theory
 can be endowed with a grading structure as follows: first of all one introduces the groups
$$
\begin{aligned}
KQ^{j}(X,\tau)\;&:=\;KQ(X\times\n{D}^{0,j} ; X\times \n{S}^{0,j} ,\tau\times\vartheta)\\
KQ^{-j}(X,\tau)\;&:=\;KQ(X\times\n{D}^{j,0} ; X\times \n{S}^{j,0} ,\tau)
\end{aligned}\;\qquad\qquad j=0,1,2,3,\ldots
$$
where $\n{D}^{p,q}$ and $\n{S}^{p,q}$ are the unit ball and unit sphere in the  space $\R^{p,q}:=\R^p\oplus{\rm i}\R^q$ made involutive by  the complex conjugation (\cf \cite[Exampe 4.2]{denittis-gomi-14}). The relative group $KQ(X;Y,\tau)$ of an involutive space $(X,\tau)$ with respect to a $\tau$-invariant subset $Y\subset X$ is defined as $\widetilde{KQ}(X/Y,\tau)$ and corresponds to the Grothendieck group of $\rr{Q}$-bundles over $X$ which vanish on $Y$. The negative groups $KQ^{-j}$ agree with the usual suspension groups since the spaces 
$\n{D}^{j,0}$ and $\n{S}^{j,0}$ are invariant. The positive groups $KQ^{j}$ are \virg{twisted} suspension groups since
$\n{D}^{0,j}$ and $\n{S}^{0,j}$ are endowed with the  $\Z_2$-action of the antipodal map $\vartheta$. With respect to this grading the $KQ$ groups are 8-periodic, \ie
$$
KQ^{j}(X,\tau)\;\simeq\; KQ^{j+8}(X,\tau)\;,\qquad\qquad j\in\Z\;.
$$
If $X$ has fixed points one can extend the isomorphism \eqref{eq:KQ0} for negative groups:
\beql{eq:KQ1}
KQ^{-j}(X,\tau)\;\simeq\;\widetilde{KQ}^{-j}(X,\tau)\;\oplus\; KQ^{-j}(\{\ast\})\;,\qquad\qquad  j=0,1,2,3,\ldots
\eeq
where $KQ^{-j}(\{\ast\})\simeq K\text{\em Sp}^{-j}(\{\ast\})$ for all $j$. 
\begin{table}[htp]
 \label{tab:KR1}
 \centering
 \begin{tabular}{|c||c|c|c|c|c|c|c|c|}
\hline
   & $j=0$ & $j=1$&$j=2$&$j=3$&$j=4$&$j=5$&$j=6$&$j=7$\\
\hline
 \hline
 \rule[-2mm]{0mm}{6mm}
$KQ^{-j}(\{\ast\})$&   $\Z$ & $0$ & $0$ & $0$ &$\Z$&$\Z_2$&$\Z_2$&$0$\\
\hline
\end{tabular}\vspace{1mm}
 \caption{
 {\footnotesize 
 The table is calculated using  $KQ^{-j}(\{\ast\})\simeq K\text{\em Sp}^{-j}(\{\ast\})\simeq \widetilde{KO}(\n{S}^{j+4})$ \cite[Theorem 5.19, Chapter III]{karoubi-97}.  The Bott periodicity implies $KQ^{-j}(\{\ast\})\simeq KQ^{-j-8}(\{\ast\})$.
 }
 }
 \end{table}

\noindent
Finally, the following connection with $KR$-theory
\beql{eq:KQ1.2}
{KQ}^{j}(X,\tau)\;\simeq\;KR^{j\pm 4}(X,\tau)\;,\quad\quad\; \widetilde{KQ}^{j}(X,\tau)\;\simeq\;\widetilde{KR}^{j\pm 4}(X,\tau)\qquad j\in\Z
\eeq
has been proved in \cite{dupont-69}.

\medskip

We can use equation \eqref{eq:KQ1.2}  to compute the
$KQ$-theory for the involutive sphere $\tilde{\n{S}}^d\equiv({\n{S}}^d,\tau)$ by means of the $KR$-theory. Indeed, one has that
\beql{eq:KQ1.3}
\widetilde{KQ}(\tilde{\n{S}}^d)\;\simeq\;\widetilde{KR}^{4}(\tilde{\n{S}}^d)\;\simeq\;\widetilde{KR}(\tilde{\n{S}}^{d+4})\qquad\quad {\rm mod.}\; 8
\eeq
where in the last isomorphism we used the
{equivariant} reduced suspension formula \cite[eq. (A.6)]{denittis-gomi-14}.

\begin{table}[htp]
 \label{tab:KR3}
 \centering
 \begin{tabular}{|c||c|c|c|c|c|c|c|c|}
\hline
   & $d=1$ & $d=2$&$d=3$&$d=4$&$d=5$&$d=6$&$d=7$&$d=8$\\
\hline
\hline
\rule[-2mm]{0mm}{6mm}
$\widetilde{K}(\n{S}^d)$&   $0$ & $\Z$ & $0$ & $\Z$ &$0$&$\Z$&$0$&$\Z$\\
 \hline
 \rule[-2mm]{0mm}{6mm}
$\widetilde{KQ}(\tilde{\n{S}}^d)$&   $0$ & $\Z_2$ & $\Z_2$ & $\Z$ &$0$&$0$&$0$&$\Z$\\
\hline
\rule[-2mm]{0mm}{6mm}
$\widetilde{K\text{\em Sp}}(\n{S}^d)$&   $0$ & $0$ & $0$ & $\Z$ &$\Z_2$&$\Z_2$&$0$&$\Z$\\
\hline
\end{tabular}\vspace{1mm}
 \caption{
 {\footnotesize
 The reduced $KQ$ groups are computed with the help of the formula \eqref{eq:KQ1.3} and the Table B.3 in \cite{denittis-gomi-14}.
 The reduced $K$ groups for complex and quaternionic vector bundles are
 computed in \cite[Chapter 9, Corollary 5.2]{husemoller-94}.
 We recall the relation $\widetilde{KO}(\n{S}^d)\simeq\widetilde{K\text{\em Sp}}(\n{S}^{d+4})$ (mod. 8) which relates the $K$-theories of real and quaternionic vector bundles.}
 }
 \end{table}

\noindent
In order to compute the $KQ$ groups for TR-tori we can adapt the formula 
\cite[eq. (B.7)]{denittis-gomi-14} with the help of \eqref{eq:KQ1.2} in order to obtain
\beql{eq:KQ4bis}
{KQ}^{-j}(\tilde{\n{T}}^{d})\;\simeq\;
\bigoplus_{n=0}^d\Big({KQ}^{-(j-n)}(\{\ast\})\Big)^{\oplus \;\binom {d} {n}}\;.
\eeq

\begin{table}[htp]
 \label{tab:KR2}
 \centering
 \begin{tabular}{|c||c|c|c|c|c|c|c|c|}
\hline
   & $d=1$ & $d=2$&$d=3$&$d=4$&$d=5$&$d=6$&$d=7$&$d=8$\\
\hline
\hline
\rule[-2mm]{0mm}{6mm}
$\widetilde{K}(\n{T}^d)$&   $0$ & $\Z$ & $\Z^3$ & $\Z^7$ &$\Z^{15}$&$\Z^{31}$&$\Z^{63}$&$\Z^{127}$\\

 \hline
 \rule[-2mm]{0mm}{6mm}
$\widetilde{KQ}(\tilde{\n{T}}^{d})$&   $0$ & $\Z_2$ & $\Z_2^4$ & $\Z\oplus\Z_2^{10}$ &$\Z^5\oplus\Z_2^{20}$&$\Z^{15}\oplus\Z_2^{35}$&$\Z^{35}\oplus\Z_2^{56}$&$\Z^{71}\oplus \Z_2^{84}$\\
\hline
\rule[-2mm]{0mm}{6mm}
$\widetilde{K\text{\em Sp}}(\n{T}^d)$&   $0$ & $0$ & $0$ & $\Z$ &$\Z^5\oplus\Z_2$&$\Z^{15}\oplus\Z_2^{7}$&$\Z^{35}\oplus\Z_2^{28}$&$\Z^{71}\oplus\Z_2^{84}$\\
\hline
\end{tabular}\vspace{1mm}
 \caption{
 {\footnotesize
 The groups $\widetilde{KQ}(\tilde{\n{T}}^d)$ are obtained from equation \eqref{eq:KQ4bis} and the isomorphism \eqref{eq:KQ0}.
 In the quaternionic  case a recursive formula can be derived from the isomorphism 
 ${K\text{\em Sp}}^{-j}(\n{S}^{1}\times Y)\simeq{K\text{\em Sp}}^{-(j+1)}(Y)\;\oplus\; {K\text{\em Sp}}^{-j}(Y)$ .
}
 }
 \end{table}


\medskip
\medskip

\end{document}